\documentclass{LMCS} 
\usepackage{latexsym}
\usepackage{amssymb}
\usepackage{amsmath}
\usepackage{graphicx}
\usepackage{color}
\usepackage{stmaryrd} 
\usepackage{pst-all}
\usepackage{enumerate,hyperref}
\definecolor{BlueGreen}{rgb}{0,1,0}

\usepackage[plain]{algorithm} 
\usepackage[noend]{algorithmic}
\algsetup{indent=2em}
\newcommand{\Procname}[1]{\begin{flushleft}#1\end{flushleft}}

\newcommand\Real{\mathbb{R}}

\def\cost{\mathit{cost}}

\newcommand{\action}  {{\ensuremath \alpha}}
\newcommand{\actions} {{\ensuremath {Act}}}
\def\steps{\mathit{Steps}}

\newcommand\tab{\qquad}
\newcommand\ie{{i.\,e.\@}}

\newcommand\wrt{{with respect to }}
\newcommand\upto{up to\ }
\newcommand\et{\emph{et al.\ }}

\newcommand\support{\mathit{Supp}}

\newcommand\size[1]{\left\lvert#1\right\rvert}

\newcommand\R{\mathbf{R}}
\newcommand\N{\mathcal{N}}
\newcommand\M{\mathcal{M}}
\renewcommand\O{\mathcal{O}}

\newcommand\D{\mathcal{M}}
\newcommand\C{\mathcal{M}}

\newcommand\distrs{\mathit{Dist}}

\renewcommand\P{\mathbf{P}}
\newcommand\Pmu[1]{\mathbf{P}(#1,\cdot)} 
\newcommand\simrel{\precsim}
\newcommand\weight{\sqsubseteq}
\newcommand\aux{{\mathord{\bot}}}
\newcommand\source{{\mathord{\nnearrow}}}
\newcommand\sink{{\mathord{\ssearrow}}}
\newcommand\init{\mathit{init}}
\newcommand\dist{\mathit{Dist}}
\newcommand\rate{\mathit{Rate}}
\newcommand\post{\mathit{post}}
\newcommand\pre{\mathit{pre}}
\newcommand\psize{\mathit{h}}
\newcommand\listener{\mathit{Listener}}

\newcommand\overeq[1]{\stackrel{#1}{=}}
\newcommand\wsrel{\precapprox} 

\newcommand{\prog}[1]{\textnormal{\scshape#1}}

\newcommand\mbaier{k}

\def\capacity{c}

\newcommand\fanout{\mathit{fan}}

\newcommand\pa{\mathcal{M}=(S,Act,\P,L)}
\newcommand\cpa{\mathcal{M}=(S,Act,\R,L)}
\newcommand\dtmc{\D=(S,\P,L)}
\newcommand\ctmc{\C=(S,\mathbf{R},L)}
\newcommand{\transby}[1]{\ensuremath{\xrightarrow{#1}}}

\newcommand{\combinedby}[1]{\ensuremath{\stackrel{#1}\leadsto}}

\def\doi{4 (4:6) 2008}
\lmcsheading%
{\doi}
{1--43}
{}
{}
{Aug.~31, 2007}
{Nov.~11, 2008}
{}   

\begin{document}
\title[Deciding Probabilistic Simulations]{Flow Faster: Efficient
Decision Algorithms for Probabilistic Simulations\rsuper *}

\author[L. Zhang]{Lijun Zhang\rsuper a}    
\address{{\lsuper{a,b}}Department of Computer Science, Saarland  University, Germany}  
\email{\{zhang,hermanns\}@cs.uni-sb.de}  
\thanks{{\lsuper{a,b}}This work is supported by the DFG as part of the Transregional
Collaborative Research Center SFB/TR 14 AVACS and by the European
Commission under the IST framework 7 project QUASIMODO}   

\author[H.~Hermanns]{Holger Hermanns\rsuper b}      
\address{\vskip-6 pt}  

\author[F.~Eisenbrand]{Friedrich Eisenbrand\rsuper c} 
\address{{\lsuper c}Department of Mathematics, EPFL,  Switzerland}        
\email{friedrich.eisenbrand@epfl.ch}  
\thanks{{\lsuper c}Supported by Deutsche Forschungsgemeinschaft (DFG)  within
Priority Programme 1307 "Algorithm Engineering".}    

\author[D.~N.~Jansen]{David N.\ Jansen\rsuper d}
\address{{\lsuper d}%
Software Modelling and Verification, RWTH Aachen University, Germany, and Model-Based System Development, Radboud University, Nijmegen, The Netherlands}
\email{D.Jansen@cs.ru.nl}
\thanks{{\lsuper d}
Work partly supported by the DFG Research Training Group AlgoSyn
(project nr. 1298).}



\keywords{Markov chains, strong- and weak-simulation, decision
algorithms, parametric maximum flow}
\subjclass{F.2.1, F.3.1, G.2.2, G.3}
\titlecomment{{\lsuper *}An extended abstract of the paper has appeared in~\cite{ZHEJ07,ZH07}.}

\begin{abstract}
  \noindent 
  Strong and weak simulation relations have been proposed for Markov
  chains, while strong simulation and strong probabilistic simulation
  relations have been proposed for probabilistic automata.  This paper
  investigates whether they can be used as effectively as their
  non-probabilistic counterparts.  It presents drastically improved
  algorithms to decide whether some (discrete- or continuous-time)
  Markov chain strongly or weakly simulates another, or whether a
  probabilistic automaton strongly simulates another.  The key
  innovation is the use of parametric maximum flow techniques to
  amortize computations.  We also present a novel algorithm for
  deciding strong probabilistic simulation preorders on probabilistic
  automata, which has polynomial complexity via a reduction to an LP
  problem. When extending the algorithms for probabilistic automata to
  their continuous-time counterpart, we retain the same complexity for
  both strong and strong probabilistic simulations.
\end{abstract}

\maketitle

\section{Introduction}
\noindent
Many verification methods have been introduced to prove the
correctness of systems exploiting rigorous mathematical foundations.
As one of the automatic verification techniques, model checking has
successfully been applied to automatically find errors in complex
systems. The power of model checking is limited by the  state
space explosion problem.  Notably, minimizing the system to the
bisimulation~\cite{Milner80,Park81} quotient is a favorable approach.
As a more aggressive attack to the problem, simulation
relations~\cite{Milner71} have been proposed for these models.  In
particular, they provide the principal ingredients to perform
abstractions of the models, while preserving safe CTL properties
(formulas with positive universal path-quantifiers only)~\cite{CGL94}.

Simulation relations are preorders on the state space such that
whenever state $s'$ simulates state $s$ (written $s \simrel s'$) then
$s'$ can mimic all stepwise behaviour of $s$, but  $s'$ may
perform steps that cannot be matched by $s$. One of the interesting
aspects of simulation relations is that they allow a verification by
``local'' reasoning.  Based on this, efficient algorithms for deciding
simulation preorders have been proposed
in~\cite{BloomP95,HenzingerHK95}.

Randomisation has been employed widely for performance and
dependability models, and consequently the study of verification
techniques of probabilistic systems with and without nondeterminism
has drawn a lot of attention in recent years.  A variety of
equivalence and preorder relations, including strong and weak
simulation relations, have been introduced and widely considered for
probabilistic models.  In this paper we consider discrete-time Markov
chains (DTMCs) and discrete-time probabilistic automata
(PAs)~\cite{Segala-thesis}. PAs extend labelled transition systems
(LTSs) with probabilistic selection, or, viewed differently, extend
DTMCs with nondeterminism.  They constitute a natural model of
concurrent computation involving random phenomena. In a PA, a labelled
transition leads to a probability distribution over the set of states,
rather than a single state. The resulting model thus exhibits both
non-deterministic choice (as in LTSs) and probabilistic choice (as in
Markov chains).

Strong simulation relations have been
introduced~\cite{JonssonL91,LarsenS91} for probabilistic systems.  For
$s \simrel s'$ ($s'$ strongly simulates $s$), it is required that
every successor distribution of $s$ via action $\alpha$ (called
$\alpha$-successor distribution) has a corresponding
$\alpha$-successor distribution at $s'$.  Correspondence of
distributions is naturally defined with the concept of weight
functions~\cite{JonssonL91}.  In the context of model checking, strong
simulation relations preserve safe PCTL formulas~\cite{SegalaL95}.
Probabilistic simulation~\cite{SegalaL95} is a relaxation of strong
simulation in the sense that it allows for convex combinations of
multiple distributions belonging to equally labelled transitions. More
concretely, it may happen that for an $\alpha$-successor distribution
$\mu$ of $s$, there is no $\alpha$-successor distribution of $s'$
which can be related to $\mu$, yet there exists a so-called
$\alpha$-\emph{combined transition}, a convex combination of several
$\alpha$-successor distributions of $s'$.  Probabilistic simulation
accounts for this and is thus coarser than strong simulation, but
still preserves the same class of PCTL-properties as strong simulation
does.

Apart from discrete time models, this paper considers continuous-time
Markov chains (CTMCs) and continuous-time probabilistic automata
(CPAs)~\cite{Knast69,WolovickJ06}. In CPAs, the transition delays are
governed by exponential distributions. CPAs can be considered also as
extensions of CTMCs with nondeterminism. CPAs are natural foundational
models for various performance and dependability modelling formalisms
including stochastic activity networks~\cite{SandersM91}, generalised
stochastic Petri nets~\cite{MarsanBCDF98} and interactive Markov
chains~\cite{Hermanns02}.  Strong simulation and probabilistic
simulation have been introduced for continuous-time
models~\cite{BKHW05,ZH07}. For CTMCs, $s\simrel s'$ requires that
$s\simrel s'$ holds in the embedded DTMC, and additionally, state $s'$
must be ``faster'' than $s$ which manifests itself by a higher exit
rate. Both strong simulation and probabilistic simulation preserve
safe CSL formulas~\cite{BaierHHK03}, which is a continuous stochastic
extension of PCTL, tailored to continuous-time models.

Weak simulation is proposed in~\cite{BKHW05} for Markov chains. In
weak simulation, the successor states are split into visible and
invisible parts, and the weight function conditions are only imposed
on the transitions leading to the visible parts of the successor
states.  Weak simulation is strictly coarser than the afore-mentioned
strong simulation for Markov chains, thus allows further reduction of
the state space. It preserves the safe PCTL- and CSL-properties
without the next state formulas for DTMCs and CTMCs
respectively~\cite{BKHW05}.

Decision algorithms for strong and weak simulations over Markov
chains, and for strong simulation over probabilistic automata are not
efficient, which makes it as yet unclear whether they can be used as
effectively as their non-probabilistic counterparts.  In this paper we
improve efficient decision algorithms, and devise new algorithms for
deciding strong and strong probabilistic simulations for probabilistic
automata.  Given the simulation preorder, the simulation quotient
automaton is in general smaller than the bisimulation quotient
automaton. Then, for safety and liveness properties, model checking
can be performed on this smaller quotient automata.  The study of
decision algorithms is also important for specification relations: The
model satisfies the specification if the automaton for the
specification simulates the automaton for the model.  In many
applications the specification cannot be easily expressed by logical
formulas: it is rather a probabilistic model itself.  Examples of this
kind include various recent wireless network protocols, such as
ZigBee~\cite{GrossHP07}, Firewire Zeroconf~\cite{BohnenkampSHV03}, or
the novel IEEE~802.11e, where the central mechanism is selecting among
different-sided dies, readily expressible as a probabilistic
automaton~\cite{MangoldZHW04}.

The common strategy used by decision algorithms for simulations is as
follows. The algorithm starts with a relation $R$ which is guaranteed
to be coarser than the simulation preorder $\simrel$. Then, the
relation $R$ is successively be refined.  In each iteration of the
refinement loop, pairs $(s,s')$ are eliminated from the relation $R$
if the corresponding simulation conditions are violated with respect
to the current relation. In the context of labelled transitions
systems, this happens if $s$ has a successor state $t$, but we cannot
find a successor state $t'$ of $s'$ such that $(t,t')$ is also in the
current relation $R$. For DTMCs, this correspondence is formulated by
the existence of a weight function for distributions
$(\Pmu{s},\Pmu{s'})$ with respect to the current relation $R$.
Checking this weight function condition amounts to checking whether
there is a maximum flow over the network constructed out of
$(\Pmu{s},\Pmu{s'})$ and the current relation $R$.  The complexity for
one such check is however rather expensive, it has time complexity
$\O(n^3/\log n)$.  If the iterative algorithm reaches a fix-point, the
strong simulation preorder is obtained. The number of iterations of
the refinement loop is at most $\O(n^2)$, and the overall
complexity~\cite{BEMC00} amounts to $\O(n^7/\log n)$ in time and
$\O(n^2)$ in space.

Fixing a pair $(s,s')$, we observe that the networks for this pair across
iterations of the refinement loop are very similar: They differ from
iteration to iteration only by deletion of some edges induced by the
successive clean up of $R$.  We exploit this by adapting a parametric
maximum flow algorithm~\cite{GalloGT89} to solve the maximum flow problems
for the arising sequences of similar networks, hence arriving at efficient
simulation decision algorithms.  The basic idea is that all computations
concerning the pair $(s,s')$ can be performed in an incremental way: after
each iteration we save the current network together with maximum flow
information.  Then, in the next iteration, we update the network, and
derive the maximum flow while using the previous maximum flow function. The
maximum flow problems for the arising sequences of similar networks with
respect to the pair $(s,s')$ can be computed in time $\O(|V|^3)$ where
$|V|$ is the number of nodes of the network. This leads to an overall time
complexity $\O(m^2n)$ for deciding the simulation preorder. Because of the
storage of the networks, the space complexity is increased to
$\O(m^2)$. Especially in the very common case where the state fanout of a
model is bounded by a constant $g$ (and hence $m\le g n$), our strong
simulation algorithm has time and space complexity $\O(n^2)$.  The
algorithm can be extended easily to handle CTMCs with same time and space
complexity.  For weak simulation on Markov chains, the parametric maximum
flow technique cannot be applied directly.  Nevertheless, we manage to
incorporate the parametric maximum flow idea into a decision algorithm with
time complexity $\O(m^2n^3)$ and space complexity $\O(n^2)$.  An earlier
algorithm~\cite{BHK04} uses LP problems~\cite{Karmarkar84,Schr86} as
subroutines. The maximum flow problem is a special instance of an LP
problem but can be solved much more efficiently~\cite{AMO93}.

We extend the algorithm to compute strong simulation preorder to also
work on PAs. It takes the skeleton of the algorithm for Markov chains:
It starts with a relation $R$ which is coarser than $\simrel$, and
then refines $R$ until $\simrel$ is achieved. In the refinement loop,
a pair $(s,s')$ is eliminated if the corresponding simulation
conditions are violated with respect to the current relation. For PAs,
this means that there exists an $\alpha$-successor distribution $\mu$
of $s$, such that for all $\alpha$-successor distribution $\mu'$ of
$s'$, we cannot find a weight function for $(\mu,\mu')$ with respect
to the current relation $R$. Again, as for Markov chains, the
existence of such weight functions can be reduced to maximum flow
problems. Combining with the parametric maximum flow
algorithm~\cite{GalloGT89}, we arrive at the same time complexity
$\O(m^2n)$ and space complexity $\O(m^2)$ as for Markov chains.
The above maximum flow based procedure cannot be applied to deal with
strong \emph{probabilistic} simulation for PAs. The reason is that an
$\alpha$-combined transition of state $s$ is a convex combination of
several $\alpha$-successor distributions of $s$, thus induces
uncountable many such possible combined transitions.  The
computational complexity of deciding strong probabilistic simulation
has not been investigated before. We show that it can be reduced to
solving LP problems. The idea is that we introduce for each
$\alpha$-successor distribution a variable, and then reformulates the
requirements concerning the combined transitions by linear constraints
over these variables. This allows us to construct a set of LP problem
such that whether a pair $(s,s')$ should be thrown out of the current
pair $R$ is equivalent to whether each of the LP problem has a
solution.

The algorithms for PAs are then extended to handle their
continuous-time analogon, CPAs. In the algorithm, for each pair
$(s,s')$ in the refinement loop, an additional rate condition is
ensured by an additional check via comparing the appropriate rates of
$s$ and $s'$.  The resulting algorithm has the same time and space
complexity.

\paragraph{Related Works}
In the non-probabilistic setting, the most efficient algorithms for
deciding simulation preorders have been proposed
in~\cite{BloomP95,HenzingerHK95}. The complexity is $\O(mn)$ where $n$ and
$m$ denote the number of states and transitions of the transition system
respectively. 
For Markov chains, Derisavi \emph{et al.}~\cite{DHS03} presented an
$\O(m\log n)$ algorithm for strong bisimulation.  Weak bisimulation for
DTMCs can be computed in $\O(n^3)$ time~\cite{BaierH97}.  For strong
simulation, Baier {\it et al.\@}~\cite{BEMC00} introduced a polynomial
decision algorithm with complexity $\O(n^7/\log n)$, by tailoring a network
flow algorithm~\cite{GoldbergT88} to the problem, embedded into an
iterative refinement loop. In \cite{BHK04}, Baier {\it et al.\@} proved
that weak simulation is decidable in polynomial time by reducing it to
linear programming (LP) problems.
%
For a subclass of PAs (reactive systems), Huynh and
Tian~\cite{HuynhT92} presented an $\O(m \log n)$ algorithm for
computing strong bisimulation.  Cattani and Segala~\cite{CattaniS02}
have presented decision algorithms for strong and \emph{bi}simulation
for PAs.  They reduced the decision problems to LP problems.  To
compute the coarsest strong simulation for PAs, Baier
\et~\cite{BEMC00} presented an algorithm which reduces the query
whether a state strongly simulates another to a maximum flow problem.
Their algorithm has complexity $\O((mn^6+m^2n^3)/\log n)$\footnote{The $m$
used in paper~\cite{BEMC00} is slightly different from the $m$ as we use
it. A detailed comparison is provided later, in
Remark~\ref{remark:baier_complexity} of Section~\ref{sec:pp}. }.  Recently,
algorithm for computing simulation and bisimulation metrics for concurrent
games~\cite{CdAMR08} has been studied.

\paragraph{Outline of The Paper} 
The paper proceeds by recalling the definition of the models and simulation
relations in Section~\ref{sec:pre}.  In Section~\ref{sec:flow} we give a
short interlude on maximum flow problems. In Section~\ref{sec:strong} we
present a combinatorial method to decide strong simulations.  In this
section we also introduce new decision algorithms for deciding strong
probabilistic simulations for PAs and CPAs.  In Section~\ref{sec:weak} we
focus on algorithms for weak simulations.  Section~\ref{sec:concl}
concludes the paper.

\section{Preliminaries}\label{sec:pre}
\noindent In Subsection~\ref{sec:mm}, we recall the definitions of fully
probabilistic systems, discrete- and continuous-time Markov
chains~\cite{Steward94}, and the nondeterministic extensions of these
discrete-time~\cite{SegalaL95} and continuous-time
models~\cite{Puter94,BaierHKH05}. In Subsection~\ref{sec:sim_intro} we
recall the definition of simulation relations.

\subsection{Markov Models} \label{sec:mm}
Firstly, we introduce some general notations. Let $X,Y$ be finite
sets. For $f:X\to \Real$, let $f(A)$ denote $\sum_{x\in A}f(x)$ for
all $A\subseteq X$.  For $f:X\times Y \to \Real$, we let $f(x,A)$
denote $\sum_{y\in A}f(x,y)$ for all $x\in X$ and $A\subseteq Y$, and
$f(A,y)$ is defined similarly. Let $AP$ be a fixed, finite set of
atomic propositions.

For a finite set $S$, a distribution $\mu$ over $S$ is a function
$\mu:S\to [0,1]$ satisfying the condition $\mu(S)\le 1$. The support
of $\mu$ is defined by $\support(\mu) = \{s \mid \mu(s) > 0\}$, and
the size of $\mu$ is defined by $\size{\mu} =\size{\support(\mu)}$.
The distribution $\mu$ is called stochastic if $\mu(S)=1$, absorbing
if $\mu(S)=0$. We sometimes use an auxiliary state (not a~\emph{real}
state) $\aux\not\in S$ and set $\mu(\aux)=1-\mu(S)$. If $\mu$ is not
stochastic we have $\mu(\aux)>0$. Further, let $S_\aux$ denote the set
$S\cup\{\aux\}$, and let $\support_\aux(\mu) = \support(\mu) \cup \{
\aux\}$ if $\mu(\aux)>0$ and $\support_\aux(\mu) = \support(\mu)$
otherwise.  We let $\distrs(S)$ denote the set of distributions over
the set $S$.

\begin{defi}
  A labelled fully probabilistic system (FPS) is a tuple $\dtmc$ where
  $S$ is a finite set of states, $\P:S\times S\to [0,1]$ is a
  probability matrix such that $\Pmu{s} \in \distrs(S)$ for
  all $s \in S$, and $L:S\to 2^{AP}$ is a labelling function.
\end{defi}

A state $s$ is called stochastic and absorbing if the distribution
$\Pmu{s}$ is stochastic and absorbing respectively. For $s\in S$, let
$\post(s) = \support(\Pmu{s})$, and let $\post_\aux(s) =
\support_\aux(\Pmu{s})$.

\begin{defi}
  A labelled discrete-time Markov chain (DTMC) is an FPS $\dtmc$ where
  $s$ is either absorbing or stochastic for all $s\in S$.
\end{defi}

FPSs and DTMCs are \emph{time-abstract}, since the duration between
triggering transitions is disregarded. We observe the state only at a
discrete set of time points $0,1,2,\ldots$.  We recall the definition of
CTMCs which are \emph{time-aware}:

\begin{defi}
A labelled continuous-time Markov chain (CTMC) is a
tuple $\ctmc$ with $S$ and $L$ as before, and  $\R:
S\times S\to \Real_{\ge 0}$ is a rate matrix. 
\end{defi}

For CTMC $\M$, let $\post(s)=\{s'\in S\mid \R(s,s')>0\}$ for all $s\in S$.
The rates give the average delay of the corresponding transitions.
Starting from state $s$, the probability that within time $t$ a successor
state is chosen is given by $1-e^{-\R(s,S)t}$.  The probability that a
specific successor state $s'$ is chosen within time $t$ is thus given by
$(1-e^{-\R(s,S)t})\cdot
\R(s,s')/\R(s,S)$. A CTMC induces an embedded DTMC, which
captures the time-abstract behaviour of it:

\begin{defi}
Let $\ctmc$ be a CTMC. The embedded DTMC of $\C$ is defined by
$emb(\C)=(S,\P,L)$ with $\P(s,s')=\R(s,s')/\R(s,S)$ if $\R(s,S)>0$ and $0$
otherwise.
\end{defi}
We will also use $\P$ for a CTMC directly, without referring
to its embedded DTMC explicitly. If one is interested in time-abstract
properties (e.\,g.\@, the probability to reach a set of states) of a
CTMC, it is sufficient to analyse its embedded DTMC.

For a given FPS, DTMC or CTMC, its \emph{fanout} is defined by $\max_{s\in
S} \size{\post(s)}$. The number of states is defined by $n=\size{S}$, and
the number of transitions is defined by $m=\sum_{s\in
S}\size{\post(s)}$. For $s\in S$, $reach(s)$ denotes the set of states that
are reachable from $s$ with positive probability.  For a relation $R
\subseteq S\times S$ and $s\in S$, let $R(s)$ denote the set $\{s'\in
S\mid (s,s')\in R\}$. Similarly, for $s'\in S$, let $R^{-1}(s')$ denote the set $\{s\in S\mid
(s,s')\in R\}$. If $(s,s')\in R$, we write also $s\mathrel{R}s'$.

Markov chains are purely probabilistic.  Now we consider extensions of
Markov chains with nondeterminism.  We first recall the definition of
probabilistic automata, which can be considered as the \emph{simple
  probabilistic automata} with transitions allowing deadlocks
in~\cite{Segala-thesis}.

\begin{defi}
  A probabilistic automaton (PA) is a tuple $\pa$ where
  $S$ and $L$ are defined as before,
  $Act$ is a finite set of actions,
  $\P\subseteq S\times Act\times \distrs(S)$ is a finite set,
  called the probabilistic transition matrix.
\end{defi}

For $(s,\alpha,\mu)\in\P$, we use $s\transby{\alpha}\mu$ as a
shorthand notation, and call $\mu$ an $\alpha$-successor distribution
of $s$. Let $Act(s)=\{\alpha \mid \exists \mu: s\transby{\alpha}
\mu\}$ denote the set of actions enabled at $s$. For $s\in S$ and
$\alpha\in Act(s)$, let $Steps_\alpha(s)=\{\mu\in \distrs(S)\mid
s\transby{\alpha}\mu\}$ and $\steps(s)=\bigcup_{\alpha\in
Act(s)}Steps_\alpha(s)$. The fanout of a state $s$ is defined by
$\fanout(s) =
\sum_{\alpha \in Act(s)}\sum_{\mu\in Steps_\alpha(s)} (\size{\mu}+1)$.
Intuitively, $\fanout(s)$ denotes the total sum of the sizes of outgoing
distributions of state $s$ plus their labelling. The fanout of $\M$ is
defined by $\max_{s\in S}\fanout(s)$. Summing up over all states, we define
the size of the transitions by $m=\sum_{s \in S}
\fanout(s)$.

A \emph{Markov decision process} (MDP)~\cite{Puter94} arises from a
PA $\M$ if for $s\in S$ and $\alpha\in Act$, there is at most
one $\alpha$-successor distribution $\mu$ of $s$, which must be
stochastic.

We consider a continuous-time counterpart of PAs where the
transitions are described by rates instead of probabilities. A rate
function is simply a function $r: S \to \Real_{\ge 0}$. Let
$\size{r}=\size{\{s \mid r(s)>0\}}$ denote the size of $r$. Let $\rate(S)$
denote the set of all rate functions.
\begin{defi}
  A continuous-time PA (CPA) is a tuple $(S, Act, \R, L)$ where $S$,
  $Act$, $L$ as defined for PAs, and $\R \subseteq S\times Act\times
  \rate(S)$ a finite set, called the rate matrix.
\end{defi}
We write $s \transby{\alpha} r$ if $(s,\alpha,r) \in \R$, and call $r$ an
$\alpha$-successor rate function of $s$. For transition $s \transby{\alpha}
r$, the sum $r(S)$ is also called the exit rate of it. Given that the
transition $s\transby{\alpha} r$ is chosen from state $s$, the probability
that any successor state is chosen within time $t$ is given by
$1-e^{-r(S)t}$, and a specific successor state $s'$ is chosen within time
$t$ is given by $(1-e^{-r(S)t})\cdot \frac{r(s')}{r(S)}$.  The notion of
$Act(s)$, $Steps_\alpha(s)$, $\steps(s)$, fanout and size of transitions
for PAs can be extended to CPAs by replacing occurrence of distribution
$\mu$ by rate function $r$ in an obvious way.

The model continuous-time Markov decision processes
(CTMDPs)~\cite{Puter94,BaierHKH05} can be considered as special CPAs where
for $s\in S$ and $\alpha \in Act$, there exists at most one rate
function $r\in \rate(S)$ such that $s \transby{\alpha} r$. The model
CTMDPs considered in paper~\cite{WolovickJ06} essentially agree with
our CPAs.

\subsection{Strong and Weak Simulation Relations}
\label{sec:sim_intro}
We first recall the notion of strong simulation on Markov
chains~\cite{BKHW05}, PAs~\cite{SegalaL95}, and CPAs~\cite{ZH07}.  Strong
probabilistic simulation is defined in Subsection~\ref{sec:spsdef}.  Weak
simulation for Markov chains will be given in Subsection~\ref{sec:ws}. The
notion of simulation up to $R$ is introduced in
Subsection~\ref{subs:simulation_upto_R}.

 \subsubsection{Strong Simulation}
\label{sec:ss}
Strong simulation requires that each successor distribution of one
state has a corresponding successor distribution of the other
state. The correspondence of distributions is naturally defined with
the concept of \emph{weight functions}~\cite{JonssonL91}, adapted to FPSs
as in~\cite{BKHW05}. 
For a relation $R
\subseteq S\times S$, we let $R_\aux$ denote the set $R \cup \{(\aux,s)\mid s \in S_\aux\}$.

\begin{defi}\label{def:weight}
  Let $\mu,\mu'\in\dist(S)$ and $R\subseteq S\times S$. A weight function
  for $(\mu, \mu')$ \wrt $R$ is a function $\Delta:S_\aux\times
  S_\aux\rightarrow[0,1]$ such that 
\begin{enumerate}[(1)] 
  \item
  $\Delta(s,s')>0$ implies $s\mathrel{R_\aux}s'$, 
  \item \label{def:weight_relate_mu_Delta} $\mu(s)=\Delta(s,S_\aux)$ for $s\in
  S_\aux$ and 
  \item \label{def:weight_relate_muprime_Delta} $\mu'(s')=\Delta(S_\aux,s')$ for $s'\in S_\aux$.  
\end{enumerate} 

We write $\mu\weight_R\mu'$ if there exists a weight function for
$(\mu,\mu')$ \wrt $R$.
\end{defi}

The first condition requires that only pairs $(s,s')$ in the relation
$R_\aux$ have a positive weight. In other words, for $s,s'\in S$ with $s'
\not \in R_\aux(s)$, it holds that $\Delta(s,s')=0$. Strong simulation requires
similar states to be related via weight functions on their
distributions~\cite{JonssonL91}.

\begin{defi}
\label{def:strongsimulation-fps}
Let $\dtmc$ be an FPS, and let $R\subseteq S\times S$.  The relation $R$ is
a strong simulation on $\D$ iff for all $s_1,s_2$ with $s_1\mathrel{R}s_2$:
$L(s_1)=L(s_2)$ and $\Pmu{s_1}
\weight_R \Pmu{s_2}$.

We say that $s_2$ strongly
simulates $s_1$ in $\D$, denoted by $s_1\simrel_\M s_2$, iff there
exists a strong simulation $R$ on $\D$ such that $s_1\mathrel{R}s_2$.
\end{defi}

By definition, it can be shown~\cite{BKHW05} that $\simrel_\M$ is reflexive
and transitive, thus a
\emph{preorder}. Moreover, $\simrel_\M$ is the coarsest strong simulation relation for
$\M$. If the model $\M$ is clear from the context, the subscript $\M$ may
be omitted. Assume that $s_1\simrel s_2$ and let $\Delta$ denote the
corresponding weight function. If $\P(s_2,\bot)>0$, we have that
$\P(s_2,\bot)=\sum_{s\in S_\bot}\Delta(s,\bot)=\Delta(\aux,\aux)$. The
second equality follows by the fact that $\bot$ can not strongly simulate
any real state in $S$. Another observation is that if $s$ is absorbing,
then it can be strongly simulated by any other state $s'$ with
$L(s)=L(s')$.

\begin{figure}[tbp]
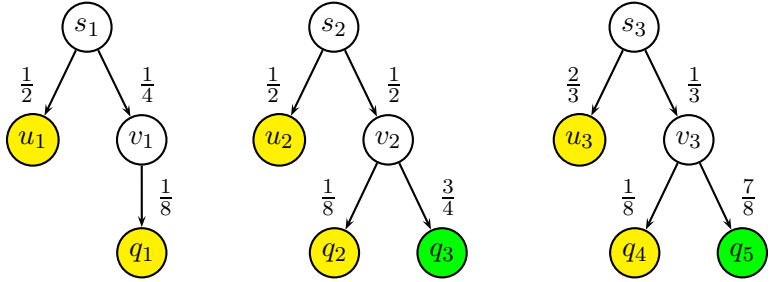

\begin{center}
\psset{levelsep=1.5cm,arrows=->}
\pstree{\Tcircle{$s_1$}}{%
  \Tcircle[fillstyle=solid,fillcolor=yellow]{$u_1$} \tlput{$\frac{1}{2}$}
  \pstree{\Tcircle{$v_1$} \trput{$\frac{1}{4}$} }{%
  \Tcircle[fillstyle=solid,fillcolor=yellow]{$q_1$} \trput{$\frac{1}{8}$} }
  }\hspace{1cm}
\pstree{\Tcircle{$s_2$}}{%
  \Tcircle[fillstyle=solid,fillcolor=yellow]{$u_2$}  \tlput{$\frac{1}{2}$}
  \pstree{\Tcircle{$v_2$} \trput{$\frac{1}{2}$} }{%
    \Tcircle[fillstyle=solid,fillcolor=yellow]{$q_2$} \tlput{$\frac{1}{8}$}
    \Tcircle[fillstyle=solid,fillcolor=BlueGreen]{$q_3$} \trput{$\frac{3}{4}$}
  }
}\hspace{1cm}
\pstree{\Tcircle{$s_3$}}{%
  \Tcircle[fillstyle=solid,fillcolor=yellow]{$u_3$}  \tlput{$\frac{2}{3}$}
  \pstree{\Tcircle{$v_3$} \trput{$\frac{1}{3}$} }{%
    \Tcircle[fillstyle=solid,fillcolor=yellow]{$q_4$} \tlput{$\frac{1}{8}$}
    \Tcircle[fillstyle=solid,fillcolor=BlueGreen]{$q_5$} \trput{$\frac{7}{8}$}
  }
}
\end{center}
\caption{An FPS for illustrating the simulation relations\protect\footnotemark.}
\label{fig:strongsim_fps}
\end{figure}
\footnotetext{Although
                this graph is not connected, it shows a single FPS.
                Similarly, later figures will show a single DTMC, CTMC
                etc.}
\begin{exa}
\label{exa:ss_fps}
Consider the FPS depicted in Figure~\ref{fig:strongsim_fps}. Recall that
labelling of states is indicated by colours in the states.  Since the
yellow (grey) states are absorbing, they strongly simulate each other. The
same holds for the green (dark grey) states.  We show now that $s_1\simrel
s_2$ but $s_2
\not\simrel s_3$.

Consider first the pair $(s_1,s_2)$. Let
$R=\{(s_1,s_2),(u_1,u_2),(v_1,v_2), (q_1,q_2)\}$. We show that $R$ is a
strong simulation relation. First observe that $L(s)=L(s')$ for all
$(s,s')\in R$. Since states $u_1,q_1$ are absorbing, the conditions for the
pairs $(u_1,u_2)$ and $(q_1,q_2)$ hold trivially.  To show the conditions
for $(v_1, v_2)$, we consider the function $\Delta_1$ defined by:
$\Delta_1(q_1,q_2)=\frac{1}{8}$, $\Delta_1(\aux,q_3)=\frac{3}{4}$,
$\Delta_1(\aux,\aux)=\frac{1}{8}$ and $\Delta_1(\cdot)=0$ otherwise. It is
easy to check that $\Delta_1$ is a weight function for
$(\Pmu{v_1},\Pmu{v_2})$ with respect to $R$.  Now consider $(s_1,s_2)$.
The weight function $\Delta_2$ for $(\Pmu{s_1},\Pmu{s_2})$ with respect to
$R$ is given by $\Delta_2(u_1,u_2)=\frac{1}{2}$ and
$\Delta_2(v_1,v_2)=\Delta_2(\aux,v_2)=\frac{1}{4}$ and $\Delta_2(\cdot)=0$
otherwise. Thus $R$ is a strong simulation which implies that $s_1 \simrel
s_2$.

Consider the pair $(s_2,s_3)$.  Since $\P(s_2,v_2)=\frac{1}{2}$, to
establish the Condition~\ref{def:weight_relate_mu_Delta} of
Definition~\ref{def:weight}, we should have
$\frac{1}{2}=\Delta(v_2,S_\aux)$. Observe that $v_3$ is the only successor
state of $s_3$ which can strongly simulate $v_2$, thus
$\Delta(v_2,S_\aux)= \Delta(v_2,v_3)$. However, for state $v_3$
we have $\P(s_3,v_3)<\Delta(v_2,v_3)$, which violates the
Condition~\ref{def:weight_relate_muprime_Delta} of
Definition~\ref{def:weight}, thus we cannot find such a weight
function. Hence, $s_2\not\simrel s_3$.
\end{exa}

Since each DTMC is a special case of an FPS,
Definition~\ref{def:strongsimulation-fps} applies directly for DTMCs.  For
CTMCs we say that $s_2$ strongly simulates $s_1$ if, in addition to the
DTMC conditions, $s_2$ can move stochastically
\emph{faster} than $s_1$~\cite{BKHW05}, which manifests itself by a
higher rate.

\begin{defi}
\label{def:ctmc}\label{def:strongsimulation-ctmc}
Let $\ctmc$ be a CTMC and let $R\subseteq S\times S$. The relation $R$ is a
strong simulation on $\C$ iff for all $s_1,s_2$ with $s_1\mathrel{R}s_2$:
$L(s_1)=L(s_2)$, $\Pmu{s_1}
\weight_R \Pmu{s_2}$ and $\R(s_1,S)\le \R(s_2,S)$.

We say that $s_2$ strongly simulates $s_1$ in $\C$, denoted by
$s_1\simrel_\M s_2$, iff there exists a strong simulation $R$ on $\C$
such that $s_1\mathrel{R}s_2$.
\end{defi}

Thus, $s\simrel_\M s'$ holds if $s\simrel_{emb(\M)} s'$, and $s'$ is faster
than $s$. By definition, it can be shown that $\simrel_\M$ is a preorder,
and is the coarsest strong simulation relation for $\M$.  For PAs, strong
simulation requires that every $\alpha$-successor distribution of $s_1$ is
related to an $\alpha$-successor distribution of $s_2$ via a weight
function~\cite{SegalaL95,JonssonL91}:
\begin{defi}
\label{def:simple_ss}
Let $\pa$ be a PA and let $R \subseteq S \times S$. The relation $R$ is a
strong simulation on $\mathcal{M}$ iff for all $s_1,s_2$ with $s_1
\mathrel{R} s_2$: $L(s_1)=L(s_2)$ and if $s_1\transby{\alpha}\mu_1$ then
there exists a transition $s_2\transby{\alpha}\mu_2$ with $\mu_1 \weight_R
\mu_2$.

We say that $s_2$ strongly simulates $s_1$ in $\M$, denoted $s_1
\simrel_\M s_2$, iff there exists a strong simulation $R$ on $\M$
such that $s_1 \mathrel{R} s_2$.
\end{defi}

\begin{figure}[tbp]
  \begin{center}
    \psset{levelsep=1.5cm,arrows=->,nodesep=0pt,labelsep=0pt,npos=0.3}
    \begin{pspicture}(0,-0.2)(4.1, 3)
      \rput(1.75,2.5){\circlenode{s1}{$s_1$}}
      \rput(-0.3,0){\circlenode[fillstyle=solid,fillcolor=yellow]{s2}{$u_1$}}
      \rput(0.5,0){\circlenode[fillstyle=solid,fillcolor=BlueGreen]{s3}{$v_1$}}
      \rput(3,0){\circlenode[fillstyle=solid,fillcolor=yellow]{s6}{$u_2$}}
      \rput(3.8,0){\circlenode[fillstyle=solid,fillcolor=BlueGreen]{s7}{$v_2$}}
      \ncline{s1}{s2} \nbput[npos=0.7]{.4}\nbput{$\alpha$}
      \ncline{s1}{s3} \naput[npos=0.7]{.6} 
      \ncline{s1}{s6} \naput[npos=0.7]{.6}
      \ncline{s1}{s7} \naput[npos=0.7]{.4} \naput{$\alpha$}
      \psarc{-}(1.75,2){.6}{200}{223} 
      \psarc{-}(1.75,2){.6}{317}{340} 
    \end{pspicture}\hspace{2.5cm}
    \begin{pspicture}(0,-0.2)(4.1, 3)
      \rput(1.75,2.5){\circlenode{s1}{$s_2$}}
      \rput(-0.3,0){\circlenode[fillstyle=solid,fillcolor=yellow]{s2}{$u_3$}}
      \rput(0.5,0){\circlenode[fillstyle=solid,fillcolor=BlueGreen]{s3}{$v_3$}}
      \rput(1.4,0){\circlenode[fillstyle=solid,fillcolor=yellow]{s4}{$u_4$}}
      \rput(2.1,0){\circlenode[fillstyle=solid,fillcolor=BlueGreen]{s5}{$v_4$}}
      \rput(3,0){\circlenode[fillstyle=solid,fillcolor=yellow]{s6}{$u_5$}}
      \rput(3.8,0){\circlenode[fillstyle=solid,fillcolor=BlueGreen]{s7}{$v_5$}}
      \ncline{s1}{s2} \nbput[npos=0.7]{.4} \nbput{$\alpha$}
      \ncline{s1}{s3} \nbput[npos=0.7]{.6} 
      \ncline{s1}{s4} \nbput[npos=0.7]{.5}
      \ncline{s1}{s5} \naput[npos=0.7]{.5} \naput{$\alpha$}
      \ncline{s1}{s6} \naput[npos=0.7]{.6}
      \ncline{s1}{s7} \naput[npos=0.7]{.4} \naput{$\alpha$}
      \psarc{-}(1.75,2){.6}{200}{223} 
      \psarc{-}(1.75,2){.6}{317}{340} 
      \psarc{-}(1.75,2){.5}{256}{284} 
    \end{pspicture}
  \end{center}
\caption{A PA for illustrating the simulation relations.}
\label{fig:pa_one}
\end{figure}
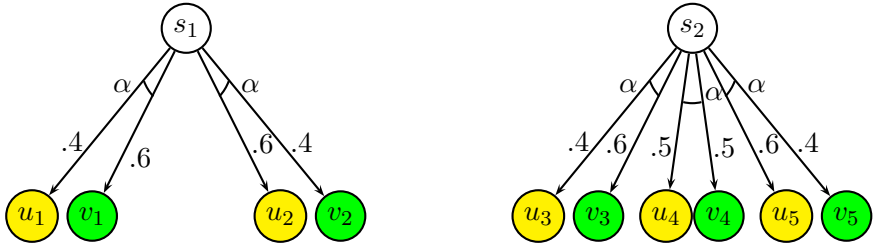

\begin{exa}\label{fig:pa_strongsimulation}
  Consider the PA in Figure~\ref{fig:pa_one}. Then, it is easy to check
  $s_1 \simrel s_2$: each $\alpha$-successor distribution of $s_1$ has a
  corresponding $\alpha$-successor distribution of $s_2$. However, $s_1$
  does not strongly simulate $s_2$, as the middle $\alpha$-successor
  distribution of $s_2$ can not be related by any $\alpha$-successor
  distribution of $s_1$.
\end{exa}

Now we consider CPAs. For a rate function $r$, we let $\mu(r) \in
\distrs(S)$ denote the induced distribution defined by: if $r(S)>0$
then $\mu(r)(s)$ equals $r(s)/r(S)$ for all $s\in S$, and if $r(S)=0$,
then $\mu(r)(s)=0$ for all $s\in S$. Now we introduce the notion of
strong simulation for CPAs~\cite{ZH07}, which can be considered as an
extension of the definition for CTMCs~\cite{BKHW05}:

\begin{defi}
Let $\cpa$ be a CPA and let $R \subseteq S \times S$. The relation $R$ is a
strong simulation on $\mathcal{M}$ iff for all $s_1,s_2$ with $s_1
\mathrel{R} s_2$: $L(s_1)=L(s_2)$ and if $s_1\transby{\alpha} r_1$
then there exists a transition $s_2\transby{\alpha} r_2$ with
$\mu(r_1) \weight_R \mu(r_2)$ and $r_1(S) \le r_2(S)$.

We write $s_1 \simrel_\M s_2$ iff there exists a strong simulation
$R$ on $\M$ such that $s_1 \mathrel{R} s_2$.
\end{defi}

Similar to CTMCs, the additional rate condition $r_1(S) \le r_2(S)$
indicates that the transition $s_2 \transby{\alpha} r_2$ is faster than
$s_1 \transby{\alpha} r_1$.  As a shorthand notation, we use $r_1\weight_R
r_2$ for the condition $\mu(r_1)\weight_R \mu(r_2)$ and $r_1(S)\le
r_2(S)$. For both PAs and CPAs, $\simrel_\M$ is the coarsest strong
simulation relation.

\subsubsection{Strong Probabilistic Simulations}\label{sec:spsdef}

We recall the definition of strong probabilistic simulation, which is
coarser than strong simulation, but still preserves the same class of
PCTL-properties as strong simulation does.  We first recall  the notion
of combined transition~\cite{Segala-thesis}, a convex combination of several
equally labelled transitions:

\begin{defi}
\label{def:combined}
  Let $\pa$ be a PA. Let $s\in S$, $\alpha\in Act(s)$ and $k=
  \size{\steps_\alpha(s)}$. Assume that $\steps_\alpha(s) = \{ \mu_1,
  \ldots, \mu_k \}$.  The tuple $(s,\alpha,\mu)$ is a combined transition,
  denoted by $s \combinedby{\alpha} \mu$, iff there exist constants $c_1,
  \ldots, c_k \in [0,1]$ with $\sum_{i=1}^k c_i=1$ such that $\mu =
  \sum_{i=1}^k c_i \mu_i$.
\end{defi}

The key difference to Definition~\ref{def:simple_ss} is
the use of $\combinedby{\alpha}$ instead of $\transby{\alpha}$:

\begin{defi}
\label{def:combined_ss}
Let $\pa$ be a PA and let $R \subseteq S \times S$. The relation $R$ is a
strong probabilistic simulation on $\mathcal{M}$ iff for all $s_1,s_2$ with
$s_1 \mathrel{R} s_2$: $L(s_1) = L(s_2)$ and if $s_1\transby{\alpha} \mu_1$
then there exists a combined transition $s_2 \combinedby{\alpha} \mu_2$
with $\mu_1 \weight_R \mu_2$.

We write $s_1 \simrel^p_\M s_2$ iff there exists a strong
probabilistic simulation $R$ on $\M$ such that $s_1 \mathrel{R} s_2$.
\end{defi}

Strong probabilistic simulation is insensitive to combined
transitions\footnote{The combined transition defined
in~\cite{Segala-thesis} is more general in two dimensions: First, successor
distributions are allowed to combine different actions. Second,
$\sum_{i=1}^kc_i\le 1$ is possible. The induced strong probabilistic
probabilistic simulation preorder is, however, the same. }, thus, it is a
relaxation of strong simulation. Similar to strong simulation,
$\simrel_\M^p$ is the coarsest strong probabilistic simulation relation for
$\M$.  Since MDPs can be considered as special PAs, we obtain the notions
of strong simulation and strong probabilistic simulation for
MDPs. Moreover, these two relations coincide for MDPs as, by definition,
for each state there is at most one successor distribution per action.

\begin{exa}
  We consider again the PA depicted in Figure~\ref{fig:pa_one}. From
  Example~\ref{fig:pa_strongsimulation} we know that $s_2\not\simrel
  s_1$. In comparison to state $s_1$, state $s_2$ has one additional
  $\alpha$-successor distribution: to states $u_4$ and $v_4$ with
  equal probability $0.5$. This successor distribution can be
  considered as a combined transition of the two successor
  distributions of $s_1$: each with constant $0.5$. Hence, we have
  $s_2 \simrel^p s_1$.
\end{exa}

We extend the notion of strong probabilistic simulation for PAs to
CPAs. First, we introduce the notion of combined transitions for CPAs. In
CPAs the probability that a transition occurs is exponentially
distributed. The combined transition should also be exponentially
distributed. The following example shows that a straightforward extension
of Definition~\ref{def:combined} does not work.

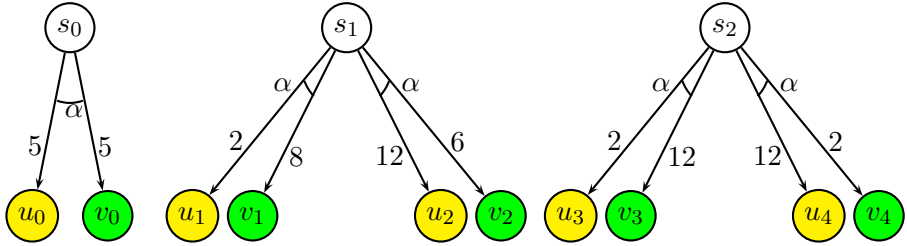
\begin{figure}[tbp]  
  \begin{center}
    \psset{levelsep=1.5cm,arrows=->,nodesep=0pt,labelsep=0pt,npos=0.3}
    \begin{pspicture}(0,0)(1, 2.5)
      \rput(0.5,2.5){\circlenode{s1}{$s_0$}}
      \rput(0,0){\circlenode[fillstyle=solid,fillcolor=yellow]{u1}{$u_0$}}
      \rput(1,0){\circlenode[fillstyle=solid,fillcolor=BlueGreen]{v1}{$v_0$}}
      \ncline{s1}{u1} \nbput[npos=0.7]{5} \naput[npos=0.4,labelsep=3pt]{$\alpha$}
      \ncline{s1}{v1} \naput[npos=0.7]{5}
      \psarc{-}(0.5,2){.5}{250}{290}
    \end{pspicture}\hspace{1cm}
    \begin{pspicture}(-0.3,0)(4.1, 2.7)
      \rput(1.75,2.5){\circlenode{s1}{$s_1$}}
      \rput(-0.3,0){\circlenode[fillstyle=solid,fillcolor=yellow]{s2}{$u_1$}}
      \rput(0.5,0){\circlenode[fillstyle=solid,fillcolor=BlueGreen]{s3}{$v_1$}}
      \rput(3,0){\circlenode[fillstyle=solid,fillcolor=yellow]{s6}{$u_2$}}
      \rput(3.8,0){\circlenode[fillstyle=solid,fillcolor=BlueGreen]{s7}{$v_2$}}
      \ncline{s1}{s2} \nbput[npos=0.7]{2}\nbput{$\alpha$}
      \ncline{s1}{s3} \naput[npos=0.7]{8} 
      \ncline{s1}{s6} \nbput[npos=0.7]{12}
      \ncline{s1}{s7} \naput[npos=0.7]{6} \naput{$\alpha$}
      \psarc{-}(1.75,2){.6}{200}{223} 
      \psarc{-}(1.75,2){.6}{317}{340} 
    \end{pspicture}\hspace{0.5cm}
    \begin{pspicture}(-0.3,0)(4.1, 2.7)
      \rput(1.75,2.5){\circlenode{s1}{$s_2$}}
      \rput(-0.3,0){\circlenode[fillstyle=solid,fillcolor=yellow]{s2}{$u_3$}}
      \rput(0.5,0){\circlenode[fillstyle=solid,fillcolor=BlueGreen]{s3}{$v_3$}}
      \rput(3,0){\circlenode[fillstyle=solid,fillcolor=yellow]{s6}{$u_4$}}
      \rput(3.8,0){\circlenode[fillstyle=solid,fillcolor=BlueGreen]{s7}{$v_4$}}
      \ncline{s1}{s2} \nbput[npos=0.7]{2}\nbput{$\alpha$}
      \ncline{s1}{s3} \naput[npos=0.7]{12} 
      \ncline{s1}{s6} \nbput[npos=0.7]{12}
      \ncline{s1}{s7} \naput[npos=0.7]{2} \naput{$\alpha$}
      \psarc{-}(1.75,2){.6}{200}{223} 
      \psarc{-}(1.75,2){.6}{317}{340} 
    \end{pspicture}
  \end{center}
  \caption{A Continuous-time Probabilistic Automaton.}
  \label{fig:simple_cpa}
\end{figure}

\begin{exa}\label{exa:cpa}
  For this purpose we consider the CPA in Figure~\ref{fig:simple_cpa}. Let
  $r_1$ and $r_2$ denote left and the right $\alpha$-successor rate
  functions out of state $s_1$. Obviously, they have different exit rates:
  $r_1(S)= 10$, $r_2(S) = 18$. Taking each with probability $0.5$, we would
  get the combined transition $r=0.5r_1+0.5r_2$: $r(\{ u_1, u_2 \})=7$ and
  $r(\{v_1, v_2\})=7$. However, $r$ is hyper-exponentially distributed: the
  probability of reaching yellow (grey) states ($u_1$ or $u_2$) within time
  $t$ under $r$ is given by: $ 0.5 \cdot \frac{2}{10} \cdot (1-e^{-10t}) +
  0.5 \cdot \frac{12}{18} \cdot (1-e^{-18t}) $.  Similarly, the probability
  of reaching green (dark grey) states within time $t$ is given by: $0.5
  \cdot \frac{8}{10} \cdot (1-e^{-10t}) + 0.5 \cdot
\frac{6}{18} \cdot (1-e^{-18t}) $. 

From state $s_2$, the two $\alpha$-successor rate functions have the same
exit rate $14$. Let $r_1'$ and $r_2'$ denote left and the right
$\alpha$-successor rate functions out of state $s_2$. In this case the
combined transition $r'=0.5r_1'+0.5r_2'$ is also exponentially distributed
with rate $14$: the probability to reach yellow (grey) states ($u_3$ and
$u_4$) within time $t$ is $\frac{7}{14}
\cdot (1-e^{-14t})$, which is the same as the probability of reaching green
(dark grey) states within time $t$.
\end{exa}

Based on the above example, it is easy to see that to get a combined
transition which is still exponentially distributed, we must consider rate
functions with the same exit rate:

\begin{defi}\label{def:combined_cpa}
  Let $\cpa$ be a CPA. Let $s\in S$, $\alpha\in Act(s)$ and let $ \{ r_1,
  \ldots, r_k \} \subseteq \steps_\alpha(s)$ where $r_i(S) = r_j(S)$ for
  $i,j\in \{1,\ldots, k\}$. The tuple $(s,\alpha,r)$ is a combined
  transition, denoted by $s \combinedby{\alpha} r$, iff there exist
  constants $c_1, \ldots, c_k \in [0,1]$ with $\sum_{i=1}^k c_i=1$ such
  that $r = \sum_{i=1}^k c_i r_i$.
\end{defi}

In the above definition, unlike for the PA case, only $\alpha$-successor
rate functions with the same exit rate are combined together. Similar to
PAs, strong probabilistic simulation is insensitive to combined
transitions, which is thus a relaxation of strong simulation:

\begin{defi}\label{def:cpa_sps}
Let $\cpa$ be a CPA and let $R \subseteq S \times S$. The relation $R$ is a
strong probabilistic simulation on $\M$ iff for all $s_1,s_2$ with $s_1
\mathrel{R} s_2$: $L(s_1) = L(s_2)$ and if $s_1\transby{\alpha} r_1$ then
there exists a combined transition $s_2
\combinedby{\alpha} r_2$ with $r_1 \weight_R r_2$.

We write $s_1 \simrel_\M^p s_2$ iff there exists a strong simulation
$R$ on $\M$ such that $s_1 \mathrel{R} s_2$.
\end{defi}

Recall $r_1\weight_R r_2$ is a shorthand notation for $\mu(r_1)
\weight_R \mu(r_2)$ and $r_1(S) \le r_2(S)$. By definition, the defined
strong probabilistic simulation $\simrel_\M^p$ is the coarsest strong
probabilistic simulation relation for $\M$.

\begin{exa}
  Reconsider the CPA in Figure~\ref{fig:simple_cpa}. As discussed in
  Example~\ref{exa:cpa}, the two $\alpha$-successor rate functions of $s_1$
  cannot be combined together, thus the relation $s_0 \simrel^p s_1$ cannot
  be established. However, $s_0 \simrel^p s_2$ holds: denoting the left rate
  function of $s_2$ as $r_1$ and the right rate function as $r_2$, we
  choose as the combined rate function $r=0.5r_1+0.5r_2$.  Obviously, the
  conditions in Definition~\ref{def:cpa_sps} are satisfied.
\end{exa}

\subsubsection{Weak Simulations}
\label{sec:ws}
We now recall the notion of weak simulation~\cite{BKHW05} on Markov
chains\footnote{In~\cite{ZHEJ07}, we have also considered decision
algorithm for weak simulation for FPSs, which is defined
in~\cite{BKHW05}. However, as indicated in~\cite{thesis}, the proposed weak
simulation for FPSs contains a subtle flaw, which cannot be fixed in an
obvious way. Thus, in this paper we restrict to weak simulation on DTMCs
and CTMCs.}.  Intuitively, $s_2$ weakly simulates $s_1$ if they have the
same labelling, and if their successor states can be grouped into sets
$U_i$ and $V_i$ for $i=1,2$, satisfying certain conditions. Consider
Figure~\ref{fig:split}. We can view steps to $V_i$ as \emph{stutter} steps
while steps to $U_i$ are
\emph{visible} steps. With respect to the visible steps, it is then
required that there exists a weight function for the conditional
distributions: ${\Pmu{s_1}}/{K_1}$ and ${\Pmu{s_2}}/{K_2}$ where
$K_i$ intuitively is the probability to perform a visible step from
$s_i$. The stutter steps must respect the weak simulation relations: thus states
in $V_2$ should weakly simulate $s_1$, and state $s_2$ should weakly
simulate states in $V_1$. This is depicted by dashed arrows in the figure.
For reasons we will explain later in Example~\ref{exa:split_dtmc}, the
definition needs to account for states which partially belong to $U_i$ and
partially to $V_i$.  Technically, this is achieved by functions $\delta_i$
that distribute $s_i$ over $U_i$ and $V_i$ in the definition
below.  For a given pair $(s_1,s_2)$ and  functions
$\delta_i:S\rightarrow[0,1]$, let
$U_{\delta_i},V_{\delta_i}\subseteq S$ (for $i=1,2$) denote the sets
\begin{equation}
U_{\delta_i}=\{u\in \post(s_i)\mid \delta_i(u)>0\},\ 
V_{\delta_i}= \{v\in \post(s_i)\mid \delta_i(v)<1\}
\label{eq:weaksimUV}
\end{equation}
If $(s_1,s_2)$ and $\delta_i$ are clear from the context, we write $U_i$,
$V_i$ instead.

\begin{figure}[tbp]  
\includegraphics[scale=.4]{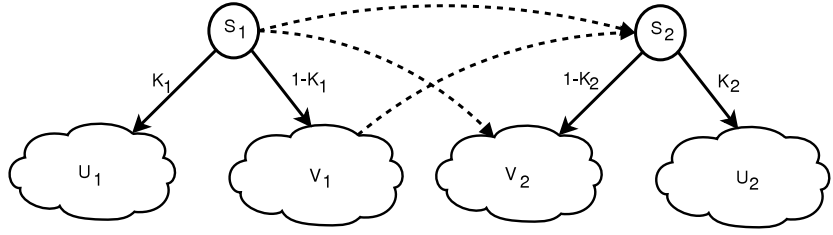}
  \caption{Splitting of successor states for weak simulations.}
  \label{fig:split}
\end{figure}

\begin{defi}  {\label{def:dtmc_weak_simulation}} 
Let $\dtmc$ be a DTMC and let $R\subseteq S\times S$. The relation $R$ is a
weak simulation on $\D$ iff for all $s_1,s_2$ with $s_1\mathrel{R}s_2$:
$L(s_1)=L(s_2)$ and there exist functions $\delta_i:S\rightarrow[0,1]$ such
that:
\begin{enumerate}[(1)]
\item   \label{ws:stuttersimulates}
        (a) $v_1\mathrel{R}s_2$ for all $v_1\in V_1$,
  and (b) $s_1\mathrel{R}v_2$ for all $v_2\in
V_2$
\item   \label{ws:weightfunction}
        there exists a function $\Delta:S\times S\rightarrow [0,1]$
  such that:
  \begin{enumerate}[(a)]
  \item \label{ws:wfgt0} $\Delta(u_1,u_2)>0$ implies $u_1\in U_1,u_2\in U_2$ and
    $u_1\mathrel{R}u_2$.
  \item \label{ws:wfsums} if $K_1>0$ and $K_2>0$ then for all states $w\in S$:
    \[K_1\cdot \Delta(w,U_2) = \P(s_1,w)\delta_1(w),\ 
    K_2\cdot \Delta(U_1,w) =
    \P(s_2,w)\delta_2(w)\] 
    where $K_i=\sum_{u_i\in U_i}\delta_i(u_i)\cdot\P(s_i,u_i)$ for
    $i=1,2$.
  \end{enumerate}
\item   \label{ws:reachability}
        for $u_1\in U_1$ there exists a path fragment
  $s_2,w_1,\ldots, w_n,u_2$ with positive probability such that $n\ge
  0$, $s_1\mathrel{R}w_j$ for $0<j\le n$, and $u_1\mathrel{R}u_2$.
\end{enumerate}
We say that $s_2$ weakly simulates $s_1$ in $\D$, denoted
$s_1\wsrel_\D s_2$, iff there exists a weak simulation
$R$ on $\D$ such that $s_1\mathrel{R}s_2$.
\end{defi}
Note again that the sets $U_i,V_i$ in the above definition are defined
according to Equation~\ref{eq:weaksimUV} with respect to the pair
$(s_1,s_2)$ and the functions $\delta_i$. The functions $\delta_i$ can be
considered as a generalisation of the characteristic function of $U_i$ in
the sense that we may \emph{split} the membership of a state to $U_i$ and
$V_i$ into
\emph{fragments} which sum up to $1$. For example, if
$\delta_1(s)=\frac{1}{3}$, we say that $\frac{1}{3}$ fragment of the state
$s$ belongs to $U_1$, and $\frac{2}{3}$ fragment of $s$ belongs to
$V_1$. Hence, $U_i$ and $V_i$ are not necessarily disjoint.  Observe that
$U_i=\emptyset$ implies that $\delta_i(s)=0$ for all $s\in S$. Similarly,
$V_i=\emptyset$ implies that $\delta_i(s)=1$ for all $s\in S$.

Condition~\ref{ws:reachability} will in the sequel be called the
\emph{reachability condition}. If  $K_1>0$ and
$K_2=0$, which implies that $U_2=\emptyset$ and $U_1\not=\emptyset$, the
reachability condition guarantees that for any visible step $s_1\to u_1$
with $u_1\in U_1$, $s_2$ can reach a state $u_2$ which simulates $u_1$
while passing only through states simulating $s_1$.  Assume that we have
$S=\{s_1,s_2,u_1\}$ where $L(s_1)=L(s_2)$ and $u_1$ has a different
labelling. There is only one transition $\P(s_1,u_1)=1$.  Obviously
$s_1\not\wsrel s_2$.  Dropping Condition~\ref{ws:reachability} would mean
that $s_1\wsrel s_2$.  We illustrate the use of fragments of states in the
following example:

\begin{figure}[tbp]
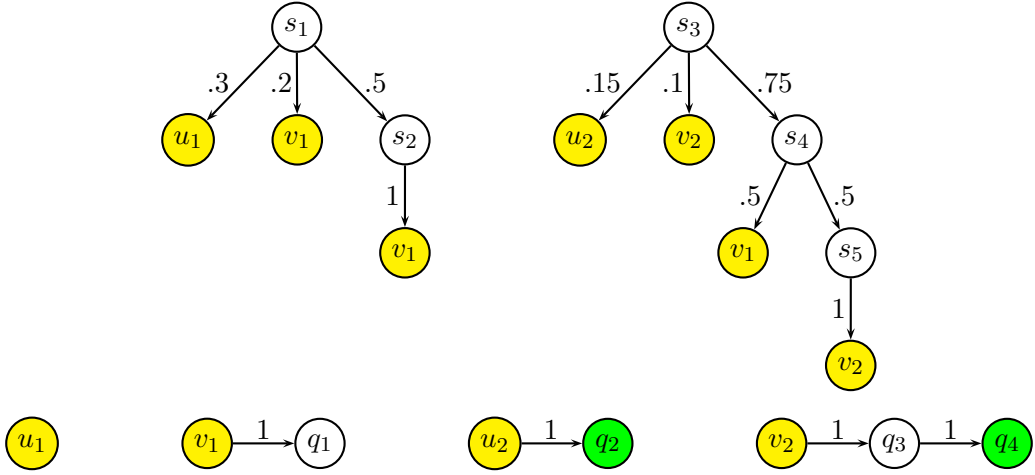

\begin{center}
\psset{levelsep=1.5cm,arrows=->,labelsep=1pt,treemode=B}
\pstree{\Tcircle{$s_1$}}{%
  \Tcircle[fillstyle=solid,fillcolor=yellow]{$u_1$}  \tlput{$.3$}
  \Tcircle[fillstyle=solid,fillcolor=yellow]{$v_1$}  \tlput{$.2$}
  \pstree{\Tcircle{$s_2$} \trput{$.5$}}{%
    \Tcircle[fillstyle=solid,fillcolor=yellow]{$v_1$}  \tlput{$1$}
  }
}\hspace{1.5cm}
\pstree{\Tcircle{$s_3$}}{%
  \Tcircle[fillstyle=solid,fillcolor=yellow]{$u_2$}  \tlput{$.15$}
  \Tcircle[fillstyle=solid,fillcolor=yellow]{$v_2$}  \tlput{$.1$}
\pstree{\Tcircle{$s_4$}  \trput{$.75$}}{%
  \Tcircle[fillstyle=solid,fillcolor=yellow]{$v_1$}  \tlput{$.5$}
  \pstree{\Tcircle{$s_5$} \trput{$.5$}}{%
    \Tcircle[fillstyle=solid,fillcolor=yellow]{$v_2$}  \tlput{$1$}
  }
}
}

\vspace{.3cm}
\hspace{.3cm}
\psset{levelsep=1.5cm,arrows=->,labelsep=1pt,treemode=R}
\Tcircle[fillstyle=solid,fillcolor=yellow]{$u_1$}  
\hspace{1.5cm}
\pstree{\Tcircle[fillstyle=solid,fillcolor=yellow]{$v_1$}}{%
  \Tcircle{$q_1$}  \taput{$1$}
}\hspace{1.5cm}
\pstree{\Tcircle[fillstyle=solid,fillcolor=yellow]{$u_2$}}{%
  \Tcircle[fillstyle=solid,fillcolor=BlueGreen]{$q_2$}  \taput{$1$}
}\hspace{1.5cm}
\pstree{\Tcircle[fillstyle=solid,fillcolor=yellow]{$v_2$}}{%
  \pstree{\Tcircle{$q_3$} \taput{$1$}}{%
    \Tcircle[fillstyle=solid,fillcolor=BlueGreen]{$q_4$}  \taput{$1$}
  }
}
\end{center}
  \caption{A DTMC where splitting states is necessary to establish the weak
  simulation. In the model some states are drawn more than once.}
  \label{fig:split_dtmc}
\end{figure}
\begin{exa}
\label{exa:split_dtmc}
Consider the DTMC depicted in Figure~\ref{fig:split_dtmc}.  For states
$u_1,u_2,v_1,v_2$, obviously the following pairs $(u_1,u_2),(u_1,v_2),
(v_1,v_2)$ are in the weak simulation relation.  The state $u_2$ cannot
weakly simulate $v_1$.  Since $v_2$ weakly simulates $v_1$, it holds that
$s_2
\wsrel s_5$. Similarly, from $u_1\wsrel v_1$ we can easily show that $s_1\wsrel s_4$.  We
observe also that $s_2\not\wsrel s_3$: $K_1>0$ and $K_2>0$ since both $s_2$
and $s_3$ have yellow (grey) successor states, but the required function
$\Delta$ cannot be established since $u_2$ cannot weakly simulate any
successor state of $s_2$ (which is $v_1$). Thus $s_2\not\wsrel s_3$.

Without considering fragments of states, we show that a weak simulation
between $s_1$ and $s_3$ cannot be established. Since $s_2\not\wsrel s_3$,
we must have $U_1=\{u_1,v_1,s_2\}$ and $V_1=\emptyset$. The function
$\delta_1$ is thus defined by $\delta_1(u_1)=\delta_1(v_1)=\delta(s_2)=1$
which implies that $K_1=1$. Now consider the successor states of
$s_3$. Obviously $\delta_2(u_2)=\delta_2(v_2)=1$, which implies that
$u_2,v_2\in U_2$. We consider the following two cases:
\begin{enumerate}[$\bullet$]
\item The case $\delta_2(s_4)=1$. In this case we have that $K_2=1$. 
A function $\Delta$ must be defined satisfying Condition~\ref{ws:wfsums} in
Definition~\ref{def:dtmc_weak_simulation}. Taking $w=s_4$, the following
must hold: $K_2\cdot \Delta(U_1,s_4) = \P(s_3,s_4)\delta_2(s_4)$. As
$K_2=1, \P(s_3,s_4)=0.75$ and $\delta_2(s_4)=1$, it follows that
$\Delta(U_1,s_4)=0.75$. The state $s_2$ is the only successor of $s_1$ that
can be weakly simulated by $s_4$, so $\Delta(s_2,s_4)=0.75$ must
hold. However, the equation $K_1\cdot
\Delta(s_2,U_2) = \P(s_1,s_2)\delta_1(s_2)$ does not hold any more, as
on the left side we have $0.75$ but on the right side we have $0.5$ instead.
\item The case $\delta_2(s_4)=0$. In this case we have still
$K_2>0$. Similar to the previous case it is easy to see that the required
function $\Delta$ cannot be defined: the equation $K_1\cdot
\Delta(s_2,U_2) = \P(s_1,s_2)\delta_1(s_2)$ does not hold since the left
side is $0$ (no states in $U_2$ can weakly simulate $s_2$) but the right
side equals $0.5$.
\end{enumerate}\smallskip

\noindent Thus without splitting, $s_3$ does not weakly simulate
$s_1$. We show it holds that $s_1 \wsrel s_3$. It is sufficient to
show that the relation $R=\{(s_1,s_3),(u_1,u_2),(v_1,v_2)$,
$(q_1,q_3), (s_1,s_4)$, $(u_1,v_1)$, $(v_1,v_1),(q_1,q_1)$,
$(s_2,s_5),(s_2,s_4)\}$ is a weak simulation relation. By the
discussions above, it is easy to verify that every pair except
$(s_1,s_3)$ satisfies the conditions in
Definition~\ref{def:dtmc_weak_simulation}. We show now that the
conditions hold also for the pair $(s_1,s_3)$. The function $\delta_1$
with $\delta_1(u_1)=\delta_1(v_1)=\delta_1(s_2)=1$ is defined as
above, also the sets $U_1=\{u_1,v_1,s_2\}$, $V_1=\emptyset$. The
function $\delta_2$ is defined by: $\delta_2(u_2)=\delta_2(v_2)=1$ and
$\delta_2(s_4)=\frac{1}{3}$, which implies that $U_2=\{u_2,v_2,s_4\}$
and $V_2=\{s_4\}$. Thus, we have $K_1=1$ and $K_2=0.5$. Since
$s_1\wsrel s_4$, Condition~\ref{ws:stuttersimulates} holds trivially
as $(s_1,s_4)\in R$. The reachability condition also holds
trivially. To show that Condition~\ref{ws:weightfunction} holds, we
define the function $\Delta$ as follows: $\Delta(u_1,u_2)=0.3$,
$\Delta(v_1,v_2)=0.2$ and $\Delta(s_2,s_4)=0.5$. We show that
$K_2\cdot \Delta(U_1,w) = \P(s_3,w)\delta_2(w)$ holds for all $w\in
S$. It holds that $K_2=0.5$. First observe that for $w\not\in U_2$
both sides of the equation equal $0$. Let first $w=u_2$ for which we
have that $\P(s_3,u_2)\delta_2(u_2)=0.15$. Since
$\Delta(U_1,u_2)=0.3$, also the left side equals $0.15$. The case
$w=v_2$ can be shown in a similar way. Now consider $w=s_4$. Observe
that $\Delta(U_1,s_4)=0.5$ thus the left side equals $0.25$. The right
side equals $\P(s_3,s_4)\delta_2(s_4)=0.75\cdot \frac{1}{3}=0.25$
thus the equation holds.  The equation $K_1\cdot \Delta(w,U_2) =
\P(s_1,w)\delta_1(w)$ can be shown in a similar way. Thus $\Delta$
satisfies all the conditions, which implies that $s_1\wsrel s_3$.
\end{exa}

Weak simulation for CTMCs is defined as follows.

\begin{defi}[\cite{BKHW05,BaierKHH02}]
\label{def:ctmc_weak_simulation}
Let $\ctmc$ be a CTMC and let $R\subseteq S\times S$. The relation $R$ is a
weak simulation on $\C$ iff for $s_1\mathrel{R}s_2$: $L(s_1)=L(s_2)$ and
there exist functions $\delta_i:S\rightarrow[0,1]$ (for $i=1,2$) satisfying
Equation~\ref{eq:weaksimUV} and Conditions~\ref{ws:stuttersimulates} and
\ref{ws:weightfunction} of Definition~\ref{def:dtmc_weak_simulation} and
the \emph{rate condition}:
\begin{enumerate}[(1')]
\item[(3')]\hfil $K_1\cdot \R(s_1,S)\le K_2\cdot\R(s_2,S)$\hfill
\end{enumerate}
  We say that $s_2$ weakly simulates $s_1$ in $\C$, denoted
  $s_1\wsrel_\M s_2$, iff there exists a weak simulation $R$
  on $\C$ such that $s_1\mathrel{R}s_2$.
\end{defi}
In this definition, the rate condition~$\ref{ws:reachability}'$ strengthens
the reachability condition of the preceding definition.  If $U_1 \not=
\emptyset$, we have that $K_1 > 0$; the rate condition then requires that
$K_2 > 0$, which implies $U_2 \not= \emptyset$. For both DTMCs and CTMCs,
the defined weak simulation $\wsrel$ is a preorder~\cite{BKHW05}, and is
the coarsest weak simulation relation for $\M$.

\subsubsection{Simulation \upto R} \label{subs:simulation_upto_R}
For an arbitrary relation $R$ on the state space $S$ of an FPS with
$s_1\mathrel{R}s_2$, we say that $s_2$ simulates $s_1$ strongly up~to
$R$, denoted $s_1\simrel_R s_2$, if $L(s_1)=L(s_2)$ and $\Pmu{s_1}
\weight_R\Pmu{s_2}$. Otherwise we write $s_1\not\simrel_R s_2$.  Since
only the first step is considered for $\simrel_R$, $s_1\simrel_R s_2$
does not imply $s_1\simrel_\M s_2$ unless $R$ is a strong simulation.
By definition, $R$ is a strong simulation if and only if for all $s_1
\mathrel{R} s_2$ it holds that $s_1 \simrel_R s_2$.  Likewise, we say
that $s_2$ simulates $s_1$ weakly \upto $R$, denoted by $s_1 \wsrel_R
s_2$, if there are functions $\delta_i$ and $U_i,V_i,\Delta$ as
required by Definition~\ref{def:dtmc_weak_simulation} for this pair of
states.  Otherwise, we write $s_1\not \wsrel_R s_2$.  Similar to
strong simulation \upto $R$, $s_1\wsrel_Rs_2$ does not imply
$s_1\wsrel_\D s_2$, since no conditions are imposed on pairs in $R$
different from $(s_1,s_2)$.  Again, $R$ is a weak simulation if and
only if for all $s_1 \mathrel{R} s_2$ it holds that $s_1 \wsrel_R
s_2$.  These conventions extend to DTMCs, CTMCs, PAs and CPAs in an
obvious way. For PAs and CPAs, strong probabilistic simulation up to
$R$, denoted by $\simrel^p_R$, is also defined analogously.

\begin{figure}[tbp]
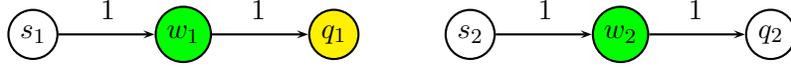

\begin{center}
\psset{levelsep=2cm,arrows=->,treemode=R}
\pstree{\Tcircle{$s_1$}}{%
  \pstree{\Tcircle[fillstyle=solid,fillcolor=BlueGreen]{$w_1$} \taput{$1$} }{%
    \Tcircle[fillstyle=solid,fillcolor=yellow]{$q_1$} \taput{$1$}
  }
}\hspace{1cm}
\pstree{\Tcircle{$s_2$}}{%
  \pstree{\Tcircle[fillstyle=solid,fillcolor=BlueGreen]{$w_2$} \taput{$1$} }{%
    \Tcircle[fillstyle=solid]{$q_2$} \taput{$1$}
  }
}
\end{center}
\caption{A simple FPS for illustrating the simulation up to $R$.}
\label{fig:fps_upto_R}
\end{figure}
\begin{exa}
  Consider the FPS in Figure~\ref{fig:fps_upto_R}. Let
  $R=\{(s_1,s_2),(w_1,w_2)\}$. Since $L(q_1)\neq L(q_2)$ we have that
  $w_1\not\simrel w_2$. Thus, $R$ is not a strong simulation. However,
  $s_1\simrel_R s_2$, as the weight function is given by
  $\Delta(w_1,w_2)=1$. Let $R'=\{(s_1,s_2)\}$, then, $s_1\not\simrel_{R'}
  s_2$.
\end{exa}

\section{Maximum Flow Problems}
\label{sec:flow}
\noindent Before introducing algorithms to decide the simulation preorder, we
briefly recall the preflow algorithm \cite{GoldbergT88} for finding the
maximum flow over the network $\N=(V,E,\capacity)$ where $V$ is a finite set
of vertices, $E\subseteq V\times V$ is a set of edges, and
$\capacity: E\to \Real_{>0} \cup \{ \infty \}$ is the capacity function.
$V$ contains a distinguished \emph{source} vertex $\source$ and a distinguished
\emph{sink} vertex $\sink$.   We extend the capacity function to all
vertex pairs: $\capacity(v,w)=0$ if $(v,w)\not \in E$. A \emph{flow} $f$ on $\N$ is
a function $f:V\times V\to\Real$ that satisfies:
\begin{enumerate}[(1)]
\item   \label{mfl:capacity}
        $f(v,w)\le \capacity(v,w)$ for all $(v,w)\in V\times V$
        \hfill \emph{capacity constraints}
\item   \label{mfl:antisymmetry}
        $f(v,w)=-f(w,v)$ for all $(v,w)\in V\times V$
        \hfill \emph{antisymmetry constraint}
\item   \label{mfl:conservation}
        $f(V,v)=0$ at vertices $v\in V\setminus \{\source,\sink\}$
        \hfill \emph{conservation rule}
\end{enumerate}
The value of a flow function $f$ is given by $f(\source,V)$. A
\emph{maximum flow} is a flow of maximum value.
A
\emph{preflow} is a function $f:V\times V\to\Real$ satisfying
Conditions~\ref{mfl:capacity} and \ref{mfl:antisymmetry} above, and the
relaxation of Condition~\ref{mfl:conservation}:
\begin{enumerate}[(1')]
\item[$(\ref{mfl:conservation}')$] $f(V,v)\ge 0$ for all
$v\in V\setminus \{\source\}$.
\end{enumerate}
The \emph{excess} $e(v)$ of a vertex
$v$ is defined by $f(V,v)$. A vertex $v\not\in\{\source,\sink\}$ is
called \emph{active} if $e(v)>0$. Observe
that if in a preflow function no vertex $v$ is active for $v\in V\setminus
\{\source,\sink\}$, it is then also a flow function.  A pair
$(v,w)$ is a \emph{residual edge} of $f$ if $f(v,w)<\capacity(v,w)$. The set
of residual edges \wrt $f$ is denoted by $E_f$. The \emph{residual
  capacity} $\capacity_f(v,w)$ of the residual edge $(v,w)$ is defined by
$\capacity(v,w)-f(v,w)$. If $(v,w)$ is not a residual edge, it is called
\emph{saturated}.  A \emph{valid distance function} (called
\emph{valid labelling} in \cite{GoldbergT88}) $d$ is a function $V\to
\mathbb{N}\cup\{\infty\}$ satisfying: $d(\source)=\size{V}$, $d(\sink)=0$
and $d(v)\le d(w)+1$ for every residual edge $(v,w)$. A residual edge
$(v,w)$ is \emph{admissible} if $d(v)=d(w)+1$.

Related to maximum flows are minimum cuts.  A \emph{cut} of a network
$\N=(V,E,\capacity)$ is a partition of $V$ into two disjoint sets $(X, X')$ such
that $\source \in X$ and $\sink \in X'$.  The \emph{capacity} of $(X,X')$
is the sum of all capacities of edges from $X$ to $X'$, i.\,e.\@, $\sum_{v
\in X, w \in X'} \capacity(v,w)$.  A \emph{minimum cut} is a cut with
minimal capacity.  The \emph{Maximum Flow Minimum Cut
Theorem}~\cite{AMO93} states that the capacity of a minimum
cut is equal to the value of a maximum flow.

\paragraph{The Preflow Algorithm.}
We initialise the preflow $f$ by: $f(v,w)=\capacity(v,w)$ if
$v=\source$ and $0$ otherwise. The distance function $d$ is
initialised by: $d(v)=\size{V}$ if $v=\source$ and $0$ otherwise. The
preflow algorithm preserves the validity of the preflow $f$ and the distance
function $d$.  If there is an active vertex $v$ such that the residual
edge $(v,w)$ is admissible, we \emph{push}
$\delta:=\min\{e(v),\capacity_f(v,w)\}$ amount of flow from $v$ toward the sink along
the admissible edge $(v,w)$ by increasing $f(v,w)$ (and decreasing
$f(w,v)$) by $\delta$.  The excesses of $v$ and $w$ are then modified
accordingly by: $e(v)=e(v)-\delta$ and $e(w)=e(w)+\delta$.  If $v$ is
active but there are no admissible edges leaving it, one may
\emph{relabel} $v$ by letting $d(v) := \min\{d(w)+1\mid (v,w)\in
E_f\}$.  Pushing and relabelling are repeated until there are no
active vertices left.  The algorithm terminates if no such operations
apply. The resulting final preflow $f$ is a maximum flow.

\paragraph{Feasible Flow Problem.}
Let $A \subseteq E$ be a subset of edges of the network
$\N=(V,E,\capacity)$, and define the lower bound function $l: A \to
\Real_{>0}$ which satisfies $l(e) \le \capacity(e)$ for all $e \in A$. We
address the \emph{feasible flow} problem which consists of finding a flow
function $f$ satisfying the condition: $f(e) \ge l(e)$ for all $e\in A$. We
briefly show that this problem can be reduced to the maximum flow
problem~\cite{AMO93}.

We can replace a minimum flow requirement on
edge $v \to w$ by turning $v$ into a demanding vertex (\ie, a vertex
that consumes part of its inflow) and turning $w$ into a supplying
vertex (\ie, a vertex that creates some outflow \emph{ex nihilo}). The
capacity of edge $v \to w$ is then reduced accordingly.

Now, we are going to look for a flow-like function for the updated network. The
function should satisfy the capacity constraints, and the difference
between outflow and inflow in each vertex corresponds to its supply or
demand, except for $\source$ and $\sink$. To remove that last
exception, we add an edge from $\sink$ to $\source$ with capacity
$\infty$.

We then apply another transformation to the updated network so that we
can apply the maximum flow algorithm. We add new source and target
vertices $\source'$ and $\sink'$. For each supplying vertex $s$, we
add an edge $\source' \to s$ with the same capacity as the supply of
the vertex. For each demanding vertex $d$, we add an edge $d \to
\sink'$ with the same capacity as the demand of the vertex.
In~\cite{AMO93} it is shown that the original network has a feasible flow
if and only if the transformed network has a flow $h$ that saturates all
edges from $\source'$ and all edges to $\sink'$.  The flow
$h$ necessarily is a maximum flow, and if there is an $h$, each maximum
flow satisfies the requirement; therefore it can be found by the maximum
flow algorithm.  An example will be given in Example~\ref{exa:feasible} in
Section~\ref{sec:weak}.

\section{Algorithms for Deciding Strong Simulations}
\label{sec:strong}
\noindent We first recall  the basic algorithm to compute the largest strong
simulation relation $\simrel$ in Subsection~\ref{sec:strong_basic}.
Then, we refine this algorithm to deal with strong simulation on
Markov chains in Subsection~\ref{sec:strong_fps}, and extend it to
deal with probabilistic automata in Subsection~\ref{sec:pp}. In
Subsection~\ref{sec:sps} we present an algorithm for deciding
strong probabilistic simulation for probabilistic automata.

\subsection{Basic Algorithm to Decide Strong Simulation}
\label{sec:strong_basic}

The algorithm in~\cite{BEMC00}, copied as $\prog{SimRel}_s$ in
Algorithm~\ref{fig:simfps}, takes as a parameter a model, which, for now,
is an FPS $\D$.  The subscript `$s$' stands for strong simulation; a very
similar algorithm, namely $\prog{SimRel}_w$, will be used for weak
simulation later.  To calculate the strong simulation relation for $\D$,
the algorithm starts with the initial relation $R_1=\{(s_1,s_2)\in S\times
S\mid L(s_1)=L(s_2)\}$ which is coarser than $\simrel_\M$. In iteration
$i$, it generates $R_{i+1}$ from $R_i$ by deleting each pair $(s_1,s_2)$
from $R_i$ if $s_2$ cannot strongly simulate $s_1$ \upto $R_i$, \ie,
$s_1\not\simrel_{R_i} s_2$. This proceeds until there is no such pair left,
\ie, $R_{i+1}=R_i$.  Invariantly throughout the loop it holds that $R_i$ is
coarser than $\simrel_\M$ (i.\,e.\@, $\simrel_\M$ is a sub-relation of
$R_i$).  We obtain the strong simulation preorder $\mathord{\simrel_\M} =
R_i$, once the algorithm terminates.
 
\begin{algorithm}[tbp]
  \caption{Basic algorithm to decide strong simulation.}
  \label{fig:simfps}
    \Procname{$\prog{SimRel}_s(\D)$}
  \begin{algorithmic}[1]
        \STATE  $R_1 \gets \{(s_1,s_2)\in S\times S\mid L(s_1)=L(s_2)\}$ and $i \gets 0$
        \REPEAT
                \STATE  $i \gets i + 1$
                \STATE  $R_{i+1} \gets \emptyset$
                        \label{algfps:until.start.body}
                \FORALL{$(s_1,s_2)\in R_i$}
                        \IF{$s_1\simrel_{R_i} s_2$}   \label{algsim:checksimrel}
                                \STATE  $R_{i+1}\gets R_{i+1}\cup \{(s_1,s_2)\}$
                        \ENDIF
                \ENDFOR \label{algfps:until.end.body}
        \UNTIL{$R_{i+1}=R_i$}
        \STATE  \textbf{return} $R_i$
  \end{algorithmic}
\end{algorithm}

The decisive part of the algorithm is the check in
Line~\ref{algsim:checksimrel}, \ie, whether $s_1\simrel_{R_i} s_2$.
This can be answered via solving a maximum flow problem on a
particular network $\N(\Pmu{s_1}$, $\Pmu{s_2},R_i)$ constructed from
$\Pmu{s_1}$, $\Pmu{s_2}$ and $R_i$.  This network is the relevant part
of a graph containing two copies $t\in S_\aux$ and $\overline{t}\in
\overline{S_\aux}$ of each state where $\overline{S_\aux}=\{\overline
t\mid t\in S_\aux\}$ as follows: Let $\source$ (the source) and
$\sink$ (the sink) be two additional vertices not contained in
$S_\aux\cup\overline{S_\aux}$.  For $\mu,\mu'\in\distrs(S)$, and a
relation $R\subseteq S\times S$ we define the network $\N(\mu, \mu',
R)= (V,E,\capacity)$ with the set of vertices $
V=\{\source,\sink\}\cup \support_\aux(\mu) \cup
\overline{\support_\aux(\mu')} $ and the set of edges $E$ is defined
by $ E=\{(s,\overline t)\mid (s,t)\in \mathrel{R_\aux} \}\cup
\{(\source,s) \} \cup \{(\overline t, \sink) \} $ where $ s \in
\support_\aux(\mu)$ and $t \in \support_\aux(\mu')$.  Recall the
relation $R_\aux$ is defined by $R \cup \{(\aux,s)\mid s \in
S_\aux\}$.  The capacity function $\capacity$ is defined as follows:
$\capacity(\source,s)=\mu(s)$ for all $s\in \support_\aux(\mu)$,
$\capacity(\overline t,\sink)=\mu'(t)$ for all $t\in
\support_\aux(\mu')$, and $\capacity(s,\overline t)=\infty$ for all
other $(s,\overline{t})\in E$.  This network is a bipartite network,
since the vertices can be partitioned into two subsets
$V_1:=\support_\aux(\mu) \cup\{\sink\}$ and
$V_2:=\overline{\support_\aux(\mu')}\cup\{\source\}$ such that all
edges have one endpoint in $V_1$ and another in $V_2$.  Later, we will
use two variations of this network: For $\gamma\in\Real_{>0}$, we let
$\N(\mu, \gamma\mu', R)$ denote the network obtained from $\N(\mu,
\mu', R)$ by setting the capacities to the sink $\sink$ to:
$\capacity(\overline t,\sink)=\gamma\mu'(t)$.  For two states
$s_1,s_2$ of an FPS or a CTMC, we let $\N(s_1,s_2,R)$ denote the
network $\N(\Pmu{s_1},\Pmu{s_2},R)$.

The following lemma expresses the crucial relationship between maximum
flows and weight functions on which the algorithm is based. It is a
direct extension of~\cite[Lemma~5.1]{BEMC00}:
 
\begin{lem}
\label{lem:weight_equivalent}
Let $S$ be a finite set of states and $R$
be a relation on $S$. Let $\mu,\mu'\in \dist(S)$. Then, $\mu\weight_R
\mu'$ iff the maximum flow of the network $\N(\mu, \mu',R)$ has value $1$.
\end{lem}
\begin{proof}
As we introduced the auxiliary state $\aux$, $\mu$ and $\mu'$ are
stochastic distributions in $\distrs(S_\aux)$.  The rest of the proof
follows directly from \cite[Lemma~5.1]{BEMC00}.
\end{proof}

Thus we can decide $s_1\simrel_{R_i} s_2$ by computing the maximum flow in
$\N(s_1,s_2,R_i)$ and then check whether it has value $1$.  We recall the
correctness and complexity of $\prog{SimRel}_s$ which will also be used
later.
\begin{thm}[\cite{BEMC00}]
\label{thm:correctness_basic}
If $\prog{SimRel}_s(\D)$ terminates, the returned relation equals
$\simrel_\M$. Moreover, $\prog{SimRel}_s(\D)$ runs in time $\O(n^7/\log
n)$ and in space $\O(n^2)$.
\end{thm}
\begin{proof}
  First we show that after the last iteration (say iteration $k$), it
  holds that $\simrel$ is coarser than $R_k$: It holds that $R_{k+1}=R_k$,
  thus for all $(s_1,s_2)\in R_k$, we have that $s_1\simrel_{R_k} s_2$.  As
  for all $(s_1,s_2) \in R_k \subseteq R_1$, we have $L(s_1) = L(s_2)$,
  $R_k$ is a strong simulation relation by
  Definition~\ref{def:strongsimulation-fps}, thus $\simrel$ is coarser than
  $R_k$.

  Now we show by induction
  that the loop of the algorithm invariantly ensures that
  $R_i$ is coarser than $\simrel$.
  Assume $i=1$. By definition of strong simulation,
  $s_1\simrel s_2$ implies $L(s_1)=L(s_2)$.  Thus,
  the initial relation $R_1$ is coarser than the simulation
  relation $\simrel$.
  Now assume that $R_i$ is coarser than $\simrel$ for some $1 \leq i < k$;
  we will show that also $R_{i+1}$ is coarser than $\simrel$.
  Pick a pair $(s_1,s_2) \in \mathord{\simrel}$ arbitrarily.
  By Definition~\ref{def:strongsimulation-fps}, $\Pmu{s_1}
\weight_\simrel \Pmu{s_2}$,
  so there exists a weight function for $(\Pmu{s_1},\Pmu{s_2})$ \wrt
  $\simrel$.  Inspection of Definition~\ref{def:weight} shows that the same
  function is also a weight function \wrt any set coarser than
  $\simrel$.  As $R_i$ is coarser than $\simrel$ by induction hypothesis,
  we conclude that $\Pmu{s_1} \weight_{R_i} \Pmu{s_2}$, and from
  Subsection~\ref{subs:simulation_upto_R}, $s_1 \simrel_{R_i} s_2$.  This
  implies that $(s_1,s_2) \in R_{i+1}$ by line~\ref{algsim:checksimrel} for
  all $s_1 \simrel s_2$. Therefore, $R_{i+1}$ is coarser than $\simrel$ for
  all $i=1\ldots,k$.
  
  Now we show the complexity. For one network
  $\N(s_1,s_2,R_i)=(V,E,\capacity)$, the sizes of the vertices
  $\size{V}$ and edges $\size{E}$ are bounded by $2n+4$ and $(n+1)^2 +
  2n$, respectively. The number of edges meets the worst case bound
  $\O(n^2)$. To the best of our knowledge, the best complexity of the
  flow computation for the network $G$ is $\O(\size{V}^3/\log
  \size{V})=\O(n^3/\log n)$~\cite{CheriyanHM90,Gol98}.  In the
  algorithm $\prog{SimRel}_s$, only one pair, in the worst case, is
  removed from $R_i$ in iteration $i$, which indicates that the test
  whether $s_1\simrel_{R_i} s_2$ is called $\size{R_1}$ times,
  $\size{R_1}-1$ times and so on. Altogether it is bounded by
  $\sum_{i=1}^{\size{R_1}}i \le \sum_{i=1}^{n^2}i \in \O(n^4) $.
  Hence, the overall time complexity amounts to $\mathcal{O}(n^7/\log
  n)$.  The space complexity is $\O(n^2)$ because of the
  representation of the transitions in $\N(s_1,s_2,R_i)$.
\end{proof} 

\subsection{An Improved Algorithm for FPSs}  
\label{sec:strong_fps}

We first analyse the behaviour of $\prog{SimRel}_s$ in more detail.  For
this, we consider an arbitrary pair $(s_1,s_2)$, and assume that
$(s_1,s_2)$ stays in relation $R_1, \ldots, R_k$ throughout the iterations
$i=1, \ldots, k$, until the pair is either found not to satisfy
$s_1\simrel_{R_k} s_2$ or the algorithm terminates with a fix-point after
iteration $k$.  Then altogether the maximum flow algorithms are run
$k$-times for this pair.  However, the networks $\N(s_1,s_2,R_i)$
constructed in successive iterations are very similar, and may often be
identical across iterations: They differ from iteration to iteration only
by deletion of some edges induced by the successive cleanup of $R_i$. For
our particular pair $(s_1,s_2)$ the network might not change at all in some
iterations, because the deletions from $R_i$ do not affect their direct
successors.  We are going to exploit this observation by an algorithm that
reuses the already computed maximum flow, in a way that whatever happens is
good: If no changes occur from $\N(s_1,s_2,R_{i-1})$ to $\N(s_1,s_2,R_i)$,
then the maximum flow is the same as the one in the previous iteration. If
changes do occur, the preflow algorithm can be applied to get the new
maximum flow very fast, using the maximum flow and distance function
constructed in the previous iteration as a starting point.

To understand the algorithm, we look at the network
$\N(s_1,s_2,R_1)$.  Let $D_1,\ldots,D_k$ be pairwise disjoint
subsets of $R_1$, which correspond to the pairs deleted from
$R_1$ in iteration $i$, so $R_{i+1} = R_i
\setminus D_i$ for $1 \leq i \leq k$.  Let $f_i^{(s_1,s_2)}$ denote the maximum
flow of the network $\N(s_1,s_2,R_i)$ for $1 \leq i \leq k$.  We
sometimes omit the superscript $(s_1,s_2)$ in the parameters if the pair
$(s_1,s_2)$ is clear from the context. We address the problem of checking
$\size{f_i}=1$ for all $i = 1,
\ldots, k$.  Our algorithm \emph{sequence of maximum flows}
$\prog{Smf}(i, \N(s_1,s_2,R_{i-1}), f_{i-1}, d_{i-1}, D_{i-1})$ is shown as
Algorithm~\ref{fig:smf}. It executes iteration $i$ of a parametric flow
algorithm, where $\N(s_1,s_2,R_{i-1})$ is the network for $(s_1,s_2)$ and
$f_{i-1}$ and $d_{i-1}$ are the flow and the distance function resulting
from the previous iteration $i-1$; and $D_{i-1}$ is a set of edges that have
to be deleted from $\N(s_1,s_2,R_{i-1})$ to get the current network.  The
algorithm returns a tuple, in which the first component is a boolean that
tells whether $\size{f_i} = 1$; it also returns the new network
$\N(s_1,s_2,R_i)$, flow  $f_i$ and distance function $d_i$ to be
reused in the next iteration.  $\prog{Smf}$ is inspired by the
\emph{parametric maximum algorithm} in \cite{GalloGT89}.  A
variant of $\prog{Smf}$ is used in the first iteration, shown in
lines~\ref{algsmfinit.start}--\ref{algsmfinit.end}.

\begin{algorithm}[tbp]
  \caption{Algorithm for a sequence of maximum flows.}
\label{fig:smf}
    \Procname{$\prog{Smf}(i, \N(s_1,s_2,R_{i-1}), f_{i-1}, d_{i-1}, D_{i-1})$}
  \begin{algorithmic}[1]
        \STATE  $\N(s_1,s_2,R_i) \gets \N(s_1,s_2,R_{i-1}\setminus
  D_{i-1})$
                and $f_i \gets f_{i-1}$
                and $d_i \gets d_{i-1}$
                \label{algsmf.reinitall.start}
        \FORALL{$(u_1,u_2) \in D_{i-1}$}   \label{algsmf.reinit.start}
                \STATE  $f_i(\overline{u_2},\sink) \gets f_i(\overline{u_2},\sink) - f_i(u_1,\overline{u_2})$
                \STATE  $f_i(u_1,\overline{u_2}) \gets 0$
        \ENDFOR   \label{algsmf.reinit.end} \label{algsmf.reinitall.end}
        \STATE  Apply the preflow algorithm to calculate the maximum flow
                for $\N(s_1,s_2,R_i)$, \hfill\\ 
                but initialise the preflow to $f_i$
                and the distance function to $d_i$.
                \label{algsmf.callpreflow}
        \STATE  \textbf{return} $(\size{f_i} = 1, \N(s_1,s_2,R_i), f_i, d_i)$
                        \label{algsmf.return}
  \end{algorithmic}

\bigskip

    \Procname{$\prog{Smf}_\init(i, s_1, s_2, R_i)$}
  \begin{algorithmic}[1]
        \addtocounter{ALC@line}{10}
        \STATE  Initialise the network $\N(s_1,s_2,R_i)$.
                \label{algsmfinit.start}
        \STATE  Apply the preflow algorithm to calculate the maximum flow for $\N(s_1,s_2,R_i)$.        \label{algsmfinit.callpreflow}
        \STATE  \textbf{return} $(\size{f_i} = 1, \N(s_1,s_2,R_i), f_i, d_i)$
                \label{algsmfinit.end}
  \end{algorithmic}
\end{algorithm}

\begin{algorithm}[tbp]
  \caption{Improved algorithm for deciding strong simulation for FPSs.}
  \label{fig:simfpssmf}
    \Procname{$\prog{SimRel}^\mathrm{FPS}_s(\D)$}
  \begin{algorithmic}[1]
        \STATE  $R_1 \gets \{(s_1,s_2)\in S\times S\mid L(s_1)=L(s_2)\}$ and $i \gets 1$        \label{algsrf:initR}
        \STATE  $R_2 \gets \emptyset$
                \label{algsrf:iteration1.start}
        \FORALL{$(s_1,s_2)\in R_1$}
                \STATE  $\listener_{(s_1,s_2)}$ $\gets \{(u_1,u_2)\mid u_1\in
                        \pre(s_1)\wedge u_2\in \pre(s_2) \wedge L(u_1)=L(u_2)\}$        \label{algsrf:initlistener}
                \STATE  $(\mathit{match}, \N(s_1,s_2,R_1), f_1^{(s_1,s_2)}, d_1^{(s_1,s_2)}) \gets
                        \prog{Smf}_\init(1, s_1, s_2, R_1)$ 
                \IF{$\mathit{match}$}
                        \STATE  $R_2 \gets R_2 \cup \{(s_1,s_2)\}$
                \ENDIF
        \ENDFOR \label{algsrf:iteration1.end}
        \WHILE{$R_{i+1} \not= R_i$}   \label{algsrf:until.start}
                \STATE  $i \gets i + 1$
                \STATE  $R_{i+1} \gets \emptyset$ and $D_{i-1} \gets R_{i-1} \setminus R_i$     \label{algsrf:until.start.body}
                \FORALL{$(s_1,s_2)\in R_i$}
                        \label{algsrf:setDs1s2.start}
                        \STATE  $D_{i-1}^{(s_1,s_2)} \gets \emptyset$
                \ENDFOR
                \FORALL{$(s_1,s_2)\in D_{i-1}, \quad (u_1,u_2)\in \listener_{(s_1,s_2)} \cap R_{i-1}$}
                        \STATE  $D_{i-1}^{(u_1,u_2)} \gets D_{i-1}^{(u_1,u_2)} \cup \{(s_1,s_2)\}$
                \ENDFOR   \label{algsrf:setDs1s2.end}
                \FORALL{$(s_1,s_2)\in R_i$}
                        \STATE $(\mathit{match}, \N(s_1,s_2,R_i),
                                f_i^{(s_1,s_2)}, d_i^{(s_1,s_2)})$ \\
                                $\qquad \qquad \qquad \gets \prog{Smf}(i,
                                \N(s_1,s_2,R_{i-1}), f_{i-1}^{(s_1,s_2)}, 
                                d_{i-1}^{(s_1,s_2)},
                                D_{i-1}^{(s_1,s_2)})$ 
                                \label{algsrf:callsmf}
                        \IF{$\mathit{match}$}   \label{algsrf:testmatch}
                                \STATE  $R_{i+1}\gets
                                R_{i+1}\cup \{(s_1,s_2)\}$ \label{algsrf:insertnew}
                        \ENDIF
                \ENDFOR \label{algsrf:until.end.body}
        \ENDWHILE       \label{algsrf:until.end}
        \STATE  \textbf{return} $R_i$
  \end{algorithmic}
\end{algorithm}

This algorithm for sequence of maximum flow problems is called in an
improved version of $\prog{SimRel}_s$ shown as
Algorithm~\ref{fig:simfpssmf}.
Lines~\ref{algsrf:iteration1.start}--\ref{algsrf:iteration1.end} contain
the first iteration, very similar to the first iteration of
Algorithm~\ref{fig:simfps}
(lines~\ref{algfps:until.start.body}--\ref{algfps:until.end.body}).  At
line~\ref{algsrf:initlistener} we prepare for later iterations the set
\[\listener_{(s_1,s_2)}=
  \{(u_1,u_2) \mid u_1\in \pre(s_1)\wedge u_2\in \pre(s_2) \wedge
  L(u_1)=L(u_2)\} \quad,
\]
where $\pre(s)=\{t\in S\mid \P(t,s)>0\}$.  This set contains all pairs
$(u_1,u_2)$ such that the network $\N(u_1,u_2,R_1)$ contains the edge
$(s_1,\overline{s_2})$.  Iteration~$i$ (for $i>1$) of the loop
(lines~\ref{algsrf:until.start.body}--\ref{algsrf:until.end.body})
calculates $R_{i+1}$ from $R_i$.  In
lines~\ref{algsrf:setDs1s2.start}--\ref{algsrf:setDs1s2.end}, we collect
edges that should be removed from $\N(u_1,u_2,R_{i-1})$ in the sets
$D_{i-1}^{(u_1,u_2)}$.  At line~\ref{algsrf:callsmf}, the algorithm
$\prog{Smf}$ constructs the maximum flow for parameters using information
from iteration~$i-1$.  It uses the set $D_{i-1}^{(s_1,s_2)}$ to update the
network $\N(s_1,s_2,R_{i-1})$, flow $f_{i-1}$, a distance function
$d_{i-1}$; then it constructs the maximum flow $f_i$ for the network
$\N(s_1,s_2,R_i)$.  If $\prog{Smf}$ returns true, $(s_1,s_2)$ is inserted
into $R_{i+1}$ and survives this iteration (line~\ref{algsrf:insertnew}).

Consider the algorithm $\prog{Smf}$ and assume that $i>1$.  At
lines~\ref{algsmf.reinitall.start}--\ref{algsmf.reinit.end}, we
remove the edges $D_{i-1}$ from the network $\N(s_1,s_2,R_{i-1})$ and
generate the preflow $f_i$ based on the flow $f_{i-1}$, which is the maximum flow of the network
$\N(s_1,s_2,R_{i-1})$, by
\begin{enumerate}[$\bullet$]
\item setting $f_i(u_1,\overline{u_2})=0$ for all deleted edges
  $(u_1,u_2)\in D_{i-1}$, and
\item reducing $f_i(\overline{u_2}, \sink)$ such that the
  preflow $f_i$ becomes consistent with the (relaxed) flow conservation
  rule.
\end{enumerate}
The excess $e(v)$ is increased if there exists $(v,w)\in D_{i-1}$ such that
$f_{i-1}(v,w)>0$, and unchanged otherwise.  Hence, $f_i$ after
line~\ref{algsmf.reinitall.end} is a preflow.  The distance function
$d_{i-1} = d_i$ is still valid for this preflow, since removing the set of
edges $D_{i-1}$ does not introduce new residual edges.  This guarantees
that, at line~\ref{algsmf.callpreflow}, the \emph{preflow algorithm} finds
a maximum flow over the network $\N(s_1,s_2,R_i)$.  In
line~\ref{algsmf.return}, $\prog{Smf}$ returns whether the flow has value 1
together with information to be reused in the next iteration.  (If
$\size{f_k}<1$ at some iteration $k$, then $\size{f_j}<1$ for all
iterations $j\ge k$ because deleting edges does not increase the maximum
flow.  In that case, it would be sufficient to return $\mathbf{false}$.)
We prove the correctness and complexity of the algorithm $\prog{Smf}$:

\begin{lem}
\label{smf_correctness} 
  Let $(s_1,s_2)\in R_1$.  Then, $\prog{Smf}_\init$ returns true iff $s_1
  \simrel_{R_1} s_2$.
  For some $i>1$, let $\N(s_1,s_2,R_{i-1})$, $f_{i-1}$, and $d_{i-1}$ be as
  returned by some earlier call to $\prog{Smf}$ or $\prog{Smf}_\init$.  Let
  $D_{i-1} = (R_{i-1} \setminus R_{i}) \cap (\post(s_1) \times \post(s_2))$
  be the set of edges that will be removed from the network
  $\N(s_1,s_2,R_{i-1})$ during the $(i-1)$th call of $\prog{Smf}$.  Then,
  the $(i-1)$th call of $\prog{Smf}$ returns true iff $s_1 \simrel_{R_i}
  s_2$.
\end{lem}
\begin{proof}
  By Lemma~\ref{lem:weight_equivalent}, $\prog{Smf}_\init$ returns true iff
  $|f_1|=1$, which is equivalent to $s_1 \simrel_{R_1} s_2$.  Let $i>1$. As
  discussed, at the beginning of line~\ref{algsmf.callpreflow}, the
  function $f_{i-1}$ is a flow (thus a preflow) with value 1, and the
  distance function $d_{i-1}$ is a valid distance function.  It follows
  directly from the correctness of the preflow algorithm~\cite{AhujaOST94}
  that after line~\ref{algsmf.callpreflow}, $f_i$ is a maximum flow for
  $\N(s_1,s_2,R_i)$. Thus, $\prog{Smf}$ returns true (i.e. $|f_i|=1$) which
  is equivalent to $s_1 \simrel_{R_i} s_2$.
\end{proof}
\begin{lem}
\label{smf_complexity}
  Consider the pair $(s_1,s_2)$ and assume that $\size{\post(s_1)}
  \le\size{\post(s_2)}$.  All calls to $\prog{Smf}(i,\N(s_1,s_2, \cdot),
  \cdots)$ related to $(s_1,s_2)$ together run in time
  $\O(\size{\post(s_1)} \size{\post(s_2)}^2)$.
\end{lem}
\begin{proof}
  In the bipartite network $\N(s_1,s_2,R_1)$, the set of vertices are
  partitioned into subsets $V_1 = \post(s_1) \cup \{ \sink \}$ and
  $V_2 = \post(s_2) \cup \{ \source \}$ as described in
  Section~\ref{sec:strong_basic}.  Generating the initial network
  (line~\ref{algsmfinit.start}) takes time in $\O(\size{V_1}
  \size{V_2})$. In our sequence of maximum flow problems, the number
  of (nontrivial) iterations, denoted by $k$, is bounded by the number
  of edges, i.\,e., $k \le \size{E} \le \size{V_1} \size{V_2} - 1$.
  We split the work being done by all calls to
  $\prog{Smf}(i,\N(s_1,s_2, \cdot), \ldots)$ together with the initial
  call to the preflow algorithm (line~\ref{algsmfinit.callpreflow} and
  line~\ref{algsmf.callpreflow}) into edge deletions, relabels,
  non-saturating pushes, saturating pushes.  (A non-saturating push
  along an edge $(u,v)$ moves all excess at $u$ to $v$; by such a
  push, the number of active nodes never increases.)
  
  All edge deletions take time proportional to $\sum_{i=1}^k
  \size{D_i}$, which is less than the number of edges in the network.
  Therefore, edge deletions take time $\O(|V_1||V_2|)$.  For all $v\in
  V$, it holds that $d_{i+1}(v)=d_i(v)$, i.e., the labelling function
  at the beginning of iteration $i+1$ is the same as the labelling
  function at the end of iteration $i$. 
  
  We discuss the time for relabelling and saturating
  pushes~\cite{AhujaOST94}. For a bipartite network, the distance of
  the source can be initialised to $d(\source) = 2\size{V_1}$ instead
  of $\size{V}$, and $d(v)$ never grows above $4\size{V_1}$ for all
  $v\in V$.  For $v\in V$, let $I(v)$ denote the set of nodes
  containing $w$ such that either $(v,w)\in E$ or $(w,v)\in E$.
  Intuitively, it represents edges which could be admissible leaving
  $v$.  The time for relabel operations with respect to node $v$ is
  thus $(4\size{V_1})|I(v)|$. Altogether, this gives the time for all
  relabel operations: $\sum_{v\in V}((4\size{V_1})|I(v)|) \in
  \O(\size{V_1}\size{E})$.  Between two consecutive saturating pushes
  on $(v,w)$, the distances $d(v)$ and $d(w)$ must increase by $2$.
  Thus, the number of saturating pushes on edge $(v,w)$ is bounded by
  $4 \size{V_1}$.  Summing over all edges, the work for saturating
  pushes is bounded by $\O(\size{V_1}\size{E})$.
  
  Now we discuss the analysis of the number of non-saturating pushes,
  which is very similar to the proof of Theorem~2.2
  in~\cite{GusfieldT94} where Max-d version of the algorithm is used.
  Assume that in iteration $l\le k$ of $\prog{Smf}$, the last
  relabelling action occurs.  In the Max-d version~\cite{GusfieldT94},
  always the active node with the highest label is selected, and once
  an active node is selected, the excess of this node is pushed until
  it becomes $0$.  This implies that, between any two relabel
  operations, there are at most $n$ active nodes processed (otherwise
  the algorithm terminates and we get the maximum flow). Also observe
  that at each time an active node is selected, at most one
  non-saturating push can occur, which implies that there are at most
  $n$ non-saturating pushes between node label increases. Since
  $d_i(v)$ is bounded by $4|V_1|$, the number of relabels altogether
  is bounded by $\O(|V_1||V|)$.  Thus, the number of non-saturating
  pushes before the iteration $l$ is bounded by $\O(|V_1||V|^2)$.
  Since the distance function does not change after iteration $l$ any
  more, inside any of the iterations $l'\ge l$, there are again at
  most $n-1$ non-saturating pushes.  Hence, the number of
  non-saturating pushes is bounded by $|V_1||V|^2+(k+1-l)(|V|-1) \in
  \O(|V_1||V|^2+k|V|)$. Since $k\le |V_1||V_2|-1$, and $|V|\le
  2|V_2|$, thus, the overall time complexity amounts to
  $\O(|V_1||V_2|^2)= \O(\size{\post(s_1)} \size{\post(s_2)}^2)$ as
  required.
\end{proof}

Now we give the correctness and complexity of the algorithm
$\prog{SimRel}$ for FPSs:

\begin{thm}
\label{thm:fpscorrectness}
If $\prog{SimRel}^\mathrm{FPS}_s(\D)$ terminates, the returned
relation equals $\simrel_\M$. 
\end{thm}
\begin{proof}
  By Lemma~\ref{smf_correctness}, $\prog{Smf}_\init(i, s_1, s_2, R_1)$
  returns true in iteration $i=1$ iff $s_1 \simrel_{R_1} s_2$;
  $\prog{Smf}(i, \N(s_1,s_2,R_{i-1}), \ldots)$ returns true in iteration
  $i>1$ iff $s_1 \simrel_{R_i} s_2$.  The rest of the correctness proof is
  the same as the proof of Theorem~\ref{thm:correctness_basic}.
\end{proof}

\begin{thm}
  The algorithm $\prog{SimRel}^\mathrm{FPS}_s(\D)$ runs in time
  $\O(m^2n)$ and in space $\O(m^2)$.  If the fanout is bounded by a
  constant, it has complexity $\O(n^2)$, both in time and space.
\end{thm}
\begin{proof}
We first show  the space complexity.  In most cases, it is enough to store
information from the previous iteration until the corresponding structure
for the current iteration is calculated.  The size of the set
$\listener_{(s_1,s_2)}$ is bounded by $\size{\pre(s_1)} \size{\pre(s_2)}$
where $\pre(s)=\{t\in S\mid \P(t,s)>0\}$.  Summing over all $(s_1,s_2)$, we
get
$\sum_{s_1\in S}\sum_{s_2\in S}
    \size{\pre(s_1)} \size{\pre(s_2)}
    = m^2
$.
Assume we run iteration $i$.  For every pair $(s_1,s_2)$, we generate
the set $D_{i-1}^{(s_1,s_2)}$ and the network $\N(s_1,s_2,R_i)$
together with $f_i$ and $d_i$.  Obviously, the size of
$D_{i-1}^{(s_1,s_2)}$ is bounded by
$\size{\post(s_1)}\size{\post(s_2)}$.  Summing over all $(s_1,s_2)$,
we get the bound $\O(m^2)$.  The number of edges of the network
$\N(s_1,s_2,R_1)$ (together with $f_i$ and $d_i$) is in
$\O(\size{\post(s_1)} \size{\post(s_2)})$.  Summing over all
$(s_1,s_2)$ yields a memory consumption in $\O(m^2)$ again.  Hence,
the overall space complexity is $\O(m^2)$.

Now we show the time complexity.  We observe that a pair $(s_1,s_2)$
belongs to $D_i$ in at most one iteration. Therefore, the time needed in
lines~\ref{algsrf:setDs1s2.start}--\ref{algsrf:setDs1s2.end} in all
iterations together is bounded by the size of all sets
$\listener_{(s_1,s_2)}$, which is $\O(m^2)$.  We analyse the time needed
for all calls to the algorithm $\prog{Smf}$.  Recall that the fanout $g$
equals $\max_{s\in S}\size{\post(s)}$, and therefore $\size{\post(s_i)}
\leq g$ for $i=1,2$.  By Lemma~\ref{smf_complexity}, the complexity
attributed to the pair $(s_1,s_2)$ is bounded by $\O(g \size{\post(s_1)}
\size{\post(s_2)})$.  Taking the sum over all possible pairs, we get the
bound $gm^2\in \O(m^2n)$.  If $g$ is bounded by a constant, we have $m\le
gn$, and the time complexity is $gm^2\le g^3n^2\in
\O(n^2)$. In this case the space complexity is also $\O(n^2)$.
\end{proof}

\paragraph{Strong Simulation for Markov Chains.}
We now consider DTMCs and CTMCs.  Since each DTMC is a special case of an
FPS the algorithm $\prog{SimRel}^\mathrm{FPS}_s$ applies directly.

Let $\ctmc$ be a CTMC. Recall that $s\simrel_\M s'$ holds if
$s\simrel_{emb(\M)} s'$ in the embedded DTMC, and $s'$ is faster than
$s$. We can ensure the additional rate condition by incorporating it into
the initial relation $R$. More precisely, initially $R$ contains only those
pair $(s,s')$ such that $L(s)=L(s')$, and that the
state $s'$ is faster than $s$, \ie, we replace line~\ref{algsrf:initR} of
the algorithm by
\[
R_1 \gets \{(s_1,s_2)\in S\times S \mid L(s_1)=L(s_2)
\wedge \R(s_1,S)\le \R(s_2,S)\}
\]
to ensure the additional rate condition of
Definition~\ref{def:strongsimulation-ctmc}. In the refinement steps
afterwards, only the weight function conditions need to be checked with
respect to the current relation in the embedded DTMC. Thus, we arrive at an
algorithm for CTMCs with the same time and space complexity as for FPSs.

\begin{exa}
  Consider the CTMC in the left part of Figure~\ref{fig:ctmcsmf} (it has 10
  states).  Consider the pair $(s_1,s_2)\in R_1$.  The network
  $\N(s_1,s_2,R_1)$ is depicted on the right of the figure. Assume that we
  get the maximum flow $f_1$ which sends $\frac{1}{2}$ amount of flow along
  the path $\source, u_2,\overline{u_4},\sink$ and $\frac{1}{2}$ amount of
  flow along $\source,u_1,\overline{u_3},\sink$. Hence, the check for
  $(s_1,s_2)$ is successful in the first iteration. The checks for the
  pairs $(u_1,u_3)$, $(u_1,u_4)$ and $(u_2,u_3)$ are also successful in the
  first iteration. However, the check for the pair $(u_2,u_4)$ fails, as
  the probability to go from $u_4$ to $q_3$ in the embedded DTMC is
  $\frac{2}{5}$, while the probability to go from $u_2$ to $q_1$ in the
  embedded DTMC is $1$.

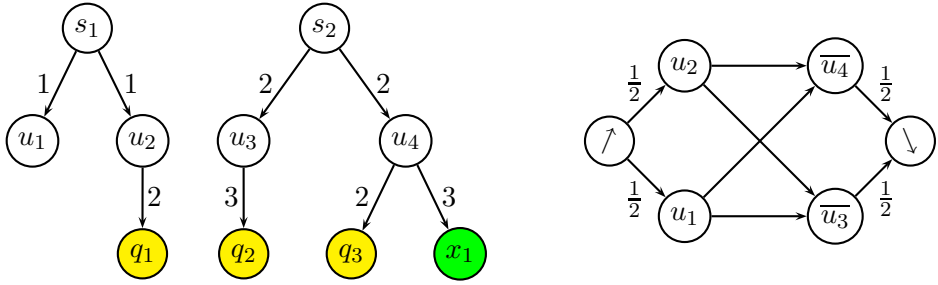
\begin{figure}[tbp]  
\begin{center}
\psset{levelsep=1.5cm,arrows=->,labelsep=1pt}
\pstree{\Tcircle{$s_1$}}{%
  \Tcircle{$u_1$}  \tlput{$1$}
  \pstree{\Tcircle{$u_2$} \trput{$1$} }{%
    \Tcircle[fillstyle=solid,fillcolor=yellow]{$q_1$} \trput{$2$}
  }
}\hspace{.5cm}
\pstree{\Tcircle{$s_2$}}{%
  \pstree{\Tcircle{$u_3$} \tlput{$2$} }{%
    \Tcircle[fillstyle=solid,fillcolor=yellow]{$q_2$}  \tlput{$3$}
  }
  \pstree{\Tcircle{$u_4$} \trput{$2$} }{%
    \Tcircle[fillstyle=solid,fillcolor=yellow]{$q_3$} \tlput{$2$}
    \Tcircle[fillstyle=solid,fillcolor=BlueGreen]{$x_1$} \trput{$3$}
  }
}\hspace{1.5cm}
\begin{pspicture}(0,2.5)(4, 2.5)
  \rput(1,0){\circlenode{u1}{$u_1$}}
  \rput(1,2){\circlenode{u2}{$u_2$}}
  \rput(3,0){\circlenode{u3}{$\overline{u_3}$}}
  \rput(3,2){\circlenode{u4}{$\overline{u_4}$}}
  \rput(0,1){\circlenode{source}{$\source$}}
  \rput(4,1){\circlenode{sink}{$\sink$}}
  \ncline{source}{u1} \nbput{$\frac{1}{2}$}
  \ncline{source}{u2} \naput{$\frac{1}{2}$}
  \ncline{u3}{sink} \nbput{$\frac{1}{2}$}
  \ncline{u4}{sink} \naput{$\frac{1}{2}$}
  \ncline{u1}{u3}
  \ncline{u1}{u4}
  \ncline{u2}{u3}
  \ncline{u2}{u4}
\end{pspicture}
\end{center}
  \caption{A CTMC example and its network $\N(s_1, s_2, R_1)$.}
  \label{fig:ctmcsmf}
\end{figure}

  In the second iteration, the network $\N(s_1,s_2,R_2)$ is obtained from
  $\N(s_1,s_2,R_1)$ by deleting the edge $(u_2,\overline{u_4})$.  In
  $\N(s_1,s_2,R_2)$, the flows on $(u_2,\overline{u_4})$ and on
  $(\overline{u_4},\sink)$ are set to $0$, and the vertex $u_2$ has a
  positive excess $\frac{1}{2}$.  Applying the preflow algorithm, we push
  the excess from $u_2$, along $\overline{u_3}, u_1, \overline{u_4}$ to
  $\sink$. We get a maximum flow $f_2$ for $\N(s_1,s_2,R_2)$ which sends
  $\frac{1}{2}$ amount of flow along the path $\source,
  u_2,\overline{u_3},\sink$ and $\frac{1}{2}$ amount of flow along
  $\source, u_1, \overline{u_4}, \sink$. Hence, the check for $(s_1,s_2)$
  is also successful in the second iteration. Once the fix-point is
  reached, $R$ still contains $(s_1,s_2)$.
\end{exa}

\subsection{Strong Simulation for Probabilistic Automata}
\label{sec:pp}
In this subsection we present algorithms for deciding strong
simulations for PAs and CPAs. It takes the skeleton of the algorithm
for FPSs: it starts with a relation $R$ which is coarser than $\simrel$,
and then refines $R$ until $\simrel$ is achieved. In the refinement
loop, a pair $(s,s')$ is eliminated from the relation $R$ if the
corresponding strong simulation conditions are violated with respect
to the current relation. For PAs, this means that there exists an
$\alpha$-successor distribution $\mu$ of $s$, such that for all
$\alpha$-successor distribution $\mu'$ of $s'$, we cannot find a
weight function for $(\mu,\mu')$ with respect to the current relation
$R$.

Let $\pa$ be a PA. We aim to extend Algorithm~\ref{fig:simfpssmf} to
determine the strong simulation on PAs. For a pair $(s_1,s_2)$, assume that
$L(s_1)=L(s_2)$ and that $\mathit{Act}(s_1)\subseteq\mathit{Act}(s_2)$,
which is guaranteed by the initialisation. We consider
line~\ref{algsrf:testmatch}, which checks the condition
$\Pmu{s_1}\weight_{R_i} \Pmu{s_2}$ using $\prog{Smf}$. By
Definition~\ref{def:simple_ss} of strong simulation for PAs, we should
instead check the condition
\begin{equation}
\label{eq:checkss}
\forall \alpha \in \mathit{Act}.\ 
\forall s_1\transby{\alpha}\mu_1.\ \exists s_2\transby{\alpha}\mu_2
\text{ with } \mu_1\weight_{R_i}\mu_2 
\end{equation}
Recall the condition $\mu_1\weight_{R_i}\mu_2$ is true iff the maximum flow
of the network $\N(\mu_1,\mu_2,R_i)$ has value one.  Sometimes, we write
$\N(s_1,\alpha,\mu_1,s_2,\mu_2,R_i)$ to denote the network
$\N(\mu_1,\mu_2,R_i)$ associated with the pair $(s_1,s_2)$ with respect to
action $\alpha$.

Our first goal is to extend $\prog{Smf}$ to check
Condition~\ref{eq:checkss} for a fixed action $\alpha$ and
$\alpha$-successor distribution $\mu_1$ of $s_1$.  To this
end, we introduce a list $\mathit{Sim}^{(s_1,\alpha,\mu_1,s_2)}$ that
contains all potential candidates of $\alpha$-successor distributions
of $s_2$ which could be used to establish the condition
$\mu_1\weight_R \mu_2$ for the relation $R$ considered.  The set
$\mathit{Sim}^{(s_1,\alpha,\mu_1,s_2)}$ is represented as a list.  This and
some subsequent notations are similar to those used by Baier \et
in~\cite{BEMC00}.  We use the function \textbf{head}$(\cdot)$ to read
the first element of a list; \textbf{tail}$(\cdot)$
to read all but the first element of a list; and
\textbf{empty}$(\cdot)$ to check whether a list is empty.  As long as
the network for a fixed candidate $\mu_2 =
\mathbf{head}(\mathit{Sim}^{(s_1,\alpha,\mu_1,s_2)})$ allows a flow
of value $1$ over the iterations, we stick to it, and we can reuse the flow
and distance function from previous iterations.  If by deleting some edges
from $\N(\mu_1,\mu_2,R)$, its flow value falls below 1, we delete $\mu_2$
from $\mathit{Sim}^{(s_1,\alpha,\mu_1,s_2)}$ and pick the next candidate.

\begin{algorithm}[tbp]
  \caption{Subroutine to calculate whether $s_1 \simrel_{R_i} s_2$,
           as far as $s_1 \transby{\alpha} \mu_1$ is concerned. The
parameter $Sim$ denotes the subsets of $\alpha$-successor distributions
of $s_2$ serving as candidates for possible $\mu_2$.}
  \label{fig:actionsmf}
    \Procname{$\prog{ActSmf}(i, Sim_{i-1},\N(\mu_1,\mu_2,R_{i-1}), f_{i-1}, d_{i-1}, D_{i-1})$}
  \begin{algorithmic}[1]
        \STATE  $Sim_i \gets Sim_{i-1}$
        \STATE  $(\mathit{match}, \N(\mu_1,\mu_2,R_i), f_i, d_i)
                \gets
                \prog{Smf}(i, \N(\mu_1,\mu_2,R_{i-1}), f_{i-1},
                d_{i-1}, D_{i-1})$    \label{algactsmf:callsmf}
        \IF{$\mathit{match}$}
                \STATE  \textbf{return} $(\mathbf{true}, Sim_i, \N(\mu_1,\mu_2,R_i), f_i, d_i)$
        \ENDIF
        \STATE  $Sim_i \gets \mathbf{tail}(Sim_i)$   \label{algactsmf:tailSim_1}
        \WHILE{$\neg \mathbf{empty}(Sim_i)$}
                \label{algactsmf:while.start}
                \STATE  $\mu_2 \gets \mathbf{head}(Sim_i)$   \label{algactsmf:reinit.start}
                \STATE  $(\mathit{match}, \N(\mu_1,\mu_2,R_i), f_i, d_i)
                \gets
                \prog{Smf}_\init(i, \mu_1, \mu_2, R_i)$ 
                \IF{$\mathit{match}$}
                        \STATE  \textbf{return} $(\mathbf{true}, Sim_i, \N(\mu_1,\mu_2,R_i), f_i, d_i)$
                \ENDIF
                \STATE  $Sim_i \gets \mathbf{tail}(Sim_i)$   \label{algactsmf:tailSim}
        \ENDWHILE
        \STATE  \textbf{return} $(\mathbf{false},\emptyset,\mathit{NIL},\mathit{NIL},\mathit{NIL})$
                \label{algactsmf:returnfalse}
  \end{algorithmic}

\bigskip

    \Procname{$\prog{ActSmf}_\init(i, \mu_1, Sim_i, R_i)$}
  \begin{algorithmic}
        \STATE  \mbox{\phantom{4.1:}}~\textbf{goto} line~\ref{algactsmf:while.start}
  \end{algorithmic}

\end{algorithm}

The algorithm $\prog{ActSmf}$, shown as Algorithm~\ref{fig:actionsmf},
implements this.  It has to be called for each pair $(s_1,s_2)$ and each
successor distribution $s_1 \transby{\alpha} \mu_1$ of $s_1$.  It takes as
input the list of remaining candidates
$\mathit{Sim}_{i-1}^{(s_1,\alpha,\mu_1,s_2)}$, the information from the
previous iteration (the network $\N(\mu_1,\mu_2,R_{i-1})$, flow
$f_{i-1}$, and distance function $d_{i-1}$), and the set of edges that have
to be deleted from the old network $D_{i-1}$.
\begin{lem}
\label{lem:correctness_actsmf}
  Let $(s_1,s_2)\in R_1$, $\alpha\in Act(s_1)$, and $\mu_1$ such that
  $s_1\transby{\alpha}\mu_1$.  Let $Sim_1 = \steps_\alpha(s_2)$.  Then
  $\prog{ActSmf}_\init$ returns true iff $\exists \mu_2 \text{ with }s_2
  \transby{\alpha} \mu_2 \wedge \mu_1 \weight_{R_1} \mu_2$.
  For some $i>1$, let $Sim_{i-1}$, $\N(\mu_1,\mu_2,R_{i-1})$, $f_{i-1}$ and
  $d_{i-1}$ be as returned by some earlier call to $\prog{ActSmf}$ or
  $\prog{ActSmf}_\init$.  Let $D_{i-1} = (R_{i-1} \setminus R_{i}) \cap
  (\support(\mu_1) \times \support(\mu_2))$ be the set of edges that will
  be removed from the network during the $(i-1)$th call of $\prog{ActSmf}$.
  Then, the $(i-1)$th call of algorithm $\prog{ActSmf}$ returns true iff:
  $\exists \mu_2 \text{ with } s_2\transby{\alpha}\mu_2 \wedge
  \mu_1\weight_{R_i}\mu_2$.
\end{lem}
\begin{proof}
  Once $\prog{Smf}$ returns false because the maximum flow for the
  current candidate $\mu_2$ has value $<1$, it will never become a
  candidate again, as edge deletions cannot lead to increased flow.
  The correctness proof is then the same as the proof of
  Lemma~\ref{smf_correctness}.
\end{proof}

The algorithm $\prog{SimRel}^\mathrm{PA}_s$ for deciding strong simulation
for PAs is presented as Algorithm~\ref{fig:simrelpa}.  During the
initialisation (lines~\ref{algsrp:init.start}--\ref{algsrp:init.end},
intermixed with iteration~1 in
lines~\ref{algsrp:iter1.start}--\ref{algsrp:iter1.end}), for $(s_1,s_2)\in
R_1$ and $s_1\transby{\alpha}\mu_1$, the list
$Sim_1^{(s_1,\alpha,\mu_1,s_2)}$ is initialised to $\steps_\alpha(s_2)$
(line~\ref{algsrp:initSim}), as no $\alpha$-successor distribution can be
excluded as a candidate a priori.  As in $\prog{SimRel}^\mathrm{FPS}_s$,
the set $\listener_{(s_1,s_2)}$ for $(s_1,s_2)$ is introduced which
contains tuples $(u_1,\alpha,\mu_1,u_2,\mu_2)$ such that the network
$\N(u_1,\alpha,\mu_1,u_2,\mu_2,R_1)$ contains the edge
$(s_1,\overline{s_2})$.

\begin{algorithm}[tbp]
  \caption{Algorithm for deciding strong simulation for PAs, where $arc$
denotes the associated parameter $(s_1,\alpha,\mu_1,s_2,\mu_2)$.}
  \label{fig:simrelpa}
    \Procname{$\prog{SimRel}^\mathrm{PA}_s(\M)$}
  \begin{algorithmic}[1]
        \STATE  $R_1 \gets \{(s_1,s_2)\in S\times S\mid L(s_1)=L(s_2)\wedge Act(s_1)\subseteq Act(s_2)\}$
                and $i \gets 1$         \label{algsrp:init.start}
        \STATE  $R_2 \gets \emptyset$
        \FORALL{$(s_1,s_2)\in R_1$}
                \STATE  $\listener_{(s_1,s_2)}$ $\gets \{(u_1,\alpha,\mu_1,u_2,\mu_2)\mid$
                        \parbox[t]{6cm}{$L(u_1)=L(u_2)
                        \wedge u_1\transby{\alpha} \mu_1
                        \wedge u_2\transby{\alpha} \mu_2$ \\
                        $\mbox{} \wedge \mu_1(s_1)>0
                        \wedge \mu_2(s_2)>0\}$}
                \FORALL{$\alpha\in Act(s_1), \quad
                         \mu_1\in \steps_\alpha(s_1)$}
                        \STATE  $\mathit{Sim}_1^{(s_1,\alpha,\mu_1,s_2)} \gets \steps_\alpha(s_2)$
                                \label{algsrp:initSim}
                                \label{algsrp:init.end}
                        \STATE  $(\mathit{match}_{\alpha,\mu_1},
                        \mathit{Sim}_1^{(s_1,\alpha,\mu_1,s_2)},
                        \N(s_1,\alpha,\mu_1,s_2,\mu_2,R_1), 
                        f_1^{arc},
                        d_1^{arc})$ \\
                        $\qquad\qquad\qquad  \gets
                        \prog{ActSmf}_\init(1, \mu_1,
                        \mathit{Sim}_1^{(s_1,\alpha,\mu_1,s_2)},
                        R_1)$                             \label{algsrp:iter1.start}
                \ENDFOR
                \IF{$\bigwedge_{\alpha\in
Act(s_1)}\bigwedge_{\mu_1\in \steps_\alpha(s_1)}\mathit{match}_{\alpha,\mu_1}$}
                        \label{algsrp:checkmatch_1}
                        \STATE  $R_2\gets R_2\cup \{(s_1,s_2)\}$
                \ENDIF  \label{algsrp:iter1.end}
        \ENDFOR
        \WHILE{$R_{i+1} \not= R_i$}
                \STATE  $i \gets i + 1$
                \STATE  $R_{i+1} \gets \emptyset$
                        and $D_{i-1} \gets R_{i-1} \setminus R_i$
                        \label{algsrp:reinit.start}
                \FORALL{$(s_1,s_2)\in R, \quad
                        \alpha\in Act(s_1), \quad
                        \mu_1\in \steps_\alpha(s_1),\quad \mu_2\in \steps_\alpha(s_2)$}
                        \label{algsrp:setDs1s2.start}
                        \STATE  $D_{i-1}^{(s_1, \alpha, \mu_1, s_2,\mu_2)} \gets \emptyset$
                \ENDFOR
                \FORALL{$(s_1,s_2)\in D_{i-1}, \quad
                         (u_1,\alpha,\mu_1,u_2,\mu_2)\in
                         \listener_{(s_1,s_2)}$}
                        \IF{$(u_1, u_2) \in R_{i-1}$}
                                \STATE  $D_{i-1}^{(u_1,\alpha,\mu_1,u_2,\mu_2)}\gets D_{i-1}^{(u_1,\alpha,\mu_1,u_2,\mu_2)} \cup
                  \{(s_1,s_2)\}$   \label{algsrp:reinit.end}
                        \ENDIF
                \ENDFOR \label{algsrp:setDs1s2.end}
                \FORALL{$(s_1,s_2)\in R$}
                        \FORALL{$\alpha\in Act(s_1), \quad
                           \mu_1\in \steps_\alpha(s_1)$}   \label{algsrp:checksim.start}
                                \STATE  $(\mathit{match}_{\alpha,\mu_1},
                                \mathit{Sim}_i^{(s_1,\alpha,\mu_1,s_2)},
                                \N(s_1,\alpha,\mu_1,s_2,\mu_2,R_i),
                                f_i^{arc},
                                d_i^{arc})$ \\
                                $\qquad \qquad \qquad \gets
                                        \prog{ActSmf}($\parbox[t]{5cm}{$i,
                                                \mathit{Sim}_{i-1}^{(s_1,\alpha,\mu_1,s_2)},
                                                \N(s_1,\alpha,\mu_1,s_2,\mu_2,R_{i-1}),$ \\
                                                $f_{i-1}^{arc},
                                                d_{i-1}^{arc},
                                                D_{i-1}^{arc})$} 
                        \ENDFOR
                        \IF{$\bigwedge_{\alpha\in
Act(s_1)}\bigwedge_{\mu_1\in Steps_\alpha(s_1)}\mathit{match}_{\alpha,\mu_1}$}
                                \label{algsrp:checkmatch}
                                \STATE  $R_{i+1}\gets R_{i+1}\cup \{(s_1,s_2)\}$
                                        \label{algsrp:updateRnew}
                        \ENDIF   \label{algsrp:checksim.end}
                \ENDFOR
        \ENDWHILE
        \STATE  \textbf{return} $R_i$   \label{algsrp:returnR}
  \end{algorithmic}
\end{algorithm}

The main iteration starts with generating the sets
$D_{i-1}^{(u_1,\alpha,\mu_1,u_2,\mu_2)}$ in
lines~\ref{algsrp:setDs1s2.start}--\ref{algsrp:setDs1s2.end} in a similar
way as $\prog{SimRel}^\mathrm{FPS}_s$.
Lines~\ref{algsrp:checksim.start}--\ref{algsrp:checksim.end} check
Condition~\ref{eq:checkss} by calling $\prog{ActSmf}$ for each action
$\alpha$ and each $\alpha$-successor distribution $\mu_1$ of $s_1$.  The
condition is true if and only if $match_{\alpha,\mu_1}$ is true for all
$\alpha\in Act(s_1)$ and $\mu_1\in Steps_\alpha(s_1)$.  In this case we
insert the pair $(s_1,s_2)$ into $R_{i+1}$ (line~\ref{algsrp:updateRnew}).
We give the correctness of the algorithm:

\begin{thm}
When $\prog{SimRel}^\mathrm{PA}_s(\M)$ terminates, the returned relation equals
$\simrel_\M$. 
\end{thm}
\begin{proof}
  The proof follows the same lines as the proof of the correctness of
  $\prog{SimRel}^\mathrm{FPS}_s$ in Theorem~\ref{thm:fpscorrectness}.  The
  only new element is that we now have to quantify over the actions and
  successor distributions as prescribed by Definition~\ref{def:simple_ss}.
  This translates to a conjunction in lines~\ref{algsrp:checkmatch_1} and
  \ref{algsrp:checkmatch} of the algorithm. Exploiting
  Lemma~\ref{lem:correctness_actsmf} we get the correctness.
\end{proof}

Now we
give the complexity of the algorithm:

\begin{thm}
The algorithm $\prog{SimRel}^\mathrm{PA}_s(\M)$ runs in time $\O(m^2n)$
and in space $\O(m^2)$.  If the fanout of $\M$ is bounded by a
constant, it has complexity $\O(n^2)$, both in time and space.
\end{thm}
\begin{proof}
We first consider  the space complexity.  We save the sets
$D_i^{(s_1,\alpha,\mu_1,s_2,\mu_2)}$, the networks
$\N(s_1,\alpha,\mu_1,s_2,\mu_2,R_i)$ which are updated in every iteration,
$\listener_{(s_1,s_2)}$ and the sets
$\mathit{Sim}_i^{(s_1,\alpha,\mu_1,s_2)}$.  The size of the set
$D_i^{(s_1,\alpha,\mu_1,s_2,\mu_2)}$ is in $\O(\size{\mu_1}
\size{\mu_2})$, which is the maximal number of
edges of $\N(s_1,\alpha,\mu_1,s_2,\mu_2,R_i)$.  Summing over all $(s_1,
\alpha, \mu_1, s_2,\mu_2)$, we get:
\begin{align}\label{eq:bounds}
    \sum_{s_1 \in S} \sum_{\alpha\in Act(s_1)}
    \sum_{\mu_1\in Steps_\alpha(s_1)} \sum_{s_2\in S}
    \sum_{\mu_2\in Steps_\alpha(s_2)}
 \size{\mu_1} \size{\mu_2} 
\leq    m^2
\end{align}
Similarly, the memory needed for saving the networks has the same bound
$\O(m^2)$.  Now we consider the set $\listener_{(s_1,s_2)}$ for the pair
$(s_1,s_2) \in R_1$.  Let $(u_1,\alpha,\mu_1,u_2,\mu_2) \in
\listener_{(s_1,s_2)}$. Then, it holds that $s_1 \in \support(\mu_1)$ and
$s_2 \in \support(\mu_2)$.  Hence, the tuple $(u_1,\alpha,\mu_1,u_2,\mu_2)$
can be an element of $\listener_{(s_1,s_2)}$ of some arbitrary pair
$(s_1,s_2)$ at most $\size{\mu_1} \size{\mu_2}$ times, which corresponds to
the maximal number of edges between the set of nodes $\support(\mu_1)$ and
$\overline{\support(\mu_2)}$ in
$\N(s_1,\alpha,\mu_1,s_2,\mu_2,R_1)$. Summing over all $(s_1,
\alpha, \mu_1, s_2,\mu_2)$,  we get that memory needed for the set
$\listener$ is also bounded by $\O(m^2)$. For each pair $(s_1,s_2)$ and
$s_1 \transby{\alpha} \mu_1$, the set
$\mathit{Sim}_1^{(s_1,\alpha,\mu_1,s_2)}$ has size
$\size{Steps_\alpha(s_2)}$.  Summing up, this is smaller than or equal to $m^2$
according to Inequality~\ref{eq:bounds}.  Hence, the overall space
complexity amounts to $\O(m^2)$.

Now we consider the time complexity.  All initialisations
(lines~\ref{algsrp:init.start}--\ref{algsrp:init.end} of
$\prog{SimRel}^\mathrm{PA}_s$ and the initialisations in
$\prog{ActSmf}_\init$, which calls $\prog{Smf}_\init$) take $\O(m^2)$ time.
We observe that a pair $(s_1,s_2)$ belongs to $D_i$ during at most one
iteration.  Because of the Inequality~\ref{eq:bounds}, the time needed in
lines~\ref{algsrp:setDs1s2.start}--\ref{algsrp:setDs1s2.end} is in
$\O(m^2)$.  The rest of the algorithm is dominated by the time needed for
calling $\prog{Smf}$ in line~\ref{algactsmf:callsmf} of $\prog{ActSmf}$.
By Lemma~\ref{smf_complexity}, the time complexity for successful and
unsuccessful checks concerning the tuple $(s_1,\alpha,
\mu_1,s_2,\mu_2)$ is bounded by
$\O(g\size{\mu_1}\size{\mu_2})$.  Taking the sum over all possible
tuples  $(s_1,\alpha,
\mu_1,s_2,\mu_2)$  we get the bound $g m^2$ according to
    Inequality~\ref{eq:bounds}.  Hence, the complexity is $\O(m^2n)$. If
    the fanout $g$ is bounded by a constant, we have $m \leq gn$. Thus, the
    time complexity is in the order of $\O(n^2)$.  In this case the space
    complexity is also $\O(n^2)$.
\end{proof}

\begin{rem}
\label{remark:baier_complexity}
Let $\pa$ be a PA, and let $\mbaier=\sum_{s\in S} \sum_{\alpha \in Act(s)}
\size{Steps_\alpha(s)}$, called the number of transitions
in~\cite{BEMC00}, denote the number of all distributions in $\M$.  The
algorithm for deciding strong simulation introduced by Baier \et has time
complexity $\O((\mbaier n^6+ \mbaier^2n^3)/\log n)$, and space complexity
$\O(\mbaier^2)$. The number of distributions $\mbaier$ and the size of
transitions $m$ are related by $\mbaier
\le m \le n\mbaier$.  The left
equality is established if $\size{\mu}=1$ for all distributions, and the
right equality is established if $\size{\mu}=n$ for all distributions. 
\end{rem}

The decision algorithm for strong simulation for CPAs can be adapted
from $\prog{SimRel}^\mathrm{PA}_s$ in Algorithm~\ref{fig:simrelpa}
easily: Notations are extended with respect to rate functions instead
of distributions in an obvious way.  To guarantee the additional rate
condition, we rule out successor rate functions of $s_2$ that violate
it by replacing line~\ref{algsrp:initSim} by:
\[\mathit{Sim}_1^{(s_1,\alpha,r_1,s_2)}\gets
  \{ r_2 \in Steps_\alpha(s_2) \mid r_1(S) \le r_2(S)\}.
\]
For each pair $(s_1,s_2)$, and successor rate functions
$r_i\in\steps_\alpha(s_i)$ ($i=1,2$), the subroutine for checking
whether $r_1\weight_{R_i} r_2$ is then performed in the network
$\N(\mu(r_1),\mu(r_2), R_i)$.  Obviously, the so obtained algorithm
for CPAs has the same complexity $\O(m^2n)$.

\subsection{Strong Probabilistic Simulation}\label{sec:sps} 

The problem of deciding strong probabilistic simulation for PAs has
not been tackled yet.  We show that it can be computed by solving LP
problems which are decidable in polynomial time~\cite{Karmarkar84}. In
Subsection~\ref{sec:algo_pa_sps}, we first present an algorithm for
PAs.  We extend the algorithm to deal with CPAs in
Subsection~\ref{sec:algo_cpa_sps}.

\subsubsection{Probabilistic Automata}\label{sec:algo_pa_sps}
Recall that strong probabilistic simulation is a relaxation of strong
simulation in the sense that it allows combined transitions, which are
convex combinations of multiple distributions belonging to equally
labelled transitions. Again, the most important part is to check
whether $s_1 \simrel_R^p s_2$ where $R$ is the current relation. By
Definition~\ref{def:combined_ss}, it suffices to check $L(s_1)=L(s_2)$
and the condition:
\begin{align}
\label{eq:checksps}
\forall \alpha \in \mathit{\actions}.\ 
\forall s_1\transby{\alpha}\mu_1.\ \exists s_2\combinedby{\alpha}\mu_2
\text{ with } \mu_1\weight_R\mu_2
\end{align}
However, since the combined transition involves the quantification of the
constants $c_i \in [0,1]$, there are possibly infinitely many such
$\mu_2$.  Thus, one cannot check $\mu_1 \weight_R
\mu_2$ for each possible candidate $\mu_2$. The following lemma  shows that
this condition can be checked by solving LP problems which are decidable in
polynomial time~\cite{Karmarkar84,Schr86}. 

\begin{lem}
\label{lem:lp}
  Let $\pa$ be a given PA, and let $R\subseteq S\times S$. Let
  $(s_1,s_2)\in R$ with $L(s_1)=L(s_2)$ and $Act(s_1)\subseteq
  Act(s_2)$. Then, $s_1 \simrel_R^p s_2$ iff for each transition
  $s_1\transby{\alpha} \mu$, the following LP has a feasible solution:
\begin{align} 
& \sum_{i=1}^k c_i = 1  \label{eq:lpone}\\ 
& 0 \le c_i \le 1 \qquad \forall \ i=1,\ldots, k \label{eq:lptwo}\\ 
& 0 \le  f_{(s,t)} \le 1 \qquad \forall (s,t) \in R_\aux \label{eq:lpthree}\\ 
& \mu(s) = \sum_{t \in R_\aux(s)} f_{(s,t)} \qquad \forall s \in S_\aux \label{eq:lpfour}\\ 
& \sum_{s \in R_\aux^{-1}(t)} f_{(s,t)} = \sum_{i=1}^k c_i \mu_i(t) \qquad
\forall t \in S_\aux \label{eq:lpfive} 
\end{align} 
where $k=\size{\steps_\alpha(s_2)}>0$ and $\steps_\alpha(s_2) = \{ \mu_1,
\ldots, \mu_k \}$.
\end{lem}
\begin{proof}
First assume that $s_1 \simrel_R^p s_2$. Let $s_1\transby{\alpha}\mu$.
By the definition of simulation up to $R$ for strong probabilistic
simulation, there exists a combined transition $s_2
\combinedby{\alpha} \mu_c$ with $\mu \weight_R \mu_c$. Let
$\steps_\alpha(s_2) = \{ \mu_1, \ldots, \mu_k \}$ where
$k=\size{\steps_\alpha(s_2)}$.  Now $Act(s_1)\subseteq Act(s_2)$
implies $k>0$. By definition of combined transition
(Definition~\ref{def:combined}), there exist constants $c_1, \ldots,
c_k \in [0,1]$ with $\sum_{i=1}^k c_i=1$ such that $\mu_c =
\sum_{i=1}^k c_i \mu_i$. Thus Constraints~\ref{eq:lpone} and
\ref{eq:lptwo} hold. Since $\mu \weight_R \mu_c$, there exists a
weight function $\Delta:S_\aux\times S_\aux\rightarrow [0,1]$ for
$(\mu,\mu_c)$ with respect to $R$. For every pair $(s,t)\in R_\aux$,
let $f_{(s,t)}:=\Delta(s,t)$. Thus, Constraint~\ref{eq:lpthree} holds
trivially. By Definition~\ref{def:weight} of weight functions, it
holds that
(i)  $\Delta(s,t)>0$ implies that $(s,t)\in R_\aux$,
(ii) $\mu(s)=\sum_{t\in S_\aux}
\Delta(s,t)$ for $s\in S_\aux$, and
(iii) $\mu_c(t)=\sum_{s \in S_\aux} \Delta(s,t)$
for all $t\in S_\aux$.
Observe that (i) implies that for all $(s,t)\not\in R_\aux$, we
have that $\Delta(s,t)=0$. Thus, (ii) and (iii) imply
Equations~\ref{eq:lpfour} and \ref{eq:lpfive} respectively.

Now we show the other direction. Let $k=\size{\steps_\alpha(s_2)}$ and
$\steps_\alpha(s_2) = \{ \mu_1, \ldots, \mu_k \}$.  By assumption, for
each $s_1\transby{\alpha}\mu$, we have a feasible solution
$c_1,\ldots,c_k$ and $f_{(s,t)}$ for all $(s,t)\in R_\aux$ which
satisfies all of the constraints. We define
$\mu_c=\sum_{i=1}^kc_i\mu_i$. By Definition~\ref{def:combined},
$\mu_c$ is a combined transition, thus $s_2\combinedby{\alpha}\mu_c$.
It remains to show that $\mu\weight_R \mu_c$.  We define a function
$\Delta$ as follows: $\Delta(s,t)$ equals $f_{(s,t)}$ if $(s,t)\in
R_\aux$ and $0$ otherwise. With the help of
Constraints~\ref{eq:lpthree}, \ref{eq:lpfour} and \ref{eq:lpfive} we
have that $\Delta$ is a weight function for $(\mu,\mu_c)$ with respect
to $R$, thus $\mu \weight_R\mu_c$.
\end{proof}

Now we are able to check Condition~\ref{eq:checksps} by solving LP
problems.  For a PA $\pa$, and a relation $R\subseteq S\times S$, let
$(s_1,s_2)\in R$ with $L(s_1)=L(s_2)$ and $Act(s_1)\subseteq
Act(s_2)$. For $s_1\transby{\action}\mu_1$, we introduce a predicate
$LP(s_1,\alpha,\mu,s_2)$ which is true iff the LP problem described as
in Lemma~\ref{lem:lp} has a solution.  Then, $s_1 \simrel_R^p s_2$ iff
the conjunction $ \bigwedge_{\alpha \in\actions(s_1)}
\bigwedge_{\mu_1\in\steps_\alpha(s_1)} LP(s_1,\alpha,\mu_1,s_2) $ is
true.  The algorithm, which is denoted by
$\prog{SimRel}_s^{\mathrm{PA},p}(\M)$, is depicted in
Algorithm~\ref{fig:simrelsps}. It takes the skeleton of
$\prog{SimRel}_s(\M)$. The key difference is that we incorporate the
predicate $LP(s_1,\alpha,\mu_1,s_2)$ in line~\ref{alg:sps_pa:match}.
The correctness of the algorithm $\prog{SimRel}_s^{\mathrm{PA},p}(\M)$
can be obtained from the one of $\prog{SimRel}_s(\M)$ together with
Lemma~\ref{lem:lp}.  We discuss briefly the complexity.  The number of
variables in the LP problem in Lemma~\ref{lem:lp} is $k+\size{R}$, and
the number of constraints is $1+k+\size{R}+2\size{S} \in
\O(\size{R})$. In iteration $i$ of
$\prog{SimRel}_s^{\mathrm{PA},p}(\M)$, for $(s_1,s_2)\in R_i$ and
$s_1\transby{\alpha} \mu_1$, the corresponding LP problem is queried
once. The number of iterations is in $\O(n^2)$. Therefore, in the
worst case, one has to solve $ n^2\sum_{s\in S} \sum_{\alpha\in
  \actions(s)} \sum_{\mu \in \steps_\alpha(s)} 1 \in \O(n^2m) $ many
such LP problems and each of them has at most $\O(n^2)$ constraints.

\begin{algorithm}[tbp]
  \caption{ Algorithm for deciding strong probabilistic simulation for PAs.}
  \label{fig:simrelsps}
    \Procname{$\prog{SimRel}_s^{\mathrm{PA},p}(\M)$}
  \begin{algorithmic}[1]
        \STATE  $R_1 \gets \{(s_1,s_2)\in S\times S\mid L(s_1)=L(s_2)\wedge Act(s_1)\subseteq Act(s_2)\}$
                and $i \gets 0$
        \REPEAT
                \STATE  $i \gets i + 1$
                \STATE  $R_{i+1}\gets\emptyset$
                \FORALL{$(s_1,s_2)\in R_i$}
                        \FORALL{$\alpha\in \actions(s_1), \quad \mu_1\in \steps_\alpha(s_1)$}\label{alg:sps_pa:for}
                                \STATE  $match_{\alpha,\mu_1} \gets LP(s_1,\alpha,\mu_1,s_2)$\label{alg:sps_pa:match}
                        \ENDFOR
                        \IF{$\bigwedge_{\alpha\in \actions(s_1)}\bigwedge_{\mu_1\in\steps_\alpha(s_1)} match_{\alpha,\mu_1}$}
                                \STATE  $R_{i+1}\gets R_{i+1}\cup \{(s_1,s_2)\}$
                        \ENDIF
                \ENDFOR
        \UNTIL{$R_{i+1}=R_i$}
        \STATE  \textbf{return} $R_i$
  \end{algorithmic}
\end{algorithm}

\subsubsection{Continuous-time Probabilistic Automata}\label{sec:algo_cpa_sps}
Now we discuss how to extend the algorithm to handle CPAs.  Let $\cpa$ be a
CPA. Similar to PAs, the most important part is to check the condition $s_1
\simrel_R^p s_2$ for some relation $R\subseteq S\times S$. By
Definition~\ref{def:cpa_sps}, it suffices to 
check $L(s_1)=L(s_2)$ and the condition:
\begin{align}
\label{eq:checksps_cpa}
\forall \alpha \in \mathit{\actions}.\ 
\forall s_1\transby{\alpha} r_1.\ \exists s_2\combinedby{\alpha} r_2
\text{ with } \mu(r_1) \weight_R \mu(r_2) \wedge r_1(S)
\le r_2(S)
\end{align}
Recall that for CPAs only successor rate functions with the same exit
rate can be combined together. For state $s\in S$, we let $E(s):=
\{r(S) \mid s \transby{\alpha} r\}$ denote the set of all possible
exit rates of $\alpha$-successor rate functions of $s$.  For $E\in
E(s)$ and $\alpha\in\actions(s)$, we let
$\steps^E_\alpha(s)=\{r\in\steps_\alpha(s)\mid r(S)=E\}$ denote the
set of $\alpha$-successor rate functions of $s$ with the same exit
rate $E$. As for PAs, to check the condition $s_1 \simrel_R^p s_2$ we
resort to a reduction to LP problems.

\begin{lem}
  Let $\cpa$ be a given CPA, and let $R\subseteq S \times S$. Let
  $(s_1,s_2)\in R$ with $L(s_1)=L(s_2)$ and that $Act(s_1)\subseteq
  Act(s_2)$. Then, $s_1 \simrel_R^p s_2$ iff for each transition
  $s_1\transby{\alpha} r$ either $r(S)=0$, or there exists $E\in E(s_2)$
  with $E \ge r(S)$ such that the following LP has a feasible solution,
  which consists of Constraints~\ref{eq:lpone}, \ref{eq:lptwo},
  \ref{eq:lpthree} of Lemma~\ref{lem:lp}, and additionally:
\begin{align} 
    & r(s) = r(S) \sum_{t \in R_\aux(s)} f_{(s,t)} \qquad
      \forall s \in S_\aux \label{eq:nlpone}\\
    &E \sum_{s \in R_\aux^{-1}(t)} f_{(s,t)} =
    \sum_{i=1}^k c_i r_i(t) \qquad \forall t \in S_\aux  \label{eq:nlptwo}
\end{align}
where $k= \size{\steps^E_\alpha(s)}$ with $\steps^E_\alpha(s)=\{r_1,\ldots,r_k\}$.
\end{lem}
\begin{proof}
  The proof follows the same strategy as the proof of
  Lemma~\ref{lem:lp}, in which the induced distribution of the
  corresponding rate function should be used.
\end{proof}

\begin{algorithm}[tbp]
  \caption{Algorithm for deciding strong probabilistic simulation for CPAs.}
  \label{fig:simrelsps_cpa}
    \Procname{$\prog{SimRel}_s^{\mathrm{CPA},p}(\M)$}
  \begin{algorithmic}[1]
        \STATE  $R_1 \gets \{(s_1,s_2)\in S\times S\mid L(s_1)=L(s_2) \wedge Act(s_1)\subseteq Act(s_2)\}$
                and $i \gets 0$
        \REPEAT
                \STATE  $i \gets i + 1$
                \STATE  $R_{i+1}\gets\emptyset$
                \FORALL{$(s_1,s_2)\in R_i$}
                        \FORALL {$\alpha\in \actions(s_1), \quad
                          r_1\in \steps_\alpha(s_1), \quad E\in E(s_2)$}
                                \STATE  $match_{\alpha,r_1,E} \gets LP'(s_1,\alpha,r_1,s_2,E)$
                        \ENDFOR
                        \IF{$\bigwedge_{\alpha\in
                            \actions(s_1)}\bigwedge_{r_1\in\steps_\alpha(s_1)}\bigwedge_{E\in E(s_2)} match_{\alpha,r_1,E}$}
                                \STATE  $R_{i+1}\gets R_{i+1}\cup \{(s_1,s_2)\}$
                        \ENDIF
                \ENDFOR
        \UNTIL{$R_{i+1}=R_i$}
        \STATE  \textbf{return} $R_i$
  \end{algorithmic}
\end{algorithm}

Now we are able to check Condition~\ref{eq:checksps_cpa} by solving LP
problems.  For a CPA $\cpa$, and a relation $R\subseteq S\times S$,
let $(s_1,s_2)\in R$ with $L(s_1)=L(s_2)$ and $Act(s_1)\subseteq
Act(s_2)$. For $s_1\transby{\action} r_1$, and $E\in E(s_2)$, we
introduce the predicate $LP'(s_1,\alpha,r_1,s_2, E)$ which is true iff
$E\ge r_1(S)$ and the corresponding LP problem has a solution.  Then,
$s_1 \simrel_R^p s_2$ iff the conjunction $ \bigwedge_{\alpha \in
  \actions(s_1)} \bigwedge_{r_1\in \steps_\alpha(s_2)} \bigwedge_{E\in
  E(s_2)} LP'(s_1,\alpha,r_1,s_2,E) $ is true. The decision algorithm
is given in Algorithm~\ref{fig:simrelsps_cpa}.  As complexity we have
to solve $\O(n^2m)$ LP problems and each of them has at most $\O(n^2)$
constraints.

\section{Algorithms for Deciding Weak Simulations}
\label{sec:weak} 
\noindent We now turn our attention to weak simulations.  Similar to strong
simulations, the core of the algorithm is to check whether $s_1\wsrel_R
s_2$, i.e., $s_2$ weakly simulates $s_1$ up to the current relation $R$.
As for strong simulation up to $R$, $s_1\wsrel_Rs_2$ does not imply
$s_1\wsrel_\D s_2$, since no conditions are imposed on pairs in $R$
different from $(s_1,s_2)$. By the definition of weak simulation, for fixed
characteristic functions $\delta_i$ ($i=1,2$), the weight function
conditions can be checked by applying maximum flow algorithms.
Unfortunately, $\delta_i$-functions are not known a priori.  Inspired by
the parametric maximum flow algorithm, in this chapter, we show that one
can determine whether such characteristic functions $\delta_i$ exist with
the help of \emph{breakpoints}, which can be computed by analysing a
parametric network constructed out of $\Pmu{s_1},\Pmu{s_2}$ and $R$.  We
present dedicated algorithms for DTMCs in Subsection~\ref{sec:weak_dtmc}
and CTMCs in Subsection~\ref{sec:weak_ctmc}.

\subsection{An Algorithm for DTMCs}
\label{sec:weak_dtmc}
Let $\dtmc$ be a DTMC. Let $R\subseteq S\times S$ be a relation and
$s_1 \mathrel{R} s_2$. Whether $s_2$ weakly simulates $s_1$ up to $R$ is
equivalent to whether there exist functions $\delta_i: S\to [0,1]$
such that the conditions in
Definition~\ref{def:dtmc_weak_simulation} are satisfied.  Assume
that we are given the $U_i$-characterising functions $\delta_i$.  In
this case, $s_1 \wsrel_R s_2$ can be checked as follows:
\begin{enumerate}[$\bullet$]
\item Concerning Condition~\ref{ws:stuttersimulates}a we check whether for all $v\in S$ with
  $\delta_1(v)<1$ it holds that $v \mathrel{R} s_2$. Similarly, for Condition~\ref{ws:stuttersimulates}b, we check whether for all $v\in S$ with $\delta_2(v)<1$ it holds
  that $s_1 \mathrel{R} v$. 
\item The reachability condition can be checked by using standard
  graph algorithms. In more detail, for each $u$ with $\delta_1(u)>0$,
  the condition holds if a state in $R(u)$ is reachable from $s_2$ via
  $R(s_1)$ states.
\item Finally consider Condition~\ref{ws:weightfunction}. From the given
$\delta_i$ functions we can compute $K_i$. In case of that $K_1>0$ and
$K_2>0$, we need to check whether there exists a weight function for the
conditional distributions ${\Pmu{s_1}}/{K_1}$ and
${\Pmu{s_2}}/{K_2}$ with respect to the current relation $R$.  From
Lemma~\ref{lem:weight_equivalent}, this is equivalent to check whether the
maximum flow for the network constructed from $({\Pmu{s_1}}/{K_1},
{\Pmu{s_2}}/{K_2})$ and $R$ has value $1$.
\end{enumerate}\smallskip

\noindent To check $s_1\wsrel_R s_2$, we want to check whether such
$\delta_i$ functions exist.  The difficulty is that there exist
uncountably many possible $\delta_i$ functions.  In this section, we
first show that whether such $\delta_i$ exists can be characterised by
analysing a parametric network in
Subsection~\ref{sec:parametric_network}. Then, in
Subsection~\ref{sec:breakpoints}, we recall the notion of breakpoints,
and show that the breakpoints play a central role in the parametric
networks considered: only these points need to be considered. Based on
this, we present the algorithm for DTMCs in
Subsection~\ref{sec:weak_alg_dtmc}. An improvement of the algorithm
for certain cases is reported in Subsection~\ref{sec:improvement}.

\subsubsection{The Parametric Network $\N(\gamma)$}
\label{sec:parametric_network}
Let $\gamma\in\Real_{\geq 0}$. Recall that $\N(\Pmu{s_1}, \gamma\Pmu{s_2}, R)$
is obtained from the network $\N(\Pmu{s_1}, \Pmu{s_2}, R)$ be setting
the capacities to the sink $\sink$ by: $\capacity(\overline
t,\sink)=\gamma\P(s_2,t)$. If $s_1,s_2,R$ are clear from the context, we
use $\N(\gamma)$ to denote the network
$\N(\Pmu{s_1},\gamma\Pmu{s_2},R)$ for arbitrary $\gamma\in\Real_{\geq 0}$.

We introduce some notations.  We focus on a particular pair
$(s_1,s_2)\in R$, where $R$ is the current relation. We partition the
set $\post(s_i)$ into $MU_i$ (for: must be in $U_i$) and $PV_i$ (for:
potentially in $V_i$).  The set $PV_1$ consists of those successors of
$s_1$ which can be either put into $U_1$ or $V_1$ or both:
$PV_1=\post(s_1)\cap R^{-1}(s_2)$.  The set $MU_1$ equals
$\post(s_1)\backslash PV_1$, which consists of the successor states
which can only be placed in $U_1$.  The sets $PV_2$ and $MU_2$ are
defined similarly by: $PV_2 = \post(s_2)\cap R(s_1)$ and
$MU_2=\post(s_2)\backslash PV_2$. Obviously,
$\delta_i(u)=1$ for $u\in MU_i$ for $i=1,2$.

\begin{figure}[tbp]
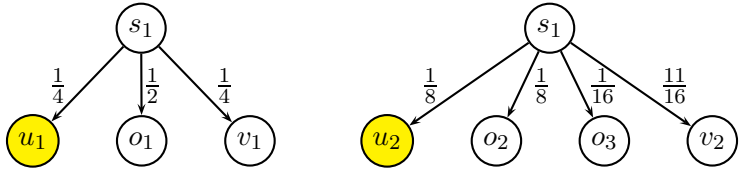

\begin{center}
\psset{levelsep=1.5cm,arrows=->,labelsep=0pt}
\pstree{\Tcircle{$s_1$}}{%
  \Tcircle[fillstyle=solid,fillcolor=yellow]{$u_1$}  \tlput{$\frac{1}{4}$}
  \Tcircle{$o_1$} \trput{$\frac{1}{2}$}
  \Tcircle{$v_1$} \trput{$\frac{1}{4}$}
}\hspace{1cm}
\pstree{\Tcircle{$s_1$}}{%
  \Tcircle[fillstyle=solid,fillcolor=yellow]{$u_2$}  \tlput{$\frac{1}{8}$}
  \Tcircle{$o_2$} \trput{$\frac{1}{8}$}
  \Tcircle{$o_3$} \trput{$\frac{1}{16}$}
  \Tcircle{$v_2$} \trput{$\frac{11}{16}$}
}
\end{center}
  \caption{A simple DTMC.}
  \label{fig:dtmc_partition}
\end{figure}
\begin{exa}\label{exa:dtmc_split}
  Consider the DTMC in Figure~\ref{fig:dtmc_partition} and the relation
  $R=\{(s_1,s_2), (s_1,v_2)$, $(v_1,s_2),
  (u_1,u_2), (o_1,o_2), (o_1,v_2), (v_1,o_3),(v_1,v_2), (o_2,o_1)\}$.
  We have $PV_1=\{v_1\}$ and $PV_2=\{v_2\}$. Thus, $MU_1 =
  \{u_1,o_1\}$, $MU_2=\{u_2,o_2,o_3\}$.
\end{exa}

We say a flow function $f$ of $\N(\gamma)$ is valid for $\N(\gamma)$
iff $f$ saturates all edges $(\source, u_1)$ with $u_1\in MU_1$ and
all edges $(\overline{u_2},\sink)$ with $u_2\in MU_2$.  If there
exists a valid flow $f$ for $\N(\gamma)$, we say that $\gamma$ is
valid for $\N(\gamma)$.  The following lemma considers the case in
which both $s_1$ and $s_2$ have visible steps:

\begin{lem}\label{valid_gamma}
  Let $s_1 \mathrel{R} s_2$. Assume that there exists a state $s_1'\in
  \post(s_1)$ such that $s_1'\not\in R^{-1}(s_2)$, and $s_2'\in
  \post(s_2)$ such that $s_2'\not\in R(s_1)$.  Then, $s_1\wsrel_R s_2$
  iff there exists a valid $\gamma$ for $\N(\gamma)$.
\end{lem}
\begin{proof}
  By assumption, we have that $s_i'\in MU_i$ for $i=1,2$, thus
  $MU_i\not=\emptyset$, and it holds that $\delta_i(s_i')=1$ for $i=1,2$.

We first show the \emph{only if} direction.  Assume $s_1\wsrel_R s_2$, and let 
$\delta_i,U_i,V_i,K_i,\Delta$ (for $i=1,2$) as described in
Definition~\ref{def:dtmc_weak_simulation}. Since $MU_i\not=\emptyset$
for $i=1,2$, both $K_1$ and $K_2$ are greater than $0$.
We let $\gamma={K_1}/{K_2}$. For $s,s'\in S$, we define
the function $f$ for 
$\N(\gamma)$:
\begin{align*}
f(\source,s) = \P(s_1,s)\delta_1(s), \tab 
f(s,\overline t) = K_1\Delta(s,t), \tab
f(\overline s,\sink) = \gamma\P(s_2,s)\delta_2(s)
\end{align*}
Since $\delta_i(s)\le 1$ ($i=1,2$) for $s\in S$,
$f(\source,s)\le\P(s_1,s)$ and
$f(\overline{s},\sink)\le\gamma\P(s_2,s)$. Therefore, $f$ satisfies the
capacity constraints. $f$ also satisfies the conservation rule:
\begin{align*}
  &f(s,\overline{S})=K_1
    \Delta(s,S) = \P(s_1,s)\delta_1(s)=f(\source,s)\\
  &f(S,\overline s)=K_1 \Delta(S,s) = \gamma 
    K_2 \Delta(S,s) = \gamma\P(s_2,s)\delta_2(s)=f(\overline s,\sink)
\end{align*}
Hence, $f$ is a flow function for $\N(\gamma)$.
For $u_1\in MU_1$, we have $\delta_1(u_1)=1$, therefore,
$f(\source,u_1)=\P(s_1,u_1)$. Analogously,
$f(\overline{u_2},\sink)=\gamma\P(s_2,u_2)$ for $u_2\in MU_2$.
Hence, $f$ is valid for $\N(\gamma)$, implying that $\gamma$ is valid for
$\N(\gamma)$. 

Now we show the \emph{if} direction.
Assume that there exists $\gamma>0$ and a valid flow $f$ for
$\N(\gamma)$.  The function $\delta_1$ is defined by: $ \delta_1(s)$
equals ${f(\source,s)}/{\P(s_1,s)}$ if $s\in \post(s_1)$ and $0$
otherwise. The function $\delta_2$ is defined similarly: $\delta_2(s)$
equals ${f(\overline s,\sink)}/{\gamma\P(s_2,s)}$ if $s\in
\post(s_2)$ and $0$ otherwise.  Let the sets $U_i$ and $V_i$ be
defined as required by Definition~\ref{def:dtmc_weak_simulation}.
It follows that
\begin{align*}
K_1 
&=\sum_{s\in
  U_1}\delta_1(s)\P(s_1,s)=\sum_{s\in U_1}f(\source,s)=f(\source,U_1)\\
K_2
&=\sum_{s\in
  U_2}\delta_2(s)\P(s_2,s)=\sum_{s\in
  U_2}\frac{f(\overline{s},\sink)}{\gamma}
=\frac{f(\overline{U_2},\sink)}{\gamma}
\end{align*}
Since the amount of flow out of $\source$ is the same as the amount of flow
into $\sink$, we have ${K_1}/{K_2} = \gamma$.  Since $\emptyset \not=
MU_i\subseteq U_i$ for $i=1,2$, both of $K_1$ and $K_2$ are greater than
$0$.  We show that the Conditions~\ref{ws:stuttersimulates}a and
\ref{ws:stuttersimulates}b of Definition~\ref{def:dtmc_weak_simulation} are
satisfied. For $v_1\in V_1$, we have that $\delta_1(v_1)<1$ which implies
that $f(\source,v_1)<\P(s_1,v_1)$. Since $f$ is valid for $\N(\gamma)$, and
since the edge $(\source,v_1)$ is not saturated by $f$, it must hold that
$v_1\in PV_1$. Therefore, $v_1\mathrel{R}s_2$. Similarly, we can prove that
$s_1\mathrel{R}v_2$ for $v_2\in V_2$.

We define $\Delta(w,w')={f(w,\overline{w'})}/{K_1}$ for $w,w'\in
S$. Assume that $\Delta(w,w')>0$. Then, $f(w,\overline{w'})>0$, which
implies that $(w,\overline{w'})$ is an edge of $\N(\gamma)$, therefore,
$(w,w')\in R$. By the flow conservation rule, $f(\source,w)\ge
f(w,\overline{w'})>0$, implying that $\delta_1(w)>0$. By the definition of
$U_1$, we obtain that $w\in U_1$. Similarly, we can show that $w'\in
U_2$. Hence, the Condition~\ref{ws:wfgt0} is satisfied. To prove Condition
\ref{ws:wfsums}:
\[
\Delta(w,U_2) 
= \sum_{u_2\in U_2}\frac{f(w,\overline{u_2})}{K_1}
= \frac{f(w,\overline{U_2})}{K_1}
\overeq{(*)} \frac{f(\source, w)}{K_1} 
=\frac{\delta_1(w)\P(s_1,w)}{K_1}
\] \nopagebreak
where equality $(*)$ follows from the flow conservation rule.
Therefore, for $w\in S$ we have that $K_1\Delta(w,U_2)=
\P(s_1,w)\delta_1(w)$. Similarly, we can show
$K_2\Delta(U_1,w)=\P(s_2,w)\delta_2(w).$  Condition~\ref{ws:wfsums} is also satisfied. As $K_1>0$ and $K_2>0$, the
reachability condition holds trivially, hence, $s_1\wsrel_R s_2$.
\end{proof}

\begin{figure}[tbp]
  \begin{center}
    \psset{levelsep=1.5cm,arrows=->,nodesep=0pt,labelsep=0pt,unit=.8}
    \begin{pspicture}(0,-1)(7, 3)
      \rput(2,0){\circlenode{o1}{$o_1$}}
      \rput(5,0){\circlenode{o2}{$\overline{o_2}$}}
      \rput(0,.5){\circlenode{source}{$\source$}}
      \rput(5,1){\circlenode{o3}{$\overline{o_3}$}}
      \rput(7,.5){\circlenode{sink}{$\sink$}}
      \rput(2,2){\circlenode{v1}{$v_1$}}
      \rput(5,2){\circlenode{v2}{$\overline{v_2}$}}
      \rput(2,-1){\circlenode[fillstyle=solid,fillcolor=yellow]{u1}{$u_1$}}
      \rput(5,-1){\circlenode[fillstyle=solid,fillcolor=yellow]{u2}{$\overline{u_2}$}}
      \ncline{source}{o1} \naput{\small $\frac{1}{2}$}
      \ncline{source}{v1} \naput{\small $\frac{1}{4}$}
      \ncline{source}{u1} \nbput{$\frac{1}{4}$}
      \ncline{v2}{sink} \naput{$\frac{11}{8}$}
      \ncline{o2}{sink} \ncput{$\frac{1}{4}$}
      \ncline{o3}{sink} \naput[npos=.24]{$\frac{1}{8}$}
      \ncline{u2}{sink} \nbput{$\frac{1}{4}$}
      \ncline{v1}{v2}
      \ncline{v1}{o3}
      \ncline{o1}{v2}
      \ncline{o1}{o2}
      \ncline{u1}{u2}
    \end{pspicture}\hspace{2cm}
    \begin{pspicture}(0,-1)(7, 3)
      \rput(2,0){\circlenode{o1}{$o_1$}}
      \rput(5,0){\circlenode{o2}{$\overline{o_2}$}}
      \rput(0,1.5){\circlenode{source'}{$\source'$}}
      \rput(5,1){\circlenode{o3}{$\overline{o_3}$}}
      \rput(7,1.5){\circlenode{sink'}{$\sink'$}}
      \rput(2,2){\circlenode{v1}{$v_1$}}
      \rput(5,2){\circlenode{v2}{$\overline{v_2}$}}
      \rput(5,3){\circlenode{source}{$\source$}}
      \rput(2,3){\circlenode{sink}{$\sink$}}
      \rput(2,-1){\circlenode[fillstyle=solid,fillcolor=yellow]{u1}{$u_1$}}
      \rput(5,-1){\circlenode[fillstyle=solid,fillcolor=yellow]{u2}{$\overline{u_2}$}}
      \ncline{source}{v1} \ncput[npos=.75]{$\frac{1}{4}$}
      \ncline{source}{sink'} \naput{$\frac{3}{4}$}
      \ncline{source'}{u1} \nbput{$\frac{1}{4}$}
      \ncline{u2}{sink'} \nbput{$\frac{1}{4}$}
      \ncline{u1}{u2}
      \ncline{v1}{v2}
      \ncline{v1}{o3}
      \ncline{o1}{v2}
      \ncline{o1}{o2}
      \ncline{v2}{sink} \ncput[npos=.25]{$\frac{11}{8}$}
      \ncline{o2}{sink'} \ncput{$\frac{1}{4}$}
      \ncline{o3}{sink'} \naput{$\frac{1}{8}$}
      \ncline{sink}{source}
      \ncline{source'}{sink} \naput{$\frac{5}{8}$}
      \ncline{source'}{o1}   \naput{$\frac{1}{2}$}
    \end{pspicture}
  \end{center}
  \caption{Left: The network $\N(2)$ of the DTMC in Figure~\ref{fig:dtmc_partition};
                Right: The transformed network $\N_t(2)$ for $\N(2)$.}
  \label{fig:dtmc_network}
\end{figure}
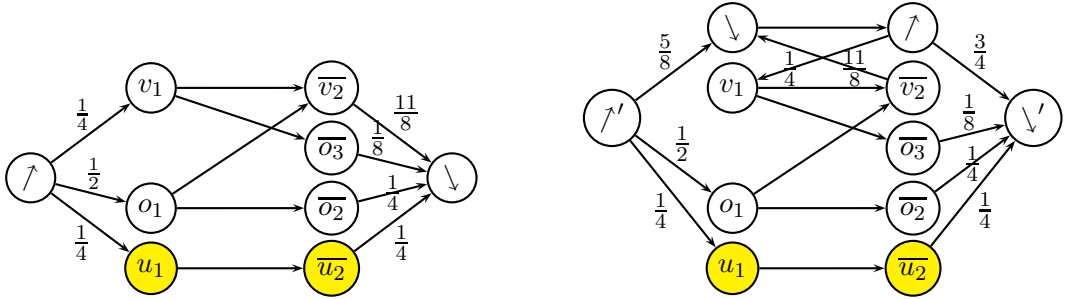

\begin{exa}
  Consider again Example~\ref{exa:dtmc_split} with the relation $R=
  \{(s_1,s_2),(s_1,v_2), (v_1,s_2)$,
  $(u_1,u_2),(o_1,o_2),(o_1,v_2),(v_1,o_3),(v_1,v_2),(o_2,o_1)\}$.
  The network $\N(2)$ is depicted on the left part of
  Figure~\ref{fig:dtmc_network}. Edges without numbers have capacity
  $\infty$. It is easy to see that $2$ is valid for $\N(2)$: the
  corresponding flow sends $\frac{1}{4}$ amount of flow along the path
  $\source,u_1,\overline{u_2},\sink$, $\frac{1}{4}$ amount of flow
  along the path $\source, o_1,\overline{o_2},\sink$, $\frac{1}{4}$
  amount of flow along the path $\source, o_1,\overline{v_2},\sink$,
  and $\frac{1}{8}$ amount of flow along the path
  $\source,v_1,\overline{o_3},\sink$.
\end{exa}

For a fixed $\gamma\in \Real_{>0}$, we now address the problem of checking
whether there exists a valid flow $f$ for $\N(\gamma)$.  This is a feasible
flow problem ($f$ has to saturate edges to $MU_1$ and from $MU_2$). As we
have discussed in Section~\ref{sec:flow}, it can be solved by applying a
simple transformation to the graph (in time $\O(\size{MU_1} +
\size{MU_2})$), solving the maximum flow problem for the transformed graph,
and checking whether the flow saturates all edges from the new source.

\begin{exa}
\label{exa:feasible}
Consider the network $\N(2)$ on the left part of
Figure~\ref{fig:dtmc_network}.  Applying the transformation for the
feasible flow problem described in Section~\ref{sec:flow}, we get the
transformed network $\N_t(2)$ depicted on the right part of
Figure~\ref{fig:dtmc_network}.  It is easy to see that the maximum
flow $h$ for $\N_t(2)$ has value $\frac{11}{8}$. Namely: It sends
$\frac{1}{4}$ amount of flow along the path
$\source',u_1,\overline{u_2},\sink'$, $\frac{1}{4}$ amount of flow
along the path $\source',o_1,\overline{o_2},\sink'$, $\frac{1}{4}$
amount of flow along
$\source',o_1,\overline{v_2},\sink,\source,\sink'$, $\frac{1}{8}$
amount of flow along
$\source',\sink,\source,v_1,\overline{o_3},\sink'$, and $\frac{1}{2}$
amount of flow along $\source',\sink,\source,\sink'$.  Thus, it uses
all capacities of edges from $\source'$. This implies that $2$ is
valid for the network $\N(2)$.
\end{exa}

\subsubsection{Breakpoints}
\label{sec:breakpoints}
Consider the pair $(s_1,s_2)\in R$. Assume the conditions of
Lemma~\ref{valid_gamma} are satisfied, thus, to check whether $s_1
\wsrel_R s_2$ it is equivalent to check whether a valid $\gamma$ for
$\N(\gamma)$ exits.  We show that only a finite possible $\gamma$, called
breakpoints, need to be considered. The breakpoints can be computed using a
variant of the parametric maximum flow algorithm.  Then, $s_1\wsrel_R s_2$
if and only if for some breakpoint it holds that the maximum flow for the
corresponding transformed network $\N_t(\gamma)$ has a large enough value.

Let $\size{V}$ denote the number
of vertices of $\mathcal{N}(\gamma)$. Let $\kappa(\gamma)$ denote the
\emph{minimum cut capacity function} of the parameter $\gamma$, which
is the capacity of a minimum cut of $\N(\gamma)$ as a function of $\gamma$.
The capacity of a minimum cut equals the value of a maximum flow.  If the
edge capacities in the network are linear functions of $\gamma$,
$\kappa(\gamma)$ is a piecewise-linear concave function with at most
$\size{V}-2$ breakpoints~\cite{GalloGT89}, i.\,e., points where the slope
$\frac{d\kappa}{d\gamma}$ changes.  The $\size{V}-1$ or fewer line segments
forming the graph of $\kappa(\gamma)$ correspond to $\size{V}-1$ or fewer
distinct minimal cuts.  The minimum cut can be chosen as the same on a
single linear piece of $\kappa(\gamma)$, and at breakpoints certain edges
become saturated or unsaturated.  The capacity of a minimum cut for some
$\gamma^*$ gives an equation that contributes a line segment to the
function $\kappa(\gamma)$ at $\gamma=\gamma^*$. Moreover, this line segment
connects the two points $(\gamma_1,\kappa(\gamma_1))$ and
$(\gamma_2,\kappa(\gamma_2))$, where $\gamma_1,\gamma_2$ are the nearest
breakpoints to the left and right, respectively.

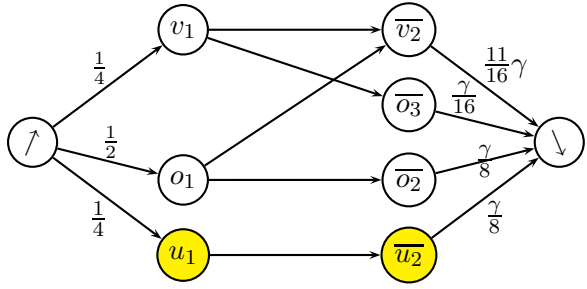
\begin{figure}[t]
  \begin{center}
    \psset{levelsep=1.5cm,arrows=->,nodesep=0pt,labelsep=0pt}
    \begin{pspicture}(0,-1)(7, 2)
      \rput(2,0){\circlenode{o1}{$o_1$}}
      \rput(5,0){\circlenode{o2}{$\overline{o_2}$}}
      \rput(0,.5){\circlenode{source}{$\source$}}
      \rput(5,1){\circlenode{o3}{$\overline{o_3}$}}
      \rput(7,.5){\circlenode{sink}{$\sink$}}
      \rput(2,2){\circlenode{v1}{$v_1$}}
      \rput(5,2){\circlenode{v2}{$\overline{v_2}$}}
      \rput(2,-1){\circlenode[fillstyle=solid,fillcolor=yellow]{u1}{$u_1$}}
      \rput(5,-1){\circlenode[fillstyle=solid,fillcolor=yellow]{u2}{$\overline{u_2}$}}
      \ncline{source}{o1} \naput{\small $\frac{1}{2}$}
      \ncline{source}{v1} \naput{\small $\frac{1}{4}$}
      \ncline{source}{u1} \nbput{$\frac{1}{4}$}
      \ncline{v2}{sink} \naput{$\frac{11}{16}\gamma$}
      \ncline{o2}{sink} \ncput{$\frac{\gamma}{8}$}
      \ncline{o3}{sink} \naput[npos=.24]{$\frac{\gamma}{16}$}
      \ncline{u2}{sink} \nbput{$\frac{\gamma}{8}$}
      \ncline{v1}{v2}
      \ncline{v1}{o3}
      \ncline{o1}{v2}
      \ncline{o1}{o2}
      \ncline{u1}{u2}
    \end{pspicture}
  \end{center}
  \caption{The network $\N(\gamma)$ of the DTMC in
Figure~\ref{fig:dtmc_partition} in Example~\ref{exa:dtmc_split}.}
  \label{fig:dtmc_network_gamma}
\end{figure}
\begin{exa}
\label{exa:one_valid_breakpoint}
Consider the DTMC in Figure~\ref{fig:dtmc_partition}, together with the relation $R=
\{(s_1,s_2)$, $(s_1,v_2), (v_1,s_2), (u_1,u_2),(o_1,o_2),(o_1,v_2),(v_1,o_3),
(v_1,v_2),(o_2,o_1)\}$.  The network $\N(\gamma)$ for the pair $(s_1,s_2)$
is depicted in Figure~\ref{fig:dtmc_network_gamma}.  

There are two breakpoints, namely $\frac{6}{7}$ and $2$. For $\gamma\le
\frac{6}{7}$, all edges leading to the sink can be saturated. This can be
established by the following flow function $f$: sending $\frac{\gamma}{8}$
amount of flow along the path $\source,u_1,\overline{u_2},\sink$,
$\frac{\gamma}{8}$ amount of flow along the path
$\source,o_1,\overline{o_2},\sink$, $\frac{11\gamma}{24}$ amount of flow
along the path $\source,o_1,\overline{v_2},\sink$, $\frac{\gamma}{16}$
amount of flow over $\source,v_1,\overline{o_3},\sink$,
$\frac{11\gamma}{48}$ amount of flow along the path
$\source,v_1,\overline{v_2},\sink$.  The amount of flow out of node $o_1$,
denoted by $f(o_1,\overline{S})$, is $\frac{7\gamma}{12}$. Given that
$\gamma\le \frac{6}{7}$, we have that $f(o_1,\overline{S})\le
\frac{1}{2}$. Similarly, consider the amount of flow out of node $v_1$,
which is denoted by $f(v_1,\overline{S})$, is $\frac{7\gamma}{24}$ which
implies that $f(v_1,\overline{S})\le
\frac{1}{4}$. The maximum flow  thus has value $\size{f} =  \frac{\gamma}{8} +
\frac{\gamma}{8} + \frac{11\gamma}{24} +  \frac{\gamma}{16} +
\frac{11\gamma}{48}=\gamma$. Thus the value of the maximum flow, or
equivalently the value of the minimum cut, for $\gamma\le \frac{6}{7}$ is
$\kappa(\gamma)=\gamma$.

Observe that for $\gamma=\frac{6}{7}$, the edges to $v_1$ and $o_1$ are
saturated, i.e., we have used full capacities of the edge $(\source,v_1)$
and $(\source,o_1)$. Thus, by a greater value of $\gamma$, although the
capacities
$\capacity(\{\overline{v_2},\overline{o_2},\overline{o_3}\},\sink)$
increase (become greater than $\frac{3}{4}$), no additional flow can be
sent through $\{v_1,o_1\}$. For the other breakpoint $2$, we observe that
for a value of $\gamma\le 2$, we can still send $\frac{\gamma}{8}$ through
the path $\source,u_1,u_2,\sink$, but for $\gamma> 2$, the edge to $u_1$
keeps saturated, thus the amount of flow sent through this path does not
increase any more. Thus, for $\gamma\in[\frac{6}{7},2]$, the maximum value,
or the value of the minimum cut, is $\frac{3}{4}+\frac{\gamma}{8}$. The
first term $\frac{3}{4}$ corresponds to the amount of flow through $v_1$
and $o_1$. The breakpoint $\frac{6}{7}$ is not valid since the edge to
$u_1$ can not be saturated. As discussed in Example~\ref{exa:feasible}, the
breakpoint $2$ is valid. The curve for $\kappa(\gamma)$ is depicted in
Figure~\ref{fig:curvebreakpoints}.
\end{exa}
\begin{figure}[tbp]
\begin{center}
  \includegraphics[angle=-90,width=60mm]{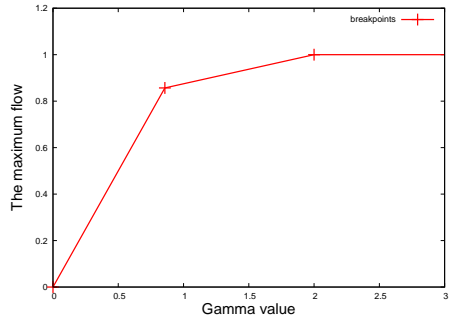}
\end{center}
\caption{The value of the maximum flow, or equivalently the value of the
minimum cut, as a function of $\gamma$ for the network in
Figure~\ref{fig:dtmc_network_gamma}.}
\label{fig:curvebreakpoints}
\end{figure}

In the following lemma we show that if there is any valid $\gamma$, then at
least one breakpoint is valid.

\begin{lem}{\label{parameter_flow}} 
Assume  $\gamma^*\in(\gamma_1,\gamma_2)$ where $\gamma_1, \gamma_2$
are two subsequent breakpoints of $\kappa(\gamma)$, or $\gamma_1=0$ 
and $\gamma_2$ is the first breakpoint, or $\gamma_1$ 
is the last breakpoint and $\gamma_2=\infty$. Assume $\gamma^*$ is
valid for $\N(\gamma^*)$, then, $\gamma$ is valid for $\N(\gamma)$ for
all $\gamma\in[\gamma_1,\gamma_2]$.
\end{lem}
\begin{proof}
  Consider the network $\N(\gamma^*)$. Assume that the maximum flow
  $f_{\gamma^*}$ is a valid maximum flow for $\N(\gamma^*)$.
  
  Assume first $\gamma' \in (\gamma^*, \gamma_2]$.  We use the augmenting
  path algorithm~\cite{AMO93} to obtain a maximum flow $f^*$ in the
  residual network $\N_{f_{\gamma^*}}(\gamma')$, requiring that the
  augmenting path contains no cycles, which is a harmless
  restriction. Then, $f_{\gamma'}:=f_{\gamma^*}+f^*$ is a maximum flow in
  $\N(\gamma')$.  Since $f_{\gamma^*}$ saturates edges from $\source$ to
  $MU_1$, $f_{\gamma'}$ saturates edges from $\source$ to $MU_1$ as well ,
  as flow along an augmenting path without cycles does not un-saturate
  edges to $MU_1$.  We choose the minimum cut $(X,X')$ for $\N(\gamma^*)$
  with respect to $f_{\gamma^*}$ such that $\overline{MU_2}\cap
  X'=\emptyset$, or equivalently $\overline{MU_2}\subseteq X$. This is
  possible since $f_{\gamma^*}$ saturates all edges $(\overline u_2,\sink)$
  with $u_2\in MU_2$.  The minimum cut for $f_{\gamma'}$, then, can also be
  chosen as $(X,X')$, as $(\gamma', \kappa(\gamma'))$ lies on the same line
  segment as $(\gamma^*,\kappa(\gamma^*))$. Hence, $f_{\gamma'}$ saturates
  the edges from $\overline{MU_2}$ to $\sink$, which indicates that
  $f_{\gamma'}$ is valid for $\N(\gamma')$. Therefore, $\gamma'$ is valid
  for $\N(\gamma')$ for $\gamma'\in (\gamma^*,\gamma_2]$.

  Now let $\gamma'\in [\gamma_1,\gamma^*)$. For the valid maximum flow
  $f_{\gamma^*}$ we select the minimal cut $(X,X')$ for $\N(\gamma^*)$ such
  that $MU_1\cap X=\emptyset$. Let $d$ denote a valid distance function
  corresponding to $f_{\gamma^*}$. We replace $f_{\gamma^*}(\overline
  v,\sink)$ by $\min\{f_{\gamma^*}(\overline v,
  \sink),\capacity_{\gamma'}(\overline v, \sink)\}$ where
  $\capacity_{\gamma'}$ is the capacity function of $\N(\gamma')$. The
  modified flow is a preflow for the network $\N(\gamma')$. Moreover, $d$
  stays a valid distance function as no new residual edges are
  introduced. Then, we apply the preflow algorithm to get a maximum flow
  $f_{\gamma'}$ for the network $\N(\gamma')$. Since no flow is pushed
  back from the sink, edges from $\overline{MU_2}$ to $\sink$ are kept
  saturated.  Since $(\gamma^*,\kappa(\gamma^*))$ and
  $(\gamma',\kappa(\gamma'))$ are on the same line segment, the minimal cut
  for $f_{\gamma'}$ can also be chosen as $(X,X')$, which indicates that
  $f_{\gamma'}$ saturates all edges to $MU_1$. This implies that $\gamma'$
  is valid for $\N(\gamma')$ for $\gamma'\in [\gamma_1,\gamma^*)$.
\end{proof}

In Example~\ref{exa:one_valid_breakpoint}, only one breakpoint is valid. In
the following example we show that it is in general possible that more than
one breakpoint is valid.
\begin{figure}[tbp]
  \begin{center}
    \psset{arrows=->,nodesep=0pt,labelsep=1pt}
    \begin{pspicture}(0,-1)(7, 3)
      \rput(2,0){\circlenode{u3}{$u_3$}}
      \rput(5,0){\circlenode{u4}{$\overline{u_4}$}}
      \rput(0,1){\circlenode{source}{$\source$}}
      \rput(2,1){\circlenode{u5}{$u_5$}}
      \rput(5,1){\circlenode{u6}{$\overline{u_6}$}}
      \rput(7,1){\circlenode{sink}{$\sink$}}
      \rput(2,2){\circlenode{o1}{$o_1$}}
      \rput(5,2){\circlenode{o2}{$\overline{o_2}$}}
      \rput(2,3){\circlenode{v1}{$v_1$}}
      \rput(5,3){\circlenode{v2}{$\overline{v_2}$}}
      \rput(2,-1){\circlenode{u1}{$u_1$}}
      \rput(5,-1){\circlenode{u2}{$\overline{u_2}$}}
      \ncline{source}{v1} \naput{$\frac{1}{4}$}
      \ncline{source}{o1} \ncput{$\frac{1}{4}$}
      \ncline{source}{u1} \nbput{$\frac{1}{8}$}
      \ncline{source}{u3}   \ncput{$\frac{1}{8}$}
      \ncline{source}{u5}   \ncput{$\frac{1}{4}$}
      \ncline{v2}{sink} \naput{$\frac{\gamma}{4}$}
      \ncline{o2}{sink} \ncput{$\frac{\gamma}{4}$}
      \ncline{u4}{sink} \ncput{$\frac{\gamma}{8}$}
      \ncline{u6}{sink} \ncput{$\frac{\gamma}{8}$}
      \ncline{u2}{sink} \nbput{$\frac{\gamma}{4}$}
      \ncline{u1}{u2}
      \ncline{o1}{o2}
      \ncline{o1}{v2}
      \ncline{u3}{u4}
      \ncline{u5}{u6}
      \ncline{v1}{v2}
    \end{pspicture}
  \end{center}
  \caption{A network in which more than one breakpoint is valid.}
  \label{fig:flow_network_breakpoints}
\end{figure}
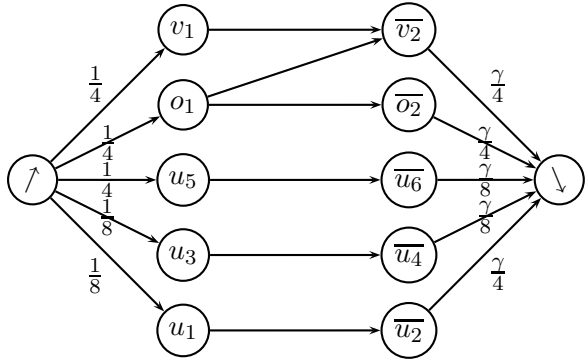
\begin{exa}
Consider the network depicted in Figure~\ref{fig:flow_network_breakpoints}.
By a similar analysis as Example~\ref{exa:one_valid_breakpoint}, we can
compute that there are three breakpoints $\frac{1}{2}$, $1$ and
$2$. Assuming that $MU_1=\{o_1\}$ and $MU_2=\{o_2\}$, we show that all
$\gamma\in[\frac{1}{2},1]$ are valid. We send $\frac{\gamma}{4}$ amount of flow
along the path $\source,o_1,\overline{o_2},\sink$, and
$\frac{1}{4}-\frac{\gamma}{4}$ amount of flow along the path
$\source,o_1,\overline{v_2},\sink$. If $\gamma\in[\frac{1}{2},1]$, then $0
\leq \frac{1}{4}-\frac{\gamma}{4} \leq
\frac{\gamma}{4}$ implying that the flow on edge $(\overline{v_2},\sink)$
satisfies the capacity constraints. Obviously this flow is feasible, and all
$\gamma\in[\frac{1}{2},1]$ are valid for $\N(\gamma)$.
\end{exa}

As we would expect now, it is sufficient to consider only the
breakpoints of $\N(\gamma)$:
\begin{lem}\label{valid_breakpoint}
  There exists a valid $\gamma$ for $\N(\gamma)$ iff one of the
  breakpoints of $\N(\gamma)$ is valid.
\end{lem}
\begin{proof}
If there exists a valid $\gamma$ for $\N(\gamma)$,
Lemma~\ref{parameter_flow} guarantees that one of the breakpoints of
$\kappa(\gamma)$ is valid.  The other direction is trivial.
\end{proof}

For a given breakpoint, we need to solve one feasible flow problem to
check whether it is valid.  In the network $\N(\gamma)$ the capacities
of the edges leading to the sink are increasing functions of a
real-valued parameter $\gamma$.  If we reverse $\N(\gamma)$, we get a
parametric network that satisfies the conditions in~\cite{GalloGT89}:
The capacities emanating from $\source$ are non-decreasing functions
of $\gamma$.  So we can apply the \emph{breakpoint
  algorithm}~\cite{GalloGT89} to obtain all of the breakpoints of
$\N(\gamma)$.

\subsubsection{The Algorithm for DTMCs}
\label{sec:weak_alg_dtmc}
Let $\M$ be a DTMC and let $\prog{SimRel}_w(\D)$ denote the weak
simulation algorithm, which is obtained by replacing
line~\ref{algsim:checksimrel} of $\prog{SimRel}_s(\D)$ in
Algorithm~\ref{fig:simfps} by: \textbf{if} $s_1\wsrel_R s_2$.  The
condition $s_1 \wsrel_R s_2$ is checked in $\prog{Ws}(\D,s_1,s_2,R)$,
shown as Algorithm~\ref{fig:ws}.

\begin{algorithm}[tbp]
  \caption{Algorithm to check whether $s_1\wsrel_R s_2$.}
  \label{fig:ws}
    \Procname{$\prog{Ws}(\D,s_1,s_2,R)$}
  \begin{algorithmic}[1]
        \IF{$\post(s_1)\subseteq R^{-1}(s_2)$}   \label{algws:onlystutter}
                \STATE  \textbf{return} $\mathbf{true}$\label{algws:onlystutter.end}
        \ENDIF
        \IF{$\post(s_2)\subseteq R(s_1)$}   \label{algws:testposts2simss1}
                \STATE  $U_1\gets \{s_1'\in \post(s_1)\mid s_1'\not\in R^{-1}(s_2)\}$
                \STATE \textbf{return} 
                ($\forall u_1\in U_1.\ \exists
                s\in \post(reach(s_2)\cap R(s_1)).\ s\in R(u_1)$)   \label{algws:testreachcond.end}
        \ENDIF
        \STATE  Compute all of the breakpoints $b_1<b_2< \ldots <b_j$ of
     $\N(\gamma)$\label{algwsf:computebreakpoints}
        \STATE \textbf{return} ($\exists i \in
    \{1,\ldots, j\}.\ b_i \mbox{ is valid for }\N(b_i)$) \label{algwsf:returnbiisvalid}
  \end{algorithmic}
\end{algorithm}

The first part of the algorithm is the preprocessing part.
Line~\ref{algws:onlystutter} tests for the case that $s_1$ could perform
only \emph{stutter} steps with respect to the current relation $R$. If
line~\ref{algws:testreachcond.end} is reached, $s_1$ has at least one
\emph{visible} step, and all successors of $s_2$ can simulate $s_1$ up to
the current relation $R$.  In this case we need to check the reachability
Condition~\ref{ws:reachability} of
Definition~\ref{def:dtmc_weak_simulation}, which is performed in
line~\ref{algws:testreachcond.end}. Recall that $reach(s)$ denotes the set
of states that are reachable from $s$ with positive probability.  If the
algorithm does not terminate in the preprocessing part, the breakpoints of
the network $\N(\gamma)$ are computed. Then, corresponding to
Lemma~\ref{valid_breakpoint}, we check whether one of the breakpoints is
valid.  We show the correctness of the algorithm $\prog{Ws}$:
\begin{lem}
\label{lem:pre}
The algorithm $\prog{Ws}(\D,s_1,s_2,R)$ returns true iff $s_1\wsrel_R
s_2$.
\end{lem}
\proof
  We first show  the \emph{only if} direction.  Assume that
  $\prog{Ws}(\D,s_1,s_2,R)$ returns true.  We consider three possible
  cases:
  \begin{enumerate}[$\bullet$]
  \item The algorithm returns true at
    line~\ref{algws:onlystutter.end}.  For this case we have that
    $\post(s_1)\subseteq R^{-1}(s_2)$. We choose $U_1=\emptyset,
    V_1=\post(s_1)$, $U_2=\post(s_2)$ and $V_2=\emptyset$ to fulfill
    the conditions in Definition~\ref{def:dtmc_weak_simulation}.
    Hence, $s_1\wsrel_R s_2$.
  \item The algorithm returns true at
    line~\ref{algws:testreachcond.end}. If the algorithm reaches
    line~\ref{algws:testreachcond.end}, the following conditions hold:
    there exists a state $s_1'\in \post(s_1)$ such that $s_1'\not\in
    R^{-1}(s_2)$ (line~\ref{algws:onlystutter}), and
    $\post(s_2)\subseteq R(s_1)$ (line~\ref{algws:testposts2simss1}).
    Let $U_1= \{s_1'\in \post(s_1)\mid s_1'\not\in R^{-1}(s_2)\}$, and
define $\delta_i$ by: $\delta_1(s)=1$ if $s\in U_1$, and $0$ otherwise,
$\delta_2(s)=0$ for all $s\in S$. By construction, to show $s_1\wsrel_R
s_2$ we only need  to show the reachability condition.
    Since the algorithm returns true at
    line~\ref{algws:testreachcond.end}, it holds that for each $u_1\in
    U_1$, there exists $s\in \post(reach(s_2)\cap R(s_1))$ such that
    $s\in R(u_1)$. This is exactly the reachability
    condition required by weak simulation up to $R$, thus $s_1\wsrel_R
    s_2$.
  \item The algorithm returns true at
    line~\ref{algwsf:returnbiisvalid}. Thus, there exists breakpoint
    $b_i$ which is valid for $\N(b_i)$. Then, there exists a state
    $s_1'\in \post(s_1)$ such that $s_1'\not\in R^{-1}(s_2)$, and
    $s_2'\in \post(s_2)$ such that $s_2'\not\in R(s_1)$.  By Lemma~\ref{valid_gamma}, we
    have that $s_1\wsrel_R s_2$.
  \end{enumerate}
  Now we show the \emph{if} direction. Assume that $\prog{Ws}$ returns
  false. It is sufficient to show that $s_1\not\wsrel_R s_2$. We consider two cases:
  \begin{enumerate}[$\bullet$]
  \item The algorithm returns false at
    line~\ref{algws:testreachcond.end}. This implies that there exists
    a state $s_1'\in \post(s_1)$ such that $s_1'\not\in R^{-1}(s_2)$
    (line~\ref{algws:onlystutter}), and $\post(s_2)\subseteq R(s_1)$
    (line~\ref{algws:testposts2simss1}). All states  $s_1'\in \post(s_1)$ with
$s_1'\not\in R^{-1}(s_2)$ must be put into $U_1$. However, since the algorithm
    returns false at line~\ref{algws:testreachcond.end}, it holds that
    there exists a state $u_1\in U_1$, such that there does not exist $s\in
    \post(reach(s_2)\cap R(s_1))$ with $s\in R(u_1)$. Thus the
    reachability condition of
    Definition~\ref{def:dtmc_weak_simulation} is violated which
    implies that $s_1\not\wsrel_R s_2$.
  \item The algorithm returns false at
    line~\ref{algwsf:returnbiisvalid}. Then, there exists a state
    $s_1'\in \post(s_1)$ such that $s_1'\not\in R^{-1}(s_2)$, and
    $s_2'\in \post(s_2)$ such that $s_2'\not\in R(s_1)$. Moreover, for
    all breakpoints $b$ of $\N(\gamma)$, $b$ is not valid for $\N(b)$.
    By Lemma~\ref{valid_breakpoint}, there does not exist a valid
    $\gamma$ for $\N(\gamma)$. By Lemma~\ref{valid_gamma}, we have
    that  $s_1\not\wsrel_R s_2$.\qed
  \end{enumerate}
Now we state the correctness of the algorithm $\prog{SimRel}_w$ for DTMCs:

\begin{thm}
\label{thm:fpsws_correct}
  If $\prog{SimRel}_w(\D)$ terminates, the returned relation $R$ equals
  $\wsrel_\D$.
\end{thm}
\begin{proof}
  The proof follows exactly the lines of the proof of
  Theorem~\ref{thm:correctness_basic}.  Let iteration $k$ be the last
  iteration of $\prog{Ws}$.  Then, by Lemma~\ref{lem:pre}, for each pair
  $(s_1,s_2) \in R_k$, it holds that $s_2$ weakly simulates $s_1$ up to
  $R_k$, so $R_k$ is a weak simulation.  On the other hand, one can prove
  by induction that each $R_i$ is coarser than $\wsrel$.
\end{proof}

\paragraph{Complexity.}
For $(s_1,s_2)\in R$, we have shown that to check whether $s_1\wsrel_R s_2$
we could first compute the breakpoints of $\N(\gamma)$, then solve
$\O(\size{V})$ many maximum flow problems. To achieve a better bound, we
first prove  that applying a binary search method over the breakpoints, we
only need to consider $\O(\log
\size{V})$ breakpoints, and thus solve $\O(\log
\size{V})$ maximum flow problems. 

Assume that the sets $MU_i$ and $PV_i$ for $i=1,2$ are constructed as before for
$\N(\gamma)$. Recall that a flow function $f$ of $\N(\gamma)$ is valid for
$\N(\gamma)$ iff $f$ saturates all edges $(\source, u_1)$ with $u_1\in
MU_1$ and all edges $(\overline{u_2},\sink)$ with $u_2\in MU_2$. If $f$ is
also a maximum flow, we say that $f$ is a valid maximum flow of
$\N(\gamma)$. We first reformulate Lemma~\ref{valid_gamma} using maximum
flow.

\begin{lem}{\label{reduced_to_exists_flow}} 
  There exists a valid flow $f$ for $\N(\gamma)$ iff
  there exists a valid maximum  flow $f_m$ for $\N(\gamma)$.
\end{lem}
\begin{proof}
  Assume there exists a valid flow $f$ for $\N(\gamma)$.  We let
  $\N_f(\gamma)$ denote the residual network.  We use the augmenting
  path algorithm to get a maximum flow $f'$ in the residual network
  $\N_f(\gamma)$. Assume that the augmenting path contains no cycles,
  which is a harmless restriction. Let $f_m=f+f'$.  Obviously, $f_m$
  is a maximum flow for $\N(\gamma)$, and it saturates all of the
  edges saturated by $f$. Hence, $f_m$ is valid for $\N(\gamma)$.  The
  other direction is simple, since a valid maximum flow is also a
  valid flow for $\N(\gamma)$.
\end{proof}

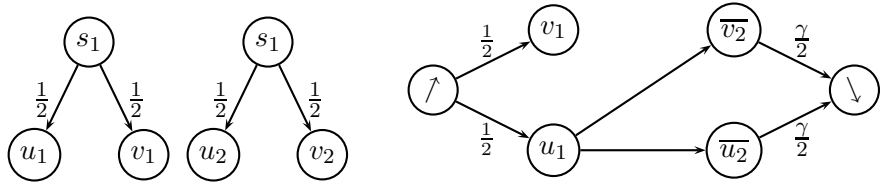
\begin{figure}[tbp]
\begin{center}
\psset{levelsep=1.5cm,arrows=->,labelsep=0pt}
\pstree{\Tcircle{$s_1$}}{%
  \Tcircle{$u_1$}  \tlput{$\frac{1}{2}$}
  \Tcircle{$v_1$} \trput{$\frac{1}{2}$}
}\hspace{.1cm}
\pstree{\Tcircle{$s_1$}}{%
  \Tcircle{$u_2$}  \tlput{$\frac{1}{2}$}
  \Tcircle{$v_2$} \trput{$\frac{1}{2}$}
}\hspace{1cm}
\psset{levelsep=1.5cm,arrows=->,labelsep=0pt,unit=.8}
    \begin{pspicture}(0,1.8)(7, 2)
      \rput(2,0){\circlenode{u1}{$u_1$}}
      \rput(5,0){\circlenode{u2}{$\overline{u_2}$}}
      \rput(0,1){\circlenode{source}{$\source$}}
      \rput(7,1){\circlenode{sink}{$\sink$}}
      \rput(2,2){\circlenode{v1}{$v_1$}}
      \rput(5,2){\circlenode{v2}{$\overline{v_2}$}}
      \ncline{source}{u1} \nbput{\small $\frac{1}{2}$}
      \ncline{source}{v1} \naput{\small $\frac{1}{2}$}
      \ncline{v2}{sink} \naput{$\frac{\gamma}{2}$}
      \ncline{u2}{sink} \nbput{$\frac{\gamma}{2}$}
      \ncline{u1}{v2}
      \ncline{u1}{u2}
    \end{pspicture}
  \end{center}
  \caption{A simple DTMC for illustrating that not all maximum flows
    are valid.}
  \label{fig:dtmc_valid}
\end{figure}

We first discuss  how to get a valid maximum flow provided that $\gamma$ is
valid. Observe that even if $\gamma$ is valid for $\N(\gamma)$, not all
maximum flows for $\N(\gamma)$ are necessarily valid. Consider the DTMC on
the left part of Figure~\ref{fig:dtmc_valid}.  Assume that the relation $R$
is given by $\{(s_1,s_2),$ $(s_1,v_2),$ $(v_1,s_2),$ $(u_1,u_2),$
$(u_1,v_2)\}$ and consider the pair $(s_1,s_2)$. Thus, we have that
$PV_1=\{v_1\}, MU_1=\{u_1\}, PV_2=\{v_2\}, MU_2=\{u_2\}$.  The network
$\N(\gamma)$ is depicted on the right part of
Figure~\ref{fig:dtmc_valid}. The maximum flow $f$ for $\N(1)$ has
value $0.5$. If $f$ contains positive sub-flow along the path
$\source,u_1,\overline{v_2},\sink$, it does not saturate the edge
$(\overline{u_2},\sink)$. On the contrary, the flow along the single path
$\source, u_1, \overline{u_2},\sink$ with value $0.5$ would be a valid
maximum flow.
This example gives us the intuition to use the augmenting path through
edges between $MU_1$ and $\overline{MU_2}$ as much as possible. For this
purpose we define a cost function $\cost$ from edges in $\N(\gamma)$ to real
numbers as follows: $\cost(u_1,\overline{u_2})=2$ for $u_1\in MU_1$ and
$u_2\in MU_2$, $\cost(u_1,\overline{v_2})=1$ for $u_1\in MU_1$ and $v_2\in
PV_2$, $\cost(v_1,\overline{u_2})=1$ for $v_1\in PV_1$ and $u_2\in MU_2$,
$\cost(s,s')=0$ otherwise. The costs of edges starting from source, or
ending at sink, or in $PV_1\times \overline{PV_2}$ are $0$. The cost of a
flow $f$ is defined by $\cost(f)=\sum_{e\in E}f(e)\cost(e)$. By definition
of the cost function, we have the property $\cost(f)=f(\source, MU_1)+
f(\overline{MU_2},\sink)$, i.e., the cost equals the sum of the amount of
flow from $\source$ into $MU_1$ and from $\overline{MU_2}$ into $\sink$.

\begin{lem}{\label{reduced_to_max_cost_flow}}
Assume that $\gamma>0$ is valid for $\N(\gamma)$. Let $f_\gamma$
denote a maximum flow over $\N(\gamma)$  
with maximum cost. Then, $f_\gamma$ is valid for $\N(\gamma)$.
\end{lem}
\begin{proof}
  By Lemma~\ref{reduced_to_exists_flow}, provided $\gamma$ is valid
  for $\N(\gamma)$, there exists a valid maximum flow function $g$ for
  $\N(\gamma)$. Since $g$ saturates edges to $MU_1$ and from $MU_2$,
  obviously, $\cost(g) = \P(s_1,MU_1) +
  \gamma\P(s_2,MU_2)$.  Assume that $f_\gamma$ is not valid,
  which indicates that $f_\gamma$ does not saturate an edge
  $(\source,u_1)$ with $u_1\in MU_1$ or an edge
  $(\overline{u_2},\sink)$ with $u_2\in \overline{MU_2}$. Then,
  $\cost(f_\gamma) = f(\source,MU_1)+f(\overline{MU_2},\sink)<
  \P(s_1,MU_1) + \gamma\P(s_2,MU_2) = \cost(g)$. This
  contradicts the assumption that $f_\gamma$ has maximum cost.
\end{proof}

Let $\N_U(\gamma)$ denote the subnetwork of $\N(\gamma)$ where the set
of vertices is restricted to $MU_1, \overline{MU_2}$ and $\{\source,\sink\}$. 
The following lemma discusses how to construct a maximum flow with
maximum cost.

\begin{lem}{\label{find_max_cost_flow}}
Assume that $f^*$ is an arbitrary maximum flow of $\N_U(\gamma)$ and
$\tilde{f}$ is an arbitrary maximum flow in the residual network
$\N_{f^*}(\gamma)$ with the residual edges from $MU_1$ back to $\source$
removed, as well as the residual edges from $\sink$ back to
$\overline{MU_2}$.  Then $f_\gamma=f^*+\tilde{f}$ is a maximum flow over
$\N(\gamma)$ with maximum cost.
\end{lem}
\begin{proof}
  Recall that the cost of $f_\gamma$ is equal to
  $\cost(f_\gamma)=f_\gamma(\source, MU_1)+
  f_\gamma(\overline{MU_2},\sink)$. Assume that $\cost(f_\gamma)$ is not
  maximal for the sake of contradiction. Let $f$ be a maximum flow such
  that $\cost(f_\gamma)<\cost(f)$. Without loss of generality, we assume that
  $f_\gamma(\source, MU_1)<f(\source, MU_1)$. It holds that
  $f_\gamma=f^*+\tilde{f}$ where $f^*$ is a maximum flow of $\N_U(\gamma)$,
  and $\tilde{f}$ is a maximum flow in the residual network
  $\N_{f^*}(\gamma)$ with the residual edges from $MU_1$ back to $\source$
  removed, as well as the residual edges from $\sink$ back to
  $\overline{MU_2}$. On the one hand, $f^*$ sends as much  flow as possible
  along $MU_1$ in $\N_U(\gamma)$. Since in the residual network
  $\N_{f^*}(\gamma)$ edges from $MU_1$ back to $\source$ are removed, this
  guarantees that no flow can be sent back to $\source$ from $MU_1$. On
  the other hand, $\tilde{f}$ sends as much flow  as possible from $MU_1$
  to $\overline{PV_2}$ (and also from $PV_1$ to $\overline{MU_2}$) in
  $\N_{f^*}(\gamma)$. Thus, $f_\gamma(\source, MU_1)$ must be maximal which
  contradicts the assumption $f_\gamma(\source, MU_1)<f(\source, MU_1)$.
\end{proof}

Assume that $\gamma^*$ is not valid. The following lemma determines,
provided a valid $\gamma$ exists, whether it is greater or smaller
than $\gamma^*$:

\begin{lem}{\label{binary_search}}
  Let $\gamma^*\in [0,\infty)$, and let $f$ be a maximum flow
  function with maximum cost for $\N(\gamma^*)$, as described in Lemma~\ref{find_max_cost_flow}. Then,

  \begin{enumerate}[\em(1)]
  \item If $f$ saturates all edges $(\source,u_1)$ with $u_1\in MU_1$ and
    $(\overline{u_2},\sink)$ with  $u_2\in
    MU_2$, $\gamma^*$ is valid for $\N(\gamma^*)$.
  \item Assume that  $\exists u_1\in MU_1$ such that $(\source,u_1)$ is not
    saturated by $f$, and all edges $(\overline{u_2},\sink)$ with $u_2\in
    MU_2$ are saturated by $f$. Then, $\gamma^*$ is not valid. If there
    exists a valid $\gamma$, $\gamma >\gamma^*$. 
  \item Assume that all edges $(\source,u_1)$ with $u_1\in MU_1$  are
    saturated by $f$, and $\exists u_2\in MU_2$ such that
    $(\overline{u_2},\sink)$ is not saturated by $f$. Then, $\gamma^*$ is not
    valid. If there exists a valid $\gamma$, $\gamma < \gamma^*$.
  \item Assume that  $\exists u_1\in MU_1$ and
    $\exists u_2\in MU_2$ such that $(\source,u_1)$ and
    $(\overline{u_2},\sink)$ are not saturated by $f$. Then, there does not
    exist a valid $\gamma$.
  \end{enumerate}
\end{lem}
\begin{proof} $1:$ Follows directly from the definition. $2:$ In
this case, $f(\source, u_1) < \P(s_1,u_1)$ for some $u_1\in
MU_1$. To saturate $(\source,u_1)$, without un-saturating other edges from
$\source$ to $MU_1$, we have to increase the capacities of edges leading to
$\sink$, thus increase $\gamma^*$. $3:$ Similar 
to the previous case. $4:$ Combining $2$ and $3$. 
\end{proof}

According to the above lemma, we can use the binary search method over the
breakpoints to check whether there exists a valid breakpoint for
$\N(\gamma)$. Since there are at most $\O(\size{V})$ breakpoints, we invoke
the maximum flow algorithm at most $\O(\log \size{V})$ times where
$\size{V}$ is the number of vertices of $\N(\gamma)$.

\begin{thm}
\label{thm:fpsws_complexity}
The algorithm $\prog{SimRel}_w(\D)$ runs in time $\O(m^2n^3)$ and in space
$\O(n^2)$. If the fanout $g$ is bounded by a constant, the time complexity
is $\O(n^5)$.
\end{thm}
\begin{proof}
  First, we consider a pair $(s_1,s_2)$ out of the current relation
  $R_i$. Look at a single call of $\prog{Ws}(\M,s_1,s_2,R_i)$.  By saving
  the current relation sets $R$ and $R^{-1}$ in a two dimensional array,
  the conditions $s\in R(s')$ or $s\in R^{-1}(s')$ can be checked in
  constant time. Hence, line~\ref{algws:onlystutter} takes time
  $\size{\post(s_1)}$.  To construct the set $reach(s)$ for a state $s$,
  BFS can be used, which has complexity $\O(m)$. The size of the set
  $reach(s)$ is bounded by $n$.  Therefore, we need
  $\O(\size{\post(s_1)}n)$ time at
  lines~\ref{algws:testposts2simss1}--\ref{algws:testreachcond.end}.

  The algorithm computes all breakpoints of the network $\N(\gamma)$
  (with respect to $s_1,s_2$ and $R$) using the breakpoint
  algorithm~\cite[p.~37--42]{GalloGT89}. Assume the set of vertices of
  $\N(\gamma)$ is partitioned into subsets $V_1$ and $V_2$ similar to the
  network described in Section~\ref{sec:strong_basic}.  The number of edges
  of the network is at most $\size{E}:=\size{V_1} \size{V_2} - 1$. Let
  $\size{V}:=\size{V_1}+\size{V_2}$, and, without loss of generality, we
  assume that $\size{V_1} \le \size{V_2}$. For our bipartite networks, the
  time complexity~\cite[p.~42]{GalloGT89} for computing the breakpoints is
  $\O(\size{V_1}\size{E} \log (\frac{\size{V_1}^2}{\size{E}} + 2))$ which
  can be simplified to $\O(\size{V_1}^2 \size{V_2})$.  Then, the binary
  search can be applied over all of the breakpoints to check whether at
  least one breakpoint is valid, for which we need to solve at most
  $\O(\log \size{V})$ many maximum flow problems.  For our network
  $\N(\gamma)$, the best known complexity\footnote{For a network $G=(V,E)$
  with small $\size{E}$, there are more efficient algorithms in~\cite{CM99}
  with complexity $\O(\size{V}^2\sqrt{\size{E}})$, and in~\cite{KRT94} with
  complexity $\O(\size{E}\size{V}+\size{V}^{2+\epsilon})$ where $\epsilon$
  is an arbitrary constant.  In our bipartite networks, however, $\size{E}$
  is in the order of $\size{V}^2$.  Hence, these complexities become
  $\O(\size{V}^3)$.  } of the maximum flow problem is $\O(\size{V}^3/\log
  \size{V})$~\cite{CheriyanHM90}.  As indicated in the proof of
  Lemma~\ref{smf_complexity}, the distance function is bounded by
  $4\size{V_1}$ for our bipartite network.  Applying this fact in the
  complexity analysis in~\cite{CheriyanHM90}, we get the corresponding
  complexity for computing maximum flow for bipartite networks
  $\O(\size{V_1}\size{V}^2/\log \size{V} )$. Hence, the complexity for the
  $\O(\log |V|)$-invocations of the maximum flow algorithm is bounded by
  $\O(\size{V_1}\size{V}^2)$. As $\size{V} \le 2\size{V_2}$, the complexity
  is equal to $\O(\size{V_1}\size{V_2}^2)$.  
Summing over all $(s_1,s_2)$ over
  all $R_i$, we get the overall complexity of $\prog{SimRel}_w(\D)$:
\begin{align}\label{eq:complexity_fps}
  & \sum_{i=1}^k \sum_{(s_1,s_2)\in R_i} (\size{\post(s_1)}+m +
  \size{\post(s_1)}n + \size{V_1}\size{V_2}^2)
  \le 4knm^2
\end{align}

Recall that in the algorithm $\prog{SimRel}_w(\D)$, the number of
iterations $k$ is at most $n^2$.  Hence, the time complexity amounts
to $\O(m^2n^3)$.  The space complexity is $\O(n^2)$ because of the
representation of $R$.  If the fanout is bounded by a constant $g$, we
have $m\le gn$, and get the complexity $\O(n^5)$.
\end{proof}

\subsubsection{An Improvement}
\label{sec:improvement}
The algorithm $\prog{Ws}(\M,s_1,s_2,R)$ is dominated by the part in which
all breakpoints ($\O(n)$ many) must be computed, and a binary search is
applied to the breakpoints, with $\O(\log n)$ many feasible flow problems.
In this section we discuss how to achieve an improved algorithm if the
network $\N(\gamma)$ can be partitioned into sub-networks.

Let $H$ denote the sub-relation $R\cap [(\post(s_1)\cup\{s_1\})\times
(\post(s_2)\cup\{s_2\})]$, which is the local fragment of the relation
$R$.  Now let $A_1,A_2,\ldots A_\psize$ enumerate the classes of the
equivalence relation $(H\cup H^{-1})^*$ generated by $H$, where
$\psize$ denotes the number of classes. W.\,l.\,o.\,g.\@, we assume in
the following that $A_\psize$ is the equivalence class containing
$s_1$ and $s_2$, \ie, $s_1,s_2\in A_\psize$. The following lemma gives
some properties of the sets $A_i$ provided that $s_1\wsrel_R s_2$:

\begin{lem}\label{dtmc_prepare}
  For $(s_1,s_2)\in R$, assume that there exists a state $s_1'\in
  \post(s_1)$ such that $s_1'\not\in R^{-1}(s_2)$, and $s_2'\in
  \post(s_2)$ such that $s_2'\not\in R(s_1)$.  Let $A_1,\ldots,
  A_\psize$ be the sets constructed for $(s_1,s_2)$ as above. If
  $s_1\wsrel_R s_2$, the following hold:
  \begin{enumerate}[\em(1)]
  \item $\P(s_1,A_i)>0$ and $\P(s_2,A_i)>0$ for all $i<\psize$,
  \item $\gamma_i={K_1}/{K_2}$ for all $i<\psize$ where
    $\gamma_i = {\P(s_1,A_i)}/{\P(s_2,A_i)}$
  \end{enumerate}
\end{lem}
\begin{proof}
  Since $s_1\wsrel_R s_2$, we let $\delta_i, U_i,V_i,\Delta$ (for
  $I=1,2$) as described in Definition~\ref{def:dtmc_weak_simulation}.
  Because of states $s_1'$ and $s_2'$, we have $K_1>0,K_2>0$.  Let
  $\post_i(s_j) = A_i \cap \post(s_j)$ for $i=1,\ldots,\psize$ and
  $j=1,2$.  We prove the first part.  For $i<\psize$, the set $A_i$
  does not contain $s_1$ nor $s_2$, but only (some of) their
  successors, so it is impossible that both $\P(s_1,A_i)=0$ and
  $\P(s_2,A_i)=0$.  W.\,l.\,o.\,g.\@, assume that
  $\P(s_1,A_i)>0$. There exists $t\in \post_i(s_1)$ such that
  $\P(s_1,t)>0$.  Obviously $\delta_1(t)=1$, thus: $0<\P(s_1,t)=(K_1
  \Delta (t,U_2))/\delta_1(t) = K_1 \Delta (t,U_2) $ which implies
  that $\exists u_2\in A_i$ with $\Delta(t,u_2)>0$.  Hence, $
  \P(s_2,u_2) = K_2\Delta(U_1,u_2) \ge K_2\Delta(t,u_2) >0. $ Now we
  prove the second part. It holds that:
\begin{multline*}
  \P(s_1,A_i) = \sum_{a_i\in A_i}\P(s_1,a_i) =  \sum_{a_i\in
    \post_i(s_1)}\P(s_1,a_i) 
  \overeq{(*)}   \sum_{a_i\in \post_i(s_1)} \frac{K_1\Delta(a_i,U_2)}{\delta_1(a_i)}\\              
  \overeq{(!)}  K_1\cdot \sum_{a_i\in \post_i(s_1)}\Delta(a_i,U_2)
  \overeq{(\dag)}  K_1\cdot \sum_{a_i\in \post_i(s_1)}\Delta(a_i,A_i) = K_1\cdot \sum_{a_i\in A_i}\Delta(a_i,A_i)
\end{multline*}
where $(*)$ follows from Condition~\ref{ws:wfsums} of
Definition~\ref{def:dtmc_weak_simulation},
$(!)$ follows from the equation $\delta_1(a_i)=1$ for all $a_i\in
\post_i(a_1)$ with $i<n$,
and $(\dag)$ follows from the fact that if $a\in \post_i(s_1)$, then $\Delta(a,b)=0$ 
for $b\in U_2\backslash \post_i(s_2)$.
In the same way, we get $\P(s_2,A_i)= K_2\cdot \sum_{a_i\in
  A_i} \Delta(A_i,a_i)$.
Therefore, $\gamma_i = {K_1}/{K_2}$ for $1\le i<\psize$. 
\end{proof}

\begin{algorithm}[tbp]
  \caption{Algorithm to check whether $s_1\wsrel_R s_2$ tailored to
    DTMCs.}
  \label{fig:main}
    \Procname{$\prog{WsImproved}(\D,s_1,s_2,R)$}
  \begin{algorithmic}[1]
        \STATE  Construct the partition $A_1,\ldots, A_\psize$
\hfill{(* Assume that $h>1$ *)}\label{algwsd:constructA}
        \FORALL{$i\gets 1,2,\ldots \psize-1$} \label{algwsd:loop.start}
                \IF{$\P(s_1,A_i)=\P(s_2,A_i)=0$}
                        \STATE  \textbf{raise error}
                \ELSIF{$\P(s_1,A_i)=0$ or $\P(s_2,A_i)=0$}
                        \STATE  \textbf{return} $\mathbf{false}$
                \ELSE
                        \STATE  $\gamma_i \gets
    \frac{\P(s_1,A_i)}{\P(s_2,A_i)}$
                \ENDIF
        \ENDFOR
        \IF{$\gamma_i\not=\gamma_j$ for some $i,j<\psize$}
                \STATE  \textbf{return} $\mathbf{false}$ \label{algwsd:gammadiffers}
        \ENDIF
        \STATE \textbf{return} ($\gamma_1$ is valid for $\N(\gamma_1)$) \label{algwsd:gammavalid}
  \end{algorithmic}
\end{algorithm}

For the case $\psize>1$, the above lemma allows to check whether $s_1
\wsrel_R s_2$ efficiently. For this case we replace
lines~\ref{algwsf:computebreakpoints}--\ref{algwsf:returnbiisvalid} of
$\prog{Ws}$ by the sub-algorithm $\prog{WsImproved}$ in
Algorithm~\ref{fig:main}.  The partition $A_1,\ldots,A_\psize$ is
constructed in
line~\ref{algwsd:constructA}. Lines~\ref{algwsd:loop.start}--\ref{algwsd:gammadiffers}
follow directly from Lemma~\ref{dtmc_prepare}: if $\gamma_i\neq\gamma_j$
for some $i,j<h$, we conclude from Lemma~\ref{dtmc_prepare} that
$s_1\not\wsrel_R s_2$. Line~\ref{algwsd:gammavalid} follows from the
following lemma, which is the counterpart of Lemma~\ref{valid_gamma}:

\begin{lem}\label{valid_gamma_dtmc}
  For $(s_1,s_2)\in R$, assume that there exists a state $s_1'\in
  \post(s_1)$ such that $s_1'\not\in R^{-1}(s_2)$, and $s_2'\in \post(s_2)$
  such that $s_2'\not\in R(s_1)$.  Assume that $\psize>1$, and assume
  $\prog{WsImproved}(\D,s_1,s_2,R)$ reaches line~\ref{algwsd:gammavalid}.
  Then, $s_1\wsrel_R s_2$ iff $\gamma_1$ is valid for $\N(\gamma_1)$.
\end{lem}
\begin{proof}
   First, assume that $s_1\wsrel_R s_2$.  According to
  Lemma~\ref{valid_gamma}, there exists a valid $\gamma^*$ for
  $\N(\gamma^*)$. As in the proof of Lemma~\ref{valid_gamma},
  $\gamma^*={K_1}/{K_2}$ is valid for $\N(\gamma^*)$. If $\prog{Ws}$
  reaches line~\ref{algwsd:gammavalid}, by Lemma~\ref{dtmc_prepare}, we
  have $\gamma_1={K_1}/{K_2}$, hence, $\gamma_1$ is valid for
  $\N(\gamma_1)$. The other direction follows directly from
  Lemma~\ref{valid_gamma}.
\end{proof}

\begin{exa}\label{dtmc_example}
  Consider the DTMC in Figure~\ref{fig:dtmc_partition} together with the
  relation $R= \{(s_1,s_2)$, $(s_1,v_2), (v_1,s_2), (u_1,u_2), (o_1,o_2),
  (o_1,v_2), (v_1,o_3), (v_1,v_2), (o_2,o_1)\}$. We obtain the relation $H
  = R\backslash \{(o_2,o_1)\}$. We get two partitions $A_1=\{u_1,u_2\}$ and
  $A_2=\{s_1, s_2, v_1, v_2, o_1, o_2, o_3\}$. In this case we have
  $\psize=2$. Recall that $MU_1=\{u_1, o_1\}$, $MU_2=\{u_2, o_2,o_3\}$,
  $PV_1=\{v_1\}$, $PV_2=\{v_2\}$.  We have $\P(s_1,A_1)=\frac{1}{4}$ and
  $\P(s_2,A_1)=\frac{1}{8}$.  Hence, $\gamma_1
  ={\P(s_1,A_1)}/{\P(s_2,A_1)} = 2$. As we have shown in
  Example~\ref{exa:feasible}, $2$ is valid for the network $\N(2)$. Hence,
  $s_1\wsrel_R s_2$.
\end{exa}

Assume that $(s_1,s_2)\in R_1$ such that $h>1$ in the first iteration of
$\prog{SimRel}_w$. We consider the set $A_1$ and let
$\gamma_1={\P(s_1,A_1)}/{\P(s_2,A_1)}$. If $A_1$ is not split in the
next iteration, $\gamma_1$ would not change, and hence, we can reuse the
network constructed in the last iteration.  Assume that in the next
iteration $A_1$ is split into two sets $A_1^a$ and $A_1^b$.  There are two
possibilities:
\begin{enumerate}[$\bullet$]
\item either ${\P(s_1,A_1^a)}/{\P(s_2,A_1^b)}
  = {\P(s_1,A_1^a)}/{\P(s_2,A_1^b)}$. This
  implies that both of them are equal to $\gamma_1$.  If all $A_i$ are
  split like $A_1$, we just check whether $\gamma_1$ is valid for
  $\N(\gamma_1)$.
\item or
  ${\P(s_1,A_1^a)}/{\P(s_2,A_1^b)} \not=
  {\P(s_1,A_1^a)}/{\P(s_2,A_1^b)}$. This case
  is simple, we conclude $s_1\not\wsrel_R s_2$ because of Lemma~\ref{dtmc_prepare}.
\end{enumerate}
This indicates that once in the first iteration $\gamma_1$ is determined
for $(s_1,s_2)$, either it does not change throughout the iterations, or we
conclude that $s_1\not\wsrel_R s_2$ directly.  The above analysis
can be generalised to the case in which $A_1$ is split into more than two
sets. As the network $\N(\gamma_1)$ is fixed, we can apply an algorithm
similar to $\prog{Smf}$, which solves the maximum flow problems
during all subsequent iterations using only one parametric maximum flow, as for
strong simulation.  

The above analysis implies that if $h>1$ for all $(s_1,s_2)$ in the initial
$R_1$, we could even establish the time bound $\O(m^2n)$, the same as
for strong simulation.  Since in the worst case it could be the case that $h=1$
for all $(s_1,s_2)\in R$, the algorithm $\prog{WsImproved}$ does not
improve the worst case complexity.  

Since the case that the network cannot be partitioned ($\psize=1$) is the
one that requires most of our attention, we suggest a heuristic approach
that can reduce the number of occurrences of this case.  We may choose to
run some iterations incompletely (as long as the last iteration is run
completely).  If iteration $i$ is incomplete, we first check for each pair
$(s_1,s_2)\in R_i$ whether the corresponding $\psize_i$ is greater than
$1$.  If not, we skip the test and just add $(s_1,s_2)$ to $R_{i+1}$.  The
intuition is that in the next complete iteration $i' > i$, for each such  pair
$(s_1,s_2)$ we hope to get $h_{i'} > 1$ because some other elements of
$R_i$ have been eliminated from it.  We only perform the expensive
computation if an incomplete iteration does no longer refine the relation.

\subsection{An Algorithm for CTMCs} 
\label{sec:weak_ctmc}
Let $\ctmc$ be a CTMC.  We now discuss how to handle CTMCs. Recall that in
Definition~\ref{def:ctmc_weak_simulation}, we have the rate
condition~$\ref{ws:reachability}'$: $K_1\R(s_1,S)\le K_2\R(s_2,S)$. To
determine $\wsrel_\M$, we simplify the algorithm for DTMCs. If $K_1 > 0$
and $K_2=0$, $s_1\not\wsrel_R s_2$ because of the rate condition. Hence, we
do not need to check the reachability condition, and
lines~\ref{algws:testposts2simss1}--\ref{algws:testreachcond.end} of the
algorithm $\prog{Ws}(\C,s_1,s_2,R)$ can be skipped. For states $s_1,s_2$
and relation $R$, we use $\N(\gamma)$ to denote the network defined in the
embedded DTMC $emb(\M)$.  To check the additional rate condition we use the
following lemma:

\begin{lem}\label{valid_minimal}
Let $s_1\ R\ s_2$. Assume that there exists $s_1'\in post(s_1)$ such that
$s_1'\not\in R^{-1}(s_2)$. Then, $s_1\wsrel_R s_2$ in $\C$ iff there exists
$\gamma\le{\R(s_2,S)}/{\R(s_1,S)}$ such that $\gamma$ is valid for
$\N(\gamma)$.
\end{lem}
\begin{proof}
Assume first $s_1\wsrel_R s_2$ in $\C$. Let
$\delta_i,U_i,V_i,K_i,\Delta$ (for $i=1,2$) as described in
Definition~\ref{def:ctmc_weak_simulation}. Obviously, $s_1'$ must be in
$U_1$, implying that $K_1>0$. Because of the rate condition it holds that
$K_1\R(s_1,S)\le K_2\R(s_2,S)$, which implies that $K_2>0$. It is
sufficient to show that $\gamma:={K_1}/{K_2}$ is valid for
$\N(\gamma)$. Exactly as in the proof of Lemma~\ref{valid_gamma} (the
\emph{only if} direction), we can construct a valid flow $f$ for
$\N(\gamma)$. Thus, $\gamma$ is valid for $\N(\gamma)$ and
$\gamma\le{\R(s_2,S)}/{\R(s_1,S)}$. 

Now we show the other direction. By assumption, $\gamma$ is valid for
$\N(\gamma)$. We may assume that there exists a valid flow function $f$ for
$\N(\gamma)$. We define $\delta_i,V_i,U_i,K_i,\Delta$ (for $i=1,2$) as in
the proof (the
\emph{if} direction) of Lemma~\ref{valid_gamma}. Recall that $s_1'$ must be
in $U_1$, implying that $f(\source,s_1')>0$. Thus, there must be a node
$\overline{s}$ in $\N(\gamma)$ with $f(\overline{s},\sink)>0$, which
implies that $s\in U_2$. Thus we have $K_2>0$. Using the proof (the
\emph{if} direction) of Lemma~\ref{valid_gamma}, it holds that
$s_1\wsrel_R s_2$ in $emb(\C)$, moreover, it holds that
$\gamma={K_1}/{K_2}$. By assumption it holds that
$\gamma\le\R(s_2,S)/\R(s_1,S)$ which is exactly the rate condition.
\end{proof}

To check the rate condition for the case $\psize>1$, we replace line~\ref{algwsd:gammavalid}
of the algorithm $\prog{WsImproved}$ by:
\[\textbf{return } \left(\gamma_1\le \gamma^* \wedge
        \gamma_1 \text{ is valid for } \N(\gamma_1) \right)
\]
where $\gamma^*=\R(s_2,S)/\R(s_1,S)$ can be computed directly.  In
case $\psize=1$, we replace line~\ref{algwsf:returnbiisvalid} of
$\prog{Ws}$ by:
\[\textbf{return } \left(\exists i \in \{ 1, \ldots, j \}.\ b_i\le
  \gamma^* \wedge b_i \text{ is valid for } \N(b_i) \right)
\]
to check the rate condition. Or, equivalently, we can check whether the
minimal valid breakpoint $\gamma_m$ is smaller than or equal to
$\gamma^*$. The binary search algorithm introduced for DTMCs can also be
modified slightly to find the minimal valid breakpoint.  The idea is that,
if we find a valid breakpoint, we first save it, and then continue the
binary search on the left side.  If another breakpoint is valid, we save
the smaller one.  As the check for the reachability condition disappears
for CTMCs, we get  better bound for sparse CTMCs:
\begin{thm}
If the fanout $g$ of $\C$ is bounded by a constant,
the time complexity for CTMC is $\O(n^4)$.
\end{thm}
\begin{proof}
  In the proof of Theorem~\ref{thm:fpsws_complexity} we have shown that
  $\prog{Ws}$ has complexity $\O(\size{V_1} \size{V_2}^2)$.  As we do not need to
  check the reachability condition, the overall complexity of the
  algorithm $\prog{SimRel}_w(\C)$ (see
  Inequality~\ref{eq:complexity_fps}) is $\sum_{s_1\in S} \sum_{s_2\in
    S} \sum_{i=1}^l \left(\size{\post(s_1)} + \size{V_1}\size{V_2}^2
  \right)$ which is bounded by $2kgm^2$.  Since $k$ is bounded by
  $n^2$, the time complexity is bounded by $4gm^2n^2$. If $g$ is a
  constant, we have $m\le gn$, hence, the time complexity is
  $4g^3n^4\in \O(n^4)$.
\end{proof}

\section{Conclusion and Future Work}
\label{sec:concl}

\noindent In this paper we have proposed efficient algorithms  deciding
simulation on Markov models with complexity $\O(m^2n)$.  For sparse
models where the fanout is bounded by a constant, we achieve
for strong simulation the complexity $\O(n^2)$
and for weak simulation $\O(n^4)$ on CTMCs and $\O(n^5)$ for DTMCs.  We
extended the algorithms for computing simulation preorders to handle
PAs with the complexity $\O(m^2n)$ for strong simulation. For strong
probabilistic simulation, we have shown that the preorder can be
determined by solving LP problems. We also considered
their continuous-time analogon,
CPAs, and arrived at an algorithm with same complexities as for PAs. 

\bibliographystyle{abbrv} 
\bibliography{bib}\vskip-40 pt
\end{document}